\newif\ifsoda
\ifsoda
	\documentclass[twoside,leqno,twocolumn]{article}
	\usepackage[letterpaper]{geometry}	
	\usepackage{ltexpprt}
\else
	\documentclass[11pt,a4paper]{article}
	\usepackage{fullpage}
\fi
\ifsoda
	\usepackage{amsmath,amssymb}
\else
	\usepackage{amsthm,amsmath,amssymb}
\fi
\usepackage{algorithm}
\usepackage{algpseudocode}
\algtext*{EndWhile}
\algtext*{EndIf}
\algtext*{EndFor}
\usepackage{graphicx}
\usepackage{enumitem}
\usepackage{todonotes}
\ifsoda
\else
	\setlist{topsep=3pt, itemsep=0pt}
\fi

\newcommand{\hyphen}{\textnormal{-}}

\usepackage{xspace}
\usepackage{setspace}

\usepackage{xcolor}

\usepackage{thmtools}
\usepackage{thm-restate}

\ifsoda
	\newtheorem{observation}{Observation}[section]
	\newtheorem{definition}{Definition}[section]
\else
	\declaretheorem[numberwithin=section, name=Theorem]{theorem}
	\declaretheorem[sibling=theorem, name=Lemma]{lemma}
	
	\declaretheorem[sibling=theorem, name=Definition]{definition}
	\declaretheorem[sibling=theorem, name=Corollary]{corollary}
	\declaretheorem[sibling=theorem, name=Observation]{observation}
\fi

\newcommand{\set}[1]{\left\{#1\right\}}

\newcommand{\floor}[1]{\left\lfloor#1\right\rfloor}

\DeclareMathOperator{\E}{\mathbb E}
\DeclareMathOperator{\union}{\bigcup}

\renewcommand{\hat}{\widehat}
\renewcommand{\tilde}{\widetilde}

\newcommand{\Z}{\mathbb{Z}}

\newcommand{\calC}{{\mathcal C}}

\newcommand{\calI}{{\mathcal I}}

\newcommand{\calP}{{\mathcal P}}

\newcommand{\calS}{{\mathcal S}}

\newcommand{\calX}{{\mathcal X}}

\newcommand{\chain}{\mathrm{chain}}

\newcommand{\calD}{{\mathcal{D}}}

\newcommand{\eps}{\epsilon}

\renewcommand{\bot}{\mathrm{bot}}
\renewcommand{\top}{\mathrm{top}}
\renewcommand{\mid}{\mathrm{mid}}

\newcommand{\mac}{{\mathrm{mac}}}
\newcommand{\tim}{{\mathrm{time}}}

\newcommand{\SA}{\mathrm{SA}}

\newcommand{\mysplit}{{\mathsf{SPLIT}}}

\graphicspath{{figures/}}

\newcommand{\jbot}{J_{\mathrm{bot}}}
\newcommand{\jtop}{J_{\mathrm{top}}}
\newcommand{\jmid}{J_{\mathrm{mid}}}
\newcommand{\jspl}{J_{\mathrm{special}}}

\newcommand{\taskdis}{A_{\mathrm{discarded}}}
\newcommand{\vecx}{x}
\newcommand{\sherali[1]}{\textit{SA}(\mathcal{P}(T),#1)}
\newcommand{\drop}{\mathrm{DISCARDED}}
\newcommand{\recursive}{\mathsf{PARTIAL\hyphen SCHEDULE}}
\newcommand{\commrecursive}{\mathsf{PARTIAL\hyphen SCHEDULE \hyphen COMM}}

\newcommand{\capa}{\mathsf{cap}}

\def\DEBUG{true}

\ifdefined\DEBUG
	 \newcommand{\sh}[1]{\textcolor{red}{#1}}
	 \newcommand{\Jnote}[1]{\textcolor{blue}{#1}}
	 
	 \newcommand{\rem}[2]{%
	    \marginpar{{\scriptsize{\scriptsize\vspace{-\baselineskip}\singlespacing \textcolor{#1}{#2}}\par}}
	 } 
	\newcommand{\shr}[1]{\rem{red}{#1}}
	
\else
	\newcommand{\sh}[1]{#1}
	 \newcommand{\Jnote}[1]{}
	\newcommand{\shr}[1]{}
\fi

\begin{document}
	
\title{
Hierarchy-Based Algorithms  for  Minimizing Makespan under Precedence and Communication Constraints}
\author{ 
		Janardhan Kulkarni \thanks{Microsoft Research, {\tt jakul@microsoft.com}} \and
		Shi Li \thanks{University at Buffalo, {\tt shil@buffalo.edu}} \and
		Jakub Tarnawski \thanks{Microsoft Research, {\tt jatarnaw@microsoft.com}. Part of this work was done while the author was at École Polytechnique Fédérale de Lausanne (EPFL).} \and
		Minwei Ye \thanks{University at Buffalo, {\tt minweiye@buffalo.edu}}
	   }

\date{November 1, 2019} 
\maketitle

\ifsoda
\else
	\thispagestyle{empty}
\fi

\begin{abstract}
We consider the classic problem of scheduling jobs with precedence constraints on a set of identical machines to minimize the makespan objective function. 
Understanding the exact approximability of the problem when the number of machines is a constant is a well-known question in scheduling theory.  
Indeed, an outstanding open problem from the classic book of Garey and Johnson~\cite{Garey} asks whether this problem is NP-hard even in the case of 3 machines and unit-length jobs.  
In a recent breakthrough, Levey and Rothvoss~\cite{LR16} gave a $(1+\epsilon)$-approximation algorithm, which runs in nearly quasi-polynomial time, for the case when job have {\em unit lengths}. 
However, a substantially more difficult case where jobs have arbitrary processing lengths has remained open.
We make progress on this more general problem. 
We  show that there exists a $(1+\epsilon)$-approximation algorithm (with similar running time as that of \cite{LR16}) for the \emph{non-migratory setting}: when every job has to be scheduled entirely on {\em a single machine}, but within a machine the job need not be scheduled during consecutive time steps.
Further, we also show that our algorithmic framework generalizes to another classic scenario where, along with the precedence constraints, the jobs also have {\em communication delay} constraints.
Both of these fundamental problems are highly relevant to the practice of datacenter scheduling.

\end{abstract}	


\section{Introduction}
\label{sec:intro}

A classic problem in scheduling theory is as follows: We are given a set $J$ of $n$ jobs, where each job $j \in J$ has a processing length $p_j$. 
The jobs have {\em precedence constraints}, which are given by a partial order ``$\prec$". A constraint $j \prec j'$
requires that job $j'$  can only start after job $j$ is completed.
The jobs need to be scheduled on a set of $m$ {\em identical} machines. 
The goal is to schedule the jobs while respecting the precedence constraints so as to optimize a certain objective function.  
The most popular objective function is the {\em makespan} of a schedule, which is the focus of this paper. 
The makespan of a schedule is defined as the completion time of the last job.
In the classic three-field notation\footnote{
	The first field describes the machine environment; in this paper we only consider the identical machine setting denoted by $P$. We use $Pm$ to denote $m = O(1)$ machines, $P\infty$ when the number of machines is unbounded. The second field describes constraints that a schedule must satisfy (e.g.~$\text{prec}$ denotes precedence constraints), as well as assumptions on the input (e.g.~$p_j = 1$ denotes unit-length jobs). The third field denotes the objective function, which in this paper is always the makespan ($C_{\max}$).
}~\cite{GLLR79}, the problem is denoted by $P|\text{prec}|C_{\max}$.
Since the seminal result of Graham \cite{Gra69}, the problem has been studied quite extensively in the literature \cite{GLLR79, GANGAL20081139, CS97, CK04, MQS98, Sku16, HSW98, LR78, Sve10, BK09, LR16, MQS98, HSS97, CPS96, QS02}. 
Despite this, large gaps remain in our understanding of the problem.
The influential survey of Schuurman and Woeginger \cite{SW99a} and a more  recent one by Bansal \cite{Bansalmapsp} list determining the exact approximability of this setting as one of the top ten open problems in scheduling theory (Open Problem 1).

Already in 1966, Graham \cite{Gra69} showed that any greedy non-idling schedule is a $(2-1/m)$-approximation to the problem of minimizing makespan with precedence constraints on identical machines. 
Fifty years later, assuming a variant of the Unique Games Conjecture (UGC) introduced by Bansal and Khot \cite{BK09},  Svensson \cite{Sve10} showed that an approximation factor of $(2-\epsilon)$ is indeed hard for this problem, even for the unit-length case ($P|\text{prec}, p_j = 1|C_{\max}$). For the unit job length case, the current best approximation factor is $2-7/(3m+1)$, when $m\geq4$, due to Gangal and Ranade~\cite{GANGAL20081139}.

\medskip
In most applications, however, the number of machines is typically much smaller than the number of jobs.  
Thus, a natural question that has attracted much attention is: what is the exact complexity of the problem if $m$ is a constant (setting denoted by $Pm$)?
One of the longest-standing open problems posed in the classic book by Garey and Johnson \cite{Garey} is  whether $Pm|\text{prec}, p_j = 1|C_{\max}$ is NP-hard for any $m \geq 3$.
On the other hand, the more general problem where jobs have different processing lengths is strongly NP-hard even if all jobs have lengths 1 or 2 and there are only two machines if no preemption is allowed  \cite{Ull75}.
 On the positive side, in a recent breakthrough Levey and Rothvoss \cite{LR16} gave a $(1+\epsilon)$-approximation algorithm for the problem $Pm|\text{prec}, p_j = 1|C_{\max}$ with running time 
 $n^{(\log n)^{O_{m,\epsilon}( \log\log n)}}$, which is nearly quasi-polynomial time.
They obtain this result by an elegant, though quite technically involved, rounding of a fractional solution obtained from the Sherali-Adams lift of an LP relaxation of the problem to $(\log n)^{O_{m,\epsilon}( \log\log n)}$ levels, which explains the running time of the algorithm. Later, Garg \cite{Shashwat} made the result strictly quasi-polynomial time. However, the more general case, where jobs have arbitrary lengths, has remained open.

One of the tantalizing  questions posed in~\cite{LR16} is whether LP hierarchies can also give a $(1+\epsilon)$-approximation for  the significantly harder case of jobs with arbitrary processing times. 
In this paper, we show that the answer to this question may depend on the type of schedule we are looking for.
When jobs have arbitrary processing lengths, the optimal schedule can be of three types:

\begin{enumerate}
	\item {\em Fully preemptive} ($Pm|\text{prec}, \text{pmtn}, \text{migration}|C_{\max}$): In this case, a job $j$ can be scheduled on multiple machines (respecting the precedence constraints). However, at any time step $t$ only one machine can be processing the job $j$. 
	\item {\em Preemptive but non-migratory} ($Pm|\text{prec}, \text{pmtn}|C_{\max}$): In this case, a job $j$ must be scheduled completely on a single machine. However, within a machine the job can be preempted and need not be processed during consecutive time steps.
	\item {\em Non-preemptive} ($Pm|\text{prec}|C_{\max}$): Here we require that a job $j$ must be scheduled on a single machine during $p_j$ consecutive time steps.
\end{enumerate}

It is interesting to note that the optimal solutions for the above three cases can be quite different.\footnote{We are not aware of any other setting for makespan minimization where there is a gap between the optimal schedules 2 and 3. However, this situation is quite common for flow-time objective functions, where \emph{preemptive but non-migratory} is the standard assumption.}
See \ifsoda\else Appendix~\ref{sec:3schedules} \fi for examples where the above three schedules are a constant factor away from each other. Hence, there cannot be black-box reductions among these problems if our goal is to achieve $(1+\epsilon)$-approximation.
Observe that the fully preemptive  case ($Pm|\text{prec}, \text{pmtn}, \text{migration}|C_{\max}$) is equivalent to the unit-length case, assuming all sizes are polynomially bounded: one can break a job $j$ of length $p_j$ into a chain of $p_j$ unit-length jobs (see \ifsoda\else Appendix~\ref{sec:N-big} \fi for how to handle the case when this assumption does not hold).
Thus, the $(1+\epsilon)$-approximation algorithm due to \cite{LR16} readily extends to this case.
 
\medskip
In this paper we first consider the preemptive but non-migratory setting ($Pm|\text{prec}, \text{pmtn}|C_{\max}$).   
We prove  that there is a Sherali-Adams hierarchy based algorithm that gives a $(1+\epsilon)$-approximation,  
which generalizes the framework of  Levey and Rothvoss \cite{LR16} for the unit-length case to the general job length case, 
thus positively answering the question posed by the authors in our setting.

The first main result of this paper is the following.
\begin{theorem}
	\label{thm:secondresult}
	For any $\epsilon > 0$, there is a $(1+\epsilon)$-approximation algorithm for the problem $Pm|\text{prec}, \text{pmtn}|C_{\max}$ that runs in time $ n^{(\log n)^{O((m^2/\epsilon^2)\log \log n))}}$. 
\end{theorem}

We give a detailed explanation of our algorithm and techniques in Section~\ref{sec:modifications_to_LR16},
preceded by an introduction to Levey and Rothvoss' framework in Section~\ref{sec:LR}.
In Section~\ref{sec:makespan}, we give a complete proof of Theorem~\ref{thm:secondresult}.

\medskip
Next, we turn our attention to the non-preemptive case ($Pm|\text{prec}|C_{\max}$). 
Here we give some evidence to show that the algorithmic framework based on the Sherali-Adams hierarchy may not be sufficient to get a polynomial-time $(1+\epsilon)$-approximation, even when there are only 2 machines.
Our reasoning behind this claim is as follows. 
The rounding algorithm 	used in the proof of Theorem \ref{thm:secondresult} reduces our problem to the case of a deadline scheduling problem with the objective of maximizing throughput; that is, maximizing the number of job units completed, where a job may be processed partially.
For this problem, somewhat surprisingly, we prove that any $o(\log n)$-level Sherali-Adams lift of the basic LP has at least a constant integrality gap.
Since the throughput problem is equivalent to a special case of non-preemptive scheduling on 2 machines ($P2|\text{prec}|C_{\max}$), we believe that $P2|\text{prec}|C_{\max}$ also has a constant integrality gap with an $o(\log n)$-level Sherali-Adams lift. 
This is in contrast to the unit-length case, where a $2$-level Sherali-Adams hierarchy can solve the problem $P2|\text{prec},p_j=1|C_{\max}$ exactly \cite{Rot13}, and many experts believe that $O(m)$ levels should give either an exact solution or a $(1+\epsilon)$-approximation. 
We give more details about the hard instances for $P2|\text{prec}|C_{\max}$
in \ifsoda\else Section~\ref{sec:integrality}. \fi

\medskip

The second main result of the paper concerns another classic problem: scheduling jobs with precedence and \emph{communication delay constraints}. 
This model was introduced by Papadimitriou and Yannakakis~\cite{Papadimitriou} and Veltman et al.~\cite{veltman1990multiprocessor} to capture the fact that in  multiprocessor systems, when jobs have dependencies $j \prec j'$, it takes time to transfer the output of a job~$j$
to another machine, where it will become the input of job~$j'$.
MapReduce systems and multi-core processors are modern examples of such systems.  
The formal setting of this problem is similar to that of makespan minimization with precedence constraints.  
However, if two jobs with $j \prec j'$ are executed on {\em different machines},
the second job~$j'$ is allowed to start only $c_{j,j'} \geq 0$ time units after the completion time of job $j$.
Here $c_{j, j'}$ is the {\em communication delay} between $j$ and $j'$. 
On the other hand, if $j$ and $j'$ are executed on the same machine, then $j'$ can start right after $j$ completes.

This model has been studied quite extensively in the literature, and yet our understanding of it is very limited.
The surveys by Schuurman and Woeginger \cite{SW99a} and Bansal \cite{Bansalmapsp} list the approximability status of problems in this model as a top-ten open problem in scheduling theory.     
Most known results are for the special case where all the jobs have unit lengths and the communication delays $c_{j, j'}$ are also identically 1. 
In the classic notation, this special case is denoted by $P|\text{prec}, p_j = 1, c = 1|C_{\max}$. 
For this problem, Hanen  and Munier \cite{HM01} gave a polynomial-time $7/3$-approximation algorithm, while on the hardness side Hoogeveen et al.~\cite{HLV94} showed the problem does not admit a better than $5/4$-approximation algorithm unless $P = NP$.  
Another important case that has gained a lot of interest is when the number of machines is unbounded; the setting is non-trivial in presence of communication delays.  
The problem, denoted as $P\infty|\text{prec}, p_j = 1, c = 1|C_{\max}$, admits a $4/3$-approximation due to Munier and Konig \cite{MK97}, and it is NP-hard to do better than $7/6$~\cite{HLV94}. Papadimitriou and Yannakakis claim that there is a lower bound of 2 for the $P\infty|\text{prec}, p_j = 1, c|C_{\max}$ problem (where the communication delays are uniform but arbitrary), yet there has been no proof of this claim as far as we know; see Open Problem 3 in the survey by Schuurman and Woeginger~\cite{SW99a} for more details.

\medskip
As communication delay constraints with $c_{j,j'} \geq 0$ strictly generalize scheduling with precedence constraints\footnote{This statement is not true if $c_{j,j'} >0$ instead of $c_{j,j'} \geq 0$. Thus, Svensson's hardness result does not immediately apply to the $c=1$ case, for example.}, Svensson's hardness result \cite{Sve10} for $P|\text{prec}, p_j = 1|C_{\max}$, assuming UGC, also applies to our problem with the communication delay even for unit job lengths; that is, ($P|\text{prec}, p_j = 1, c_{j,j'}|C_{\max}$). 
Hence, we initiate the study of this problem when the number of machines is a constant.
To the best of our knowledge, we are the first to consider the communication delay problem in the $m = O(1)$ setting.
Our second main result is a generalization of our first result Theorem \ref{thm:secondresult}  to this setting.
\begin{theorem}
	\label{thm:delayresult}
	For any $\epsilon > 0$, there is a $(1+\epsilon)$-approximation algorithm for the problem $Pm|\text{prec}, \text{pmtn}, c_{j,j'}|C_{\max}$ that runs in time $ n^{(\log n)^{O((m^2/\epsilon^2)\log \log n))}}$ if $\max_{j \prec j'} c_{j, j'}  = O(1)$.
\end{theorem}

Observe also that $\max_{j \prec j'} c_{j, j'}  = O(1)$ is a weaker assumption than $c =1$, which is the only setting where previously known results hold. 
We obtain the above result by extending the Sherali-Adams hierarchy framework introduced for the problem without communication delay constraints, i.e., the setting of 	Theorem \ref{thm:secondresult}.
This shows the versatility of the LP-hierarchy based approaches to the study of these problems. 
We anticipate that such approaches should help in resolving open problems in this model,  which is one of the poorest-understood in scheduling.
 We introduce our techniques behind this result in Section~\ref{sec:ideas_comm_delays} and give the complete analysis in \ifsoda\else Section~\ref{sec:communication}. \fi

\subsection{Outline}
\ifsoda
\else
	The proofs of Theorems \ref{thm:secondresult} and \ref{thm:delayresult} are quite involved and build on \cite{LR16}, and hence need a fair amount of background. Therefore, our
\fi
 paper is organized as follows. 
In Section~\ref{sec:ourtech}, we give a detailed but informal description of our two algorithms for Theorems \ref{thm:secondresult} and \ref{thm:delayresult}, and discuss the new ideas of this paper. 
\ifsoda\else
	We formally define the properties of Sherali-Adams solutions that we use in our algorithms in Section~\ref{sec:SA}.  
\fi
In Section~\ref{sec:makespan}, we first give a formal description of our algorithm for the first problem along with all the necessary lemmas, and show how these lemmas come together for the proof of Theorem \ref{thm:secondresult}; then we give complete proofs for these lemmas. 
\ifsoda
\else
	In Section~\ref{sec:communication}, we show how to extend the algorithm to the setting with communication delays, thus proving Theorem~\ref{thm:delayresult}.  Finally, we give our integrality gap result in Section~\ref{sec:integrality}.
\fi

\section{High-Level Description of Our Algorithms and Techniques}
\label{sec:ourtech}
Our proofs for Theorems \ref{thm:secondresult} and \ref{thm:delayresult} are obtained by rounding LP-hierarchy solutions of natural LPs for the problems. 
First we focus on the makespan minimization problem with precedence constraints for general job lengths when no migrations are allowed.
Later we explain the new ideas needed to extend this result to the communication delay setting.
We often refer to the former setting as the ``no-delay'' setting, and to the latter as the ``delay setting''.

As our work generalizes the framework introduced by Levey and Rothvoss \cite{LR16}, we begin by describing their algorithm for minimizing makespan when jobs have unit lengths.  Beforehand, it will be worthwhile to give an intuitive explanation of LP hierarchies and the ``conditioning" operation of LP-hierarchy solutions used in the design of our algorithms.

\subsection{Intuitions behind Sherali-Adams Hierarchy and ``Conditioning'' Operation} 
\label{sec:SA-intuition}
We give a brief explanation of how we use the Sherali-Adams hierarchy in our algorithms; see \ifsoda\else Section~\ref{sec:SA} \fi for more details. Let us start with an ideal situation. Assume we have a set $\calX \subseteq \{0, 1\}^n$ corresponding to the set of valid integral solutions for some instance of a problem.  Further suppose that we are given a vector $x \in \text{convex\hyphen hull}(\calX)$, the convex hull of $\calX$. 
That is, there is an implicit distribution $\pi$ over $\calX$ such that $x = \E_{\tilde x \sim \pi} \tilde x$.  Then, we hope there is an oracle that, for an index $i \in [n]$ with $x_i > 0$, can return the vector $x' = \E_{\tilde x \sim \pi: \tilde x_i = 1} \tilde x$; that is, the solution corresponding to the distribution $\tilde x \sim \pi$ conditioned on the event $\tilde x_i = 1$. 
This is called the ``conditioning'' or ``inducing" operation. We use the word conditioning, as it is same as the conditioning operation known from probability theory.

In an intuitive sense, the Sherali-Adams hierarchy (and other LP/SDP hierarchies) provides a weaker form of the oracle that can support the conditioning operations.  Assume there is a polytope $\calP \supseteq \text{convex\hyphen hull}(\calX)$ that corresponds to the feasible solutions for some LP relaxation for the instance. Then the SA-hierarchy can be applied to $\calP$, giving the oracle with the following restrictions. First, every vector $x$ given to the oracle is associated with a \emph{level} $\ell \in \Z_{>0}$. If we give a level-$\ell$ vector $x$ to the oracle, the vector $x'$ returned by the oracle will only have level $\ell-1$. Second, the vector $x$ given to the oracle does not need to be in $\text{convex\hyphen hull}(\calX)$. Indeed, for each level $\ell$, the SA-hierarchy gives a polytope $\SA(\calP, \ell) \subseteq [0, 1]^n$ such that
\begin{align*}
	\calP = \SA(\calP, 1) \supseteq \SA(\calP, 2) \supseteq \cdots \supseteq \text{convex\hyphen hull}(\calX) \,.
\end{align*}
The oracle only needs a level-$\ell$ vector $x$ to be in $\SA(\calP, \ell)$ to perform the conditioning operation; the returned vector $x'$ will be in $\SA(\calP, \ell-1)$.  So for this reason, the vectors in $\SA(\calP,\ell)$ on which the conditioning operation is performed are sometimes called ``pseudo-distributions''. Finally, to construct a vector $x \in \calP_\ell$, in general we need a running time of $n^{\Theta(\ell)}$. This means that $\ell$ needs to be small, which places a limit on the number $\ell - 1$ of conditioning operations that we can apply to a vector $x$ ``sequentially''.

Certain key properties will be satisfied by the (pseudo-)conditioning operation, as in the ideal case.  For example, if we condition $x$ on the event $\tilde x_i = 1$, then the returned vector $x'$ will have $x'_i = 1$.  Also, conditioning can only shrink the support of vectors: if $x'$ is obtained from $x$ by conditioning, then $x_{i'} = 0$ implies $x'_{i'} = 0$ for every $i'$.  \ifsoda\else Section~\ref{sec:SA} \fi contains formal statements. See also the following references \cite{Laurent, sherali, Lasserre, Lovasz, Rothvoss13}; in particular, \cite{Rothvoss13} is an excellent introduction to the use of hierarchies in approximation algorithms.

	\subsection{Overview of Levey-Rothvoss Algorithm}
	\label{sec:LR}
	
	At the heart of the analysis of the Levey-Rothvoss algorithm \cite{LR16} is the following simple observation:
	if the maximum chain length of jobs in the set $J$ is at most $\epsilon T$, where $T$ is a guess of the optimal makespan, then Graham's algorithm already gives a $(1+\epsilon)$-approximation. 
	At a high level, the algorithm in \cite{LR16} uses the above observation in the following way: 
	it partitions the input instance $J$ into three sets $\jtop, \jmid$ and $\jbot$. 
	It guarantees that the size of $\jmid$ is small, so for the moment we can ignore jobs in $\jmid$.
	The partitioning is done, guided by an LP-hierarchy solution, satisfying the following three properties: 
	1) The maximum chain length among $\jtop$ is small. 
	2) The precedence constraints between the jobs in the sets $\jtop$ and $\jbot$ are {\em loose}, and can be easily satisfied. 
	3) There is a partition of the entire time horizon into small-length intervals such that each job in $\jbot$ has fractional support completely contained in one interval in the partition.
	Then one can recursively schedule the jobs in $\jbot$ by considering sub-instances defined by the small intervals, and such a schedule can be easily extended to include $\jtop$ due to Properties 1) and 2). 

	Now we give more details about the algorithm. The Levey-Rothvoss algorithm first makes a guess on the makespan $T$ of the optimum schedule, which we can assume is a power of two. The time horizon $[T]$ is partitioned into a binary laminar family $\calI$ of intervals, where the root interval is $[T]$ and each subsequent level is constructed by dividing each interval of the previous level into two equal-sized intervals. Each leaf interval contains a single slot. To obtain a schedule of $J$ with makespan at most $(1+\epsilon)T$, it suffices to obtain a schedule with makespan $T$, but with up to $\epsilon T$ jobs discarded. This new goal is indeed more convenient for the description and analysis of the algorithm. At the end, the discarded jobs are re-inserted into the final schedule, each on a time slot of its own.

	We start by solving a standard LP for the problem lifted to $r$ rounds of the Sherali-Adams hierarchy, for some $r=(\log n)^{O_{m,\epsilon}( \log\log n)}$, to obtain a solution $x$. 
	It takes $n^{O(r)}$ time to solve a $r$-round Sherali-Adams lift of an LP with $\text{poly}(n)$ variables; this is the reason for the running time of this algorithm.
	Based on the solution $x$, each job $j$ is assigned to the smallest interval $I$ in the binary decomposition $\calI$ that fully contains the fractional schedule of the job. 
	We then say $j$ is \emph{owned} by $I$.

	The recursive algorithm begins by considering the jobs assigned to the first $k^2$ levels of the binary decomposition $\calI$, for $k = \Theta(\log \log n)$.  
	In the first step, we condition on some variables of the LP solution to reduce the maximum chain length among these jobs.  
	The number of conditioning operations can be bounded due to the following reason: If there is a long chain of jobs in the set, we can condition on some variable so that the support of many jobs in the chain will shrink a lot; thus a few conditioning operations will be sufficient.
	In the second step, we partition the intervals in $\calI$ into top, middle and bottom intervals. Let $\ell^* \in [k+1, k^2]$ be selected satisfying some property. The top intervals are those from level $0$ to $\ell^* - k - 1$ in the laminar tree (the root has level $0$), the middle intervals are those from level $\ell^* - k$ to $\ell^*-1$ and the bottom intervals are those at level $\ell^*$ and below.  
	Then, we define $\jtop$, $\jmid$ and $\jbot$ as the set of jobs owned by top, middle, and bottom intervals respectively.  In the third step, the rounding algorithm recursively and separately solves each of the bottom instances: such an instance is defined by a bottom interval $I \in \calI$ at the level $\ell^*$ and the set of bottom jobs owned by sub-intervals of $I$, along with the solution $x$ restricted to the interval $I$.  This gives a schedule of the jobs inside interval $I$.  In the fourth step, we insert (a large subset of) the top jobs into the constructed schedule to obtain our tentative schedule for the current instance.  
	
	Notice that the algorithm will discard all middle jobs and some top jobs.  The number of middle jobs discarded can be bounded by choosing the parameter $\ell^*$ carefully, while the number of top jobs discarded across the entire algorithm can be bounded using Properties 1) and 2) mentioned above. 
	Property 1) (the chain lengths among top jobs are small) is satisfied because of our conditioning in the first step.
	Property 2) (the precedence constraints between top and bottom jobs are ``loose'') is easy to satisfy using the fact that
	a top interval is at least $2^k$ times longer than the bottom intervals. We can ``shrink'' the fractional support of a top job by a small amount in order to remove the dependence between top and bottom jobs; since a top interval is so long, the shrinking operations only incur a small number of discarded top jobs.
	Finally, we can put together all these pieces to show that the total number of discarded jobs is at most $\epsilon T$, which gives a $(1+\epsilon)$-approximation to the problem.

\subsection{Our Algorithm for General Job Lengths (Theorem \ref{thm:secondresult})}
\label{sec:modifications_to_LR16}
A simple way to extend \cite{LR16} to the case when jobs have arbitrary processing lengths would be the following:
Replace each job $j$ of length $p_j$ by a chain $a_{1,j} \prec a_{2,j} \prec \ldots \prec a_{p_j,j}$ of length $p_j$, where
each $a_{i,j}$ is a ``task" of unit length.
Now, treat these tasks  as separate entities, and schedule them using the algorithm in \cite{LR16}. 
However, as the subproblems in \cite{LR16} are solved independently, this approach does not guarantee that the schedule is non-migratory;
that is, tasks belonging to a single job are scheduled on the same machine.
A natural idea to get around this issue is to assign a job to an interval in the binary laminar decomposition only if {\em all the tasks of the job} belong to that interval. 
This will make sure that  the tasks belonging to a single job are handled in a single recursive call, and are not scheduled independently of each other.
This is the approach we wish to take.
However,
one difficulty that arises is: how to reduce the maximum chain length among the top jobs?
To understand this issue, consider the case when a long chain consists entirely of a single top job. 
In this case,
no matter how we condition, the support of the job will not move down to the lower intervals. 
In general, it seems unavoidable that we have to allow tasks of a same job to go to different branches of recursive calls to keep the chain lengths among top jobs small.
This brings us to the first main technical hurdle an algorithm for general job lengths has to deal with: 
How to keep track of tasks of a single job going into different recursive calls, and still ensure that they get scheduled on the same machine?

A crucial observation about the  algorithm in \cite{LR16} helps us in mitigating the tricky situation:
the total number of conditionings performed to reduce the chain length of the top jobs is small in any recursive call of the algorithm.
Let us call the jobs that the algorithm conditions on {\em special jobs}.
There cannot be more than
$O(\log T \cdot \gamma )$ 
such special jobs at any level of the recursion, where $\gamma$ is the total number of conditionings performed by our algorithm in a single recursive call. 
As we argued earlier, this number is not too large.
Further, for every special job, the entire support of the job becomes concentrated on a single machine.
This is guaranteed by the non-migratory constraints of our machine-indexed LP (see Eq.~\eqref{e:nomigration}).

Our algorithm exploits the above two properties in tandem. 
For the special jobs, we allow tasks of a single job to go to different branches of the recursion. 
On the other hand, for a job that is not special, we ensure that all the tasks belonging to the job are considered by the same recursive call. 
For each special job, our algorithm needs to know exactly which tasks belonging to it go to each branch of the recursion.
Moreover, we also require such tasks to be completely scheduled in that interval.
This is needed to argue that there will be enough free slots to insert the top jobs after the  tasks belonging to bottom jobs are recursively scheduled. 
Our algorithm accomplishes this by introducing a new type of conditioning operation called {\em splitting}, which precisely guarantees the above two invariants. 
  
Using the splitting operation, our algorithm propagates all the tasks of special jobs down to the lowest level of the recursion.
In particular, our algorithm maintains the invariant that the special jobs (more precisely, the tasks belonging to special jobs) will {\em never be a part of the set of top or middle jobs}.
At the lowest level, the tasks of special jobs eventually get scheduled by conditioning on the LP solution.
Since all the tasks of a special job have their support entirely on a single machine due to the LP constraints, this implies that they are all scheduled on the same machine. 

Thus, the key ideas of special jobs and splitting conditioning operations help us to construct a non-migratory schedule of bottom jobs.
Now, we need to extend this schedule to include top jobs, which brings us to the second major technical issue an algorithm for the general lengths case has to solve: 
How can we schedule top jobs in a non-migratory fashion in the slots  left open by the bottom jobs?
This is tricky because: a) Scheduling the bottom jobs is done independently of the top jobs and it is not guaranteed that every job has $p_j$ units of empty slots on some  machine. 
b) The LP solution for top jobs only says that there is enough space to schedule jobs if migration were allowed.
Moreover, we also need to satisfy the precedence relationships among the top jobs, and between the top and bottom jobs.

\medskip
The second key contribution of this paper is a new algorithm to insert the top jobs. Notice that the algorithm of Levey and Rothvoss \cite{LR16} does give such a procedure, but only for the unit-length case. 
In contrast, our algorithm needs to deal with the substantially harder case of arbitrary job sizes and the non-migratory constraints.
Our method proceeds in two stages.
In the first stage, we build a {\em tentative schedule} which allows migration of jobs, but guarantees that the precedence constraints between the top and the bottom jobs are satisfied.
To that end, we identify an interval $[r_j, d_j]$ for every job $j$ such that if $j$ is scheduled within this interval,
then the precedence constraints between the top and the bottom jobs are satisfied. 
Here, $[r_j, d_j]$ may be shorter than the interval defined by the support of the LP solution for $j$.
This leads to discarding some tasks; however, a simple argument based on an extension of Hall's theorem proved in \cite{LR16} shows that the number of discarded tasks is small. This stage is similar to the procedure in \cite[Section~5.1]{LR16}.

The more difficult question is how to convert the migratory schedule from the first stage into a non-migratory schedule respecting the precedence constraints among the top jobs; this is another technical contribution of our paper.
To solve this problem, our algorithm considers the jobs in the Earliest Deadline First (EDF) order of their $d_j$ values. 
Suppose $B_j \ge r_j$ is the first time slot where the job $j$ can be scheduled respecting {\em all} the precedence constraints.
Then our algorithm  assigns $j$ to the machine $i$ on which it will have {\em the Earliest Completion Time} (ECT).
Suppose $C_j$ is its completion time. 
Now we make the following crucial observation about the ECT policy: 
If one looks at the interval $[B_j, C_j]$ in which $j$ was scheduled, there can be at most $p_j$ empty slots on any other machine $i' \neq i$.
If this were not true, then it implies that there is another machine $i' \neq i$ and time slot $t' < C_j$ such that job $j$ could have been feasibly scheduled in the interval $[B_j, t']$ on the machine $i'$.
This leads to a contradiction with our policy.
This observation, combined with  the invariant maintained by our algorithm -- that the chain lengths are small among the top jobs -- guarantees that the number of slots we waste because of the precedence and non-migratory constraints is not too large. 

Finally, it is possible that there is no machine $i$ on which job $j$ can be scheduled completely within its deadline $d_j$. 
In this case, we find the machine with the maximum number of empty slots, and schedule the job partially.
We discard the tasks that we could not schedule.
This {\em fragmentation} also leads to some more slots being wasted.
However, we argue that the EDF+ECT policy guarantees that the number of  slots wasted due to fragmentation is also small.
In the end, somewhat surprisingly, we prove that the total number of tasks our algorithm discards is asymptotically the same as in \cite{LR16}.

\subsection{Our Algorithm for Problem with Communication Delays (Theorem \ref{thm:delayresult})}
\label{sec:ideas_comm_delays}

Now we give a high-level overview of how the LP hierarchy framework easily extends to the problem with communication delays where $\max_{j,j'}{c_{j,j'}} = O(1)$, i.e, $Pm|\text{prec}, p_j = 1, \text{pmtn}, c_{j,j'}|C_{\max}$, which we call the ``delay problem''. 
Recall that in the delay problem, along with precedence constraints, the algorithm also needs to enforce communication delay constraints; that is, if $j \prec j'$ and $j$ and $j'$ are scheduled on two different machines, then $j'$ cannot start earlier than $C_j + c_{j,j'}$, where $C_j$ is the completion time of job $j$.
For simplicity we assume that all communication delays are 1. The arguments directly extend to the more general case $c_{j,j'} = O(1)$.

It should be clear from the overview of our algorithm  for $Pm|\text{prec}, \text{pmtn}|C_{\max}$, which we call the ``no-delay problem'', that it is almost impossible to control how jobs are scheduled in the different branches of recursive calls. 
Hence, there is no easy way our algorithm can make scheduling choices taking into account the delay constraints.  
Instead, we let our recursive algorithm make choices without considering the communication delay constraints; when the actual assignment of tasks to time slots is carried out, we will enforce the communication delay constraints, assuming the worst possible scenario.
Before we explain our strategy to do so, let us first summarize the three points at which our algorithm for the no-delay problem actually commits to the assignment of tasks to time slots.

\begin{enumerate}
	\item Scheduling tasks by conditioning. At the lowest level of recursive call, our algorithm schedules the tasks by conditioning on the LP hierarchy solution. 
        \item Scheduling top jobs.
        \item Scheduling discarded tasks.
\end{enumerate}

Enforcing the communication delay in the first and third steps is rather easy: Our LP for the problem, which includes communication delay constraints, guarantees that if two jobs are scheduled by conditioning in the same recursive call, then the communication delay constraints are satisfied.
To argue that communication delay constraints are also satisfied if they belonged to different branches of recursion needs a bit more work, but is not difficult. 
On the other hand, as the number of discarded tasks is small, we can afford to create three new private slots for each discarded task, and schedule the task in the middle, leaving the other two slots empty. 
This will take care of communication delay constraints, no matter how other tasks are scheduled.
Thus, it only remains to argue how we guarantee that the communication delay constraints are satisfied while inserting the top jobs. 

Here, we go back to a central idea in \cite{LR16}: If chain lengths are small (bounded by $\epsilon T$), then Graham's algorithm already gives a $(1+\epsilon)$-approximation to the no-delay problem.
This crucial observation easily extends in the presence of communication delay constraints:
It is not hard to argue that Graham's list scheduling algorithm gives a $(1+2\epsilon)$-approximation to makespan if the chain lengths are small.
Now note that in our algorithm the maximum chain length among the top jobs is small due to conditioning. 
We make use of this fact, along with several new observations, to give an extension of our algorithm for inserting top jobs for the no-delay problem  to the delay problem.

To argue that communication delay constraints are also satisfied even if one job gets scheduled by conditioning and the other job gets scheduled as a top job needs some care and some additional tricks.
However, the overall argument still relies mainly on the above three cases.

\subsection{Note About Practical Applications}
Besides being fundamental problems, the scheduling models studied in this paper have gained much importance recently in the context of datacenter scheduling literature; 
see a recent workshop on this topic \cite{ttic} for more pointers.
We give some context here. 
Programming models such as Dryad \cite{Dryad} or SparkSQL \cite{spark} compile scripts into job DAGs, which give rise to precedence-constrained scheduling problems. 
See \cite{grandl2016graphene, AgrawalLLM16} and references therein. 
Similarly, communication delay problems arise in workloads from MapReduce systems and in the {\em model parallelism} paradigm \cite{amar3} for training complex machine learning models on large clusters; see \cite{amar3, bubblepaper, amar1, amar2} and references therein.
We do not claim that the algorithms in this paper can be used in these applications directly;
however, the framework of the algorithms proposed in this paper and in \cite{LR16} shares many parallels to the heuristic for DAG scheduling developed (independently) in \cite{grandl2016graphene}, which can be viewed as replacing the Sherali-Adams based ``conditioning" step by a brute-force search;
see here \cite{janatalk} for an exposition.

\ifsoda\else
\section{Basics of the Sherali-Adams Hierarchy}
\label{sec:SA}
In this section, we formally state some basic facts about the Sherali-Adams hierarchy that we will need in our proofs. We refer the reader to \cite{Laurent, sherali, Lasserre, Lovasz, Rothvoss13}  for an extensive introduction to hierarchies.  
The purpose of this section is to formally  state the  properties intuitively described in Section~\ref{sec:SA-intuition} so that there is no ambiguity in the \ifsoda proof of Theorem~\ref{thm:secondresult}. \else proofs of Theorems~\ref{thm:secondresult} and~\ref{thm:delayresult}. \fi Hence, it can be skipped on the first reading.
\ifsoda\else
	However, the integrality gap result in Section~\ref{sec:integrality} needs to use the details in the definition of the SA hierarchy.
\fi

Assume we have a feasibility LP of the form $Ax \leq b$, which includes the constraints $0 \le x_i \le 1$ for all $i \in [n]$. The set of feasible integral solutions is defined as $\calX = \{x \in \{0, 1\}^n: Ax \leq b\}$.  It is convenient to think of each $i \in [n]$ as an event, and in a solution $x \in \{0, 1\}^n$, $x_i$ indicates whether the event $i$ happens or not. 

The idea of the Sherali-Adams hierarchy is to strengthen the original LP $Ax \leq b$ by adding more variables and constraints. Of course, each $x \in \calX$ should still be a feasible solution to the strengthened LP (when extended to a vector in the higher-dimensional space). For some $r \geq 1$, the $r$-th round of the Sherali-Adams lift of the linear program has variables $x_S$ for every $S\subseteq [n]$ of size at most $r$. For every solution $x \in \calX$, $x_S$ is intended to indicate whether all the events in $S$ happen in the solution $x$; that is, $x_S = \prod_{i \in S}x_i$. Thus each $x \in \calX$ can be naturally extended to a 0/1-vector in the higher-dimensional space defined by all the variables.

To derive the set of constraints, let us focus on the $j$-th constraint $\sum_{i=1}^na_{j,i} x_i \leq b_j$ in the original linear program. Consider two subsets $S, T \subseteq [n]$ such that $|S| + |T| \leq r - 1$. Then the following constraint is valid for all $x \in \calX$:
\begin{align*}
	\prod_{i \in S}x_i\prod_{i \in T}(1-x_i)\left(\sum_{i = 1}^n a_{j, i}x_i - b_j\right) \leq 0.
\end{align*}

To {\em linearize} the above constraint, we expand the left side of the above inequality and replace each monomial with the corresponding variable $x_{S'}$.  Then, we obtain the following linear constraint:
\begin{align}
	\sum_{T' \subseteq T}(-1)^{|T'|} \left(\sum_{i =1}^na_{j, i}x_{S \cup T' \cup \set{i}} - b_jx_{S \cup {T'}}\right) \leq 0. \label{inequ:SA}
\end{align}

The $r$-th round of the Sherali-Adams lift contains the above constraint for all $j, S, T$ such that $|S| + |T| \leq r-1$, and the trivial constraint that $x_{\emptyset} = 1$. For a linear program $\calP$ and an integer $r \geq 1$, we use $\SA(\calP, r)$ to denote the $r$-th round Sherali-Adams lift of $\calP$. We also view $\calP$ (resp.\ $\SA(\calP, r)$) as the polytope of feasible solutions to the linear program $\calP$ (resp.\ $\SA(\calP, r)$). For every $i \in [n]$, we identify the variable $x_{i}$ in the original LP and $x_{\{i\}}$ in the lifted LP. 

A simple observation is that $x_{S_1} \geq x_{S_2}$ if $S_1 \subseteq S_2$, for a valid solution $x \in \SA(\calP, r)$ and $|S_1| \leq |S_2| \leq r$. For consider the case where $S_2 = S_1 \cup \{i\}$ for some $i \notin S_1$: linearizing the constraint $x_i\leq 1$ multiplied by $\sum_{i' \in S_1}x_i$ gives the constraint $x_{S_2} \leq x_{S_1}$. This implies that all the variables have values in $[0, 1]$, as $x_\emptyset = 1$.

\paragraph{Conditioning} Let $x \in \SA(\calP, r)$ for some linear program $\calP$ on $n$ variables and $r \geq 2$. Let $i \in [n]$ be an event such that $x_{i} > 0$; then we can define a solution $x' \in \SA(\calP, r-1)$ obtained from $x$ by ``conditioning" on the event $i$. For every $S \subseteq [n]$ of size at most $r-1$, $x'_S$ is defined as
\begin{align*}
	x'_S:= \frac{x_{S \cup \{i\}}}{x_i}.
\end{align*}

\begin{observation}
	Let $x'$ be obtained from $x \in \SA(\calP, r)$ by conditioning on some event $i$, for some $r \geq 2$.  Then $x' \in \SA(\calP, r-1)$ and $x'_i = 1$.
\end{observation}
\begin{proof}
	$x'_i = \frac{x_{\{i\} \cup \{i\}}}{x_i} = \frac{x_i}{x_i} = 1$, and $x'_\emptyset  = \frac{x_{\emptyset \cup \{i\}}}{x_i} = \frac{x_i}{x_i} = 1$. The constraint \eqref{inequ:SA} on $x'$ for $j, S$ and $T$ is implied by \eqref{inequ:SA} on $x$ for $j, S \cup \{i\}$ and $T$.
\end{proof}

\begin{observation}
	Let $x \in \SA(\calP, r)$ for some $r \geq 2$ and $x' \in \SA(\calP, r-1)$ be obtained from $x$ by conditioning on some event $i$. Then, if $x_{i'} \in \{0, 1\}$ for some $i' \in [n]$, then $x'_{i'} = x_{i'}$.
\end{observation}
\begin{proof}
	If $x_{i'} = 0$, then $x'_{i'} = \frac{x_{\{i'\} \cup \{i\}}}{x_i} = 0$ since $x_{\{i'\} \cup \{i\}} \leq x_{i'} = 0$. Consider the case $x_{i'} = 1$. Expanding the constraint $(1 - x_{i})(1-x_{i'}) \geq 0$ gives the constraint $1 - x_i - x_{i'} + x_{\{i', i\}} \geq 0$. This implies $x_i = x_{\{i', i\}}$. Thus, $x'_{i'} = \frac{x_{\set{i,i'}}}{x_i} = 1$.
\end{proof}

The observation says that once an event $i'$ happens with extension 0 or 1 w.r.t.~a lifted solution $x$, then it will always happen with the same extension (0 or 1) w.r.t any solution $x'$ obtained from $x$ by conditioning. To understand the conditioning operation and the above observations better, it is useful to consider the ideal case where $x$ corresponds to a convex combination of integral solutions in $\calX$. Then we can view $x$ as a distribution over $\calX$. Then, conditioning on the event $i$ over the solution $x$ corresponds to conditioning on $i$ over the distribution $x$. 

\fi
\section{Minimizing Makespan Under Precedence Constraints For General Job Lengths}
\label{sec:makespan}
In this section we consider the problem of minimizing the makespan when jobs have arbitrary processing lengths, and prove Theorem \ref{thm:secondresult}.
The input to the problem consists of a set of jobs $J$, where each job $j \in J$ has a processing length $p_j$,
and the precedence constraints $\prec$ between jobs. 
We imagine $j$ as being made up of $p_j$ tasks (or atoms) $a_{1,j}, a_{2,j}, \ldots, a_{p_j,j}$. 
The set of tasks of a job $j$ is denoted by $A(j)$. 
Similarly, we define $A(J')$ as the set of tasks of jobs in $J'$ for every $J' \subseteq J$. 
Formally, $A(J') := \union_{j \in J'} A(j)$. 
We take the transitive closure of the precedence relations among jobs; that is, if $j_1 \prec j_2$ and $j_2 \prec j_3$, then we have $j_1 \prec j_3$.

We use the notation $a \sim a'$ to mean that the tasks $a$ and $a'$ belong to same job. 
Formally, $a \sim a'$ if there is $j \in J$ such that $a, a' \in A(j)$.
The precedence constraints between jobs are extended to the tasks in the following natural way. 
For each job $j$, first we assume that  $a_{1,j} \prec a_{2,j} \prec \ldots \prec a_{p_j,j}$. 
Consider any two jobs $j$ and $j'$ with precedence constraint $j \prec j'$. 
Then, for tasks $a \in A(j)$ and $a' \in A(j')$,  we introduce a precedence constraint $a \prec a'$.  
Sometimes we also overload the precedence relation and write $a \prec j'$ to mean that every task of $j'$ needs to be scheduled after the task $a$; that is, there is $j \in J$ such that $a \in A(j)$ and  $j \prec j'$.

During the description of our algorithm  we often go back and forth between two views. 
In the {\em task view} of the problem, we imagine our input as consisting of a set of tasks $A(J)$,  each task  with unit size, and precedence constraints as defined above. 
In the {\em job view}, we treat each job as a separate entity.

Our goal is to assign each job to a single machine and specify a schedule of tasks such that the precedence constraints among jobs are satisfied. 
Our objective is to minimize the makespan of the schedule, which is defined as the completion time of the last task. 
Formally, we define a valid schedule as follows. 

	A  schedule $\calS$ for a subset $A' \subseteq A(J)$ of tasks on an interval $I \subseteq [T]$ with integer length is a function 
	$\calS: A' \rightarrow [m] \times I$ that indicates the (machine, time slot) pair that each task is assigned to.  For every $a \in A'$, we then use $\calS_\mac(a)$ and $\calS_\tim(a)$ to denote the first and second component of $\calS(a)$ respectively. 

\begin{definition}
\label{def:valid}
	A schedule $\calS$ for $A' \subseteq A(J)$ is \emph{valid} if it satisfies the following constraints.
	\begin{itemize}
	\item Capacity Constraints: for every two tasks $a \neq a' \in A'$, we have $\calS(a) \neq \calS(a')$. 
	\item No-migration Constraints:  For every pair of tasks $a \sim a' \in A'$, we have $\calS_\mac(a) = \calS_\mac(a')$.
	\item Precedence Constraints:  For every pair of tasks $a \prec a' \in A'$, we have $\calS_\tim(a) < \calS_\tim(a')$.
		\end{itemize}
\end{definition}

So, a valid schedule for $A'$ guarantees that for a job $j$ all the tasks of $j$ in $A'$ are assigned to a single machine.  
We define  the completion time of job $j$ as the time at which the last task of $j$ is scheduled and denote it by $C_j$. 
If $j \prec j'$, then $C_j < C_{j'}$. 
Throughout the section, we use $N = \sum_{j} p_j$ to denote the total number of tasks.  We can assume by appropriate scaling of the input that $N\le n^2/\epsilon$; see \ifsoda\else Appendix~\ref{sec:N-big} \fi for details.

\subsection{LP Relaxation}
Our algorithm that proves Theorem \ref{thm:secondresult} is based on rounding the Sherali-Adams lift of a natural LP for the problem. 
The variables of  LP are $x_{(a, i, t)}$, which are intended to be 1 if the task $a \in A(J)$ is assigned to machine $i \in [m]$ at time $t \in [T]$. 
A natural LP formulation to decide if there is a valid schedule with makespan at most $T$ is as follows:

\ifsoda\else
	\noindent
	\begin{minipage}{0.27\textwidth}
\fi
		\begin{alignat}{2}
			\sum_{i, t}x_{(a, i, t)} &= 1 &\quad &\forall a   \label{e:taskschedule} \\ 
			\sum_{a} x_{(a, i, t)} &\leq 1 &\quad &\forall i, t \label{e:capconstraints} \\
			\nonumber
	\end{alignat}
\ifsoda\else
	\end{minipage}
	\begin{minipage}{0.73\textwidth}
\fi
		\begin{alignat}{2}
			\sum_{i, t' \leq t+1} x_{(a', i, t')} &\leq \sum_{i, t' \leq t} x_{(a, i, t')} &\quad &\forall a \prec a', t \in [T-1] \label{e:precconstraints}\\
			\sum_{t}x_{(a', i, t)} &= \sum_{t}x_{(a, i, t)} &\quad &\forall a \sim a', i \label{e:nomigration} \\
			x_{(a', i, t)} &\in [0,1]  &\quad &\forall  a \in A(J),  i \in [m],  t \in [T] \label{e:nonnegativity}
	\end{alignat}
\ifsoda\else
	\end{minipage}
\fi
\bigskip

The constraints \eqref{e:taskschedule} guarantee that every task is feasibly scheduled, while \eqref{e:capconstraints} are the capacity constraints. 
The constraints (\ref{e:precconstraints}) impose the precedence order among tasks. 
Here we appeal to the task view of the problem. 
Finally,~\eqref{e:nomigration} are intended to enforce the no-migration constraints: in a non-migratory schedule all the tasks $A(j)$ of a job $j$ are scheduled on a single machine. 
Hence a valid schedule satisfies those constraints.
Therefore, if there is an optimal integral solution with makespan at most $T$, then there is a feasible solution to the LP. 
We use $\calP(T)$, and simply $\calP$ when $T$ is clear from the context, to denote the polytope defined by the above LP relaxation.

\medskip
Towards proving the main result (Theorem \ref{thm:secondresult}), we first design a LP rounding algorithm that only schedules a subset $A(J) \setminus \taskdis$ of tasks.
Our main goal in this section is to prove the following lemma.

\begin{lemma}
	\label{lem:mainmakespan}
	 Let $T$ be the smallest value for which the Sherali-Adams lift of LP (\ref{e:taskschedule}-\ref{e:nonnegativity}) to $r = (\log n)^{O((m^2/\epsilon^2).\log \log n)}$ rounds has a feasible solution $\vecx$. In time $n^{O(r)}$ we can find a valid schedule $\calS : A(J) \setminus \taskdis \rightarrow [m] \times [T]$ with $|\taskdis| \leq \epsilon T$. 
\end{lemma}

We give a brief sketch of the proof that a partial schedule $\calS$ satisfying the guarantees of the above lemma can be easily extended to a valid schedule $\calS^*$ for  $A(J)$ with makespan $[T + |\taskdis|]$ as follows.
Consider a discarded task $a$. 
Let $t$ be the earliest time slot such that scheduling $a$ at $t$ satisfies all the precedence constraints.
We create a separate time slot for $a$ at time $t$, and schedule it there. 
We do not schedule any other task at $t$, and shift the entire schedule of tasks following $t$ by one time step. Since this time slot can be created on any machine, the non-migratory constraints can be satisfied.
We repeat this procedure for every discarded task. 
Hence scheduling the tasks in $\taskdis$ increases the makespan by an additive factor of $|\taskdis|$.
As $|\taskdis| \leq \epsilon T$, this implies that makespan of our final schedule $\calS^*$ is at most $(1+\epsilon)T$.
As the Sherali-Adams lift of the linear program is a valid relaxation of the optimal solution, the optimal makespan has to be at least $T$. 
Moreover, one can solve the lifted linear program in time $n^{O(r)}$. 
Putting together all these facts, we conclude that Lemma \ref{lem:mainmakespan} implies Theorem \ref{thm:secondresult}.

\smallskip

\subsection{Rounding Algorithm}

Our algorithm to prove Lemma \ref{lem:mainmakespan} is a generalization of the algorithm in \cite{LR16}.
Hence, for easier reading we try to keep the notation and structure of our paper similar to~\cite{LR16} as much as possible. 

\medskip

We begin by partitioning the interval $[T]$ into a balanced binary family $\calI$ of intervals of length $T, T/2, T/4, \ldots, 2, 1$. 
W.l.o.g., we can assume that $T$ is a power of two using the padding trick in \cite{LR16}. One way to visualize $\calI$ is to think of a balanced binary tree, where the root node corresponds to the interval $[T]$, and nodes at level $\ell$ correspond to the $2^{\ell}$~intervals obtained by partitioning $[T]$ into sub-intervals of length $T/2^{\ell}$.

\smallskip

We define some notation that will be used throughout the paper.  For an interval $I \in \calI$ and a subset of jobs $J' \subseteq J$ we define 

\begin{equation*}
\label{e:jobcontained}
J'(I, \vecx) = \bigg \{j \in J':  \forall a \in A(j), \sum_{i} \sum_{t \in I} x_{(a, i, t)} = 1  \bigg \} 
\end{equation*}
as the subset of jobs with support completely in the interval $I$ in the LP-hierarchy solution $\vecx$. Similarly, we define 

\begin{equation*}
\label{e:taskcontained}
A(I,J',\vecx) = \bigg\{a \in A(J'):  \sum_{i} \sum_{t \in I} x_{(a, i, t)} = 1 \bigg\}
\end{equation*}
as the set of tasks belonging to jobs in $J'$ that have their entire support in the interval $I$. We emphasize that $I$ may contain partial support of some tasks of $A(J')$, but they are not included in the set  $A(I,J',\vecx)$. We use $A(I,j,\vecx) \subseteq A(j)$ as a shorthand for $A(I, \set{j}, \vecx)$, i.e, as the subset of tasks of job $j$ that have entire support in $I$.

Finally, we define four global constants $k, \delta, K$ and $K'$. We set $k = \frac{O(1) m}{\epsilon} \cdot \log \log T$ to be large enough and $\delta = \frac{\epsilon}{8k^2m2^{2k^2} \log T}$. Let $K = \frac1\delta \cdot m k^ 2  \cdot  2 ^ { k ^ 2 }$ and $K' = K\cdot 2 ^ {k ^ 2} =  \frac1   \delta \cdot   m   k ^2 \cdot 2^{2k^2}$. The parameter $\delta$ is used to define what constitutes  a long chain. 
The roles of constants $k, K$ and $K'$ will become clear when we give the description of our algorithm.  For a set $J' \subseteq J$ of jobs, we use $\Delta(J')$ to denote the maximum possible \emph{total size} of jobs in a precedence-chain in $J'$. 

\smallskip
The recursive algorithm $\recursive$ used to prove Lemma \ref{lem:mainmakespan} is given as Algorithm~\ref{alg:recursive}.
The main claim  of \cite{LR16} is that if we fix any interval $I \in \calI$ and consider the set of tasks that are entirely scheduled in the interval $I$ by the LP solution $\vecx$, then one can find a partial schedule of the tasks in the interval $I$ that discards few tasks. 
The main contribution of this paper is that a similar statement can be shown even when jobs have arbitrary lengths and we enforce no-migration constraints. 
In order to give the main lemma, we shall first define a \emph{partial-scheduling} instance.

\begin{definition}
	\label{def:partial-scheduling} In a partial-scheduling instance, we are given an interval ${I^*} \in \calI$, an LP-hierarchy solution $x \in \SA(\calP, r)$ for $r = \log |{I^*}| \log T   \cdot K'$, a set $J^* \subseteq J({I^*}, x)$ of jobs that have complete support in ${I^*}$ according to $x$, and a special set $\jspl^* \subseteq J$ of jobs disjoint from $J^*$, with $|\jspl^*| \leq \log \frac{T}{|I^*|}\cdot K$. Further, we are given a function $\sigma:\jspl^* \to [m]$ such that for every $j \in \jspl^*$, we have
	\begin{enumerate}
		\item $j$ is only scheduled on $\sigma(j)$ in $x$, i.e, $x_{(a, i, t)} = 0$ if $a \in A(j)$ and $i \neq \sigma(j)$, 
		\item every task $a$ of $j$ is completely scheduled in ${I^*}$ w.r.t.~$x$, or not scheduled in ${I^*}$ at all. That is, $\sum_{t \in {I^*}} x_{(a, \sigma(j), t)}$ is either $0$ or $1$.
	\end{enumerate}
	The goal of the scheduling problem is to schedule (a subset of) $A(J^*) \cup A({I^*}, \jspl^*, x)$ in ${I^*}$. We denote the partial-scheduling instance by $({I^*}, J^*, \jspl^*, \sigma, x)$.
\end{definition}

Observe in the above definition that for each special job $j \in \jspl^*$ in a partial-scheduling instance,  we are given a machine $\sigma(j)$ on which $j$ must be scheduled. Moreover, we know exactly the set of special-job tasks that must be scheduled: these are the tasks that are completely scheduled in $I^*$ on $\sigma(j)$ in the LP solution $x$, and the other tasks are not scheduled at all.  Notice that the tasks to be scheduled are consecutive in the task chain for job $j$ due to the precedence constraints.

The following main lemma bounds the number of tasks discarded by the recursive algorithm $\recursive$ that solves the partial scheduling problem.

\begin{lemma}[Main Lemma]
\label{lem:critical}
Let $({I^*}, J^*, \jspl^*, \sigma, x)$ be a partial-scheduling instance, and let $A^* := A(J^*) \cup A({I^*},\jspl^*,\vecx)$ be the tasks we need to schedule. Then, $\recursive$ returns a valid schedule $\mathcal{S}$ for $A^*\setminus \taskdis$ of makespan ${I^*}$ for a set $\taskdis$ of discarded tasks of size at most
$$
|\taskdis| \leq \frac{\epsilon}{2} \cdot \frac{\log |{I^*}|}{\log T} \cdot |{I^*}|  + \frac{\epsilon}{2m} \cdot |A^*|.
$$
Moreover, for every job $j \in \jspl^*$, the set $A({I^*},j,\vecx)$ of tasks is scheduled on the machine $\sigma(j)$.
\end{lemma}

Observe that Lemma \ref{lem:mainmakespan} follows immediately by instantiating the above lemma for the entire interval $[T]$ and the set of jobs $J^* = J$ and $\jspl^* = \emptyset$. 
Then the total number of discarded tasks will be 
$$
\frac{\epsilon}{2} \cdot \frac{\log T}{\log T} \cdot T  + \frac{\epsilon}{2m} \cdot |A(J)| \leq \epsilon \cdot T,
$$
where we used the fact that $|A(J)| \leq mT$. 
Definition~\ref{def:partial-scheduling} also requires an LP-hierarchy solution $x$ of level $r$ at least $(\log T)^2    \cdot K'$.

We have $T \leq N \le \mathrm{poly}(n)$. As we set $k = \frac{O(1) m}{\epsilon} \cdot \log \log T$ and  $\delta = \frac{\epsilon}{8k^2m2^{2k^2} \log T}$  and $K' = \frac1\delta \cdot mk^2 \cdot 2^{2k^2}$, we have
$$
r = (\log T)^2    \cdot K' = (\log n)^{O((\frac{m}{\epsilon})^2 \log \log n)}.
$$

The term that determines the asymptotic running time of our algorithm in the above equation is  $2^{2k^2}$, which is at most $(\log n)^{O((\frac{m}{\epsilon})^2 \log \log n)}$.

\medskip

\begin{algorithm}[h]
	\caption{\textsf{PARTIAL-SCHEDULE}$\left({I^*}, J^*, \jspl^*, \sigma,  x\right)$}
	\label{alg:recursive}
	\textbf{Input:} a partial-scheduling instance $({I^*}, J^*, \jspl^*, \sigma, x)$ satisfying Definition \ref{def:partial-scheduling}\\
	\textbf{Output:} a schedule $\calS^*: A(J^*) \cup A({I^*}, \jspl^*, x) \setminus \taskdis \rightarrow [m] \times I^*$ for some $\taskdis$
	\vspace*{-5pt}

	\noindent\rule{\linewidth}{0.2pt}
	\begin{algorithmic}[1]
		\State \textbf{if} $|{I^*}| < 2^{k^2}$ \textbf{then} schedule tasks by conditioning and return
		\State \textbf{for} {every job $j \in \jspl^*$} \textbf{do}: $x \gets \mysplit(x,  j)$ \label{step:initial-condition}
		\While{there exists ${{I}} \in \calI^*_{0} \cup \ldots  \cup \calI^*_{k^2-1}$ and a chain $\mathcal{C}$ of jobs in $J^*$ owned by $I$ with total size at least $\delta|{{I}}|$} \label{step:reduce-chain-start}
		\State $j \gets$ the first job in $\calC$, $a \gets $ last task of $j$
		\State take $(i, t)$ such that $x_{(a, i, t)} > 0$  with the largest $t$
		\State $x \gets $ $x$ conditioned on the event $(a, i, t)$
		\State $J^* \gets J^* \setminus \{j\}, \jspl^* \gets \jspl^* \cup \{j\}, \sigma(j) \gets i$
		\State $x \gets \mysplit(x, j)$
		\EndWhile \label{step:reduce-chain-end}
		
		\State Partition the jobs in the set $J^*$ as follows:
		
		$ \jtop^* = \union^{\ell^*-k-1}_{\ell = 0} J^*_{\ell}(\vecx) $; 
		$ \jmid^* = \union^{\ell^*-1}_{\ell = \ell^*-k} J^*_{\ell}(\vecx) $;
		$ \jbot^* = \union^{\log T^*}_{\ell = \ell^*} J^*_{\ell}(\vecx) $,
		
		\State where $\ell^* \in \{k, \ldots, k^2\}$ is chosen satisfying the condition below:
		\begin{eqnarray*}
			|A(\jmid^*)| \leq \frac{\epsilon}{4} \cdot \frac{T^*}{\log T} + \frac{\epsilon}{2m} \cdot \left(|A(\jmid^*)| + |A(\jtop^*)| \right)
		\end{eqnarray*}
		
		\For{every interval ${I} \in \calI^*_{\ell^*}$} 
		\State  $\recursive\big({I}, \jbot^*({I},\vecx), \jspl^*, \sigma, \vecx\big)$
		\EndFor
		\State Insert $\jtop^*$ into ${I^*}$ using Lemma \ref{lem:topdiscardedtsks}.			
	\end{algorithmic}
\end{algorithm}	

\begin{algorithm}[h]
	\caption{$\mysplit(x, j)$, where $j \in \jspl^*$} \label{alg:splitting}
	
	\begin{algorithmic}[1]
		\For{every ${I} \in \calI^*_{k^2}$ from left to right} 
		\State $t^* := \max \left \{t \in {I}: (\exists a \in A(I^*, j, \vecx)) \ x_{(a, \sigma(j), t)} > 0 \right \}$
			\If {$t^*$ is defined}
				\State let $a^* \in A(j)$ be any task with $x_{(a^*,\sigma(j), t^*)} > 0$
				\State $x \gets x$ conditioned on $(a^*, \sigma(j), t^*)$
			\EndIf
		\EndFor
		\State \textbf{return} $x$
	\end{algorithmic}
\end{algorithm}

From now on we focus solely on proving Lemma \ref{lem:critical}, and we assume we are given an instance $(I^*, J^*, \jspl^*, \sigma, x)$. During our algorithm, $I^*$ does not change, but we shall move jobs from $J^*$ to $\jspl^*$ and extend $\sigma$ accordingly. The LP-hierarchy solution $x$ will also be updated using the conditioning operation.
Let $T^*$ always denote $|I^*|$; note that $T^*$ is some power of two. Let $\calI^*$ denote the set of intervals in $\calI$ that are sub-intervals of $I^*$. 
For an integer $\ell \in [0, \log T^*]$, we use $\calI^*_\ell$ to denote the intervals in $\calI^*$ with length $T^*/2^\ell$; thus,  $\calI^* : = \calI^*_{0} \cup \calI^*_{1} \cup \ldots \cup \calI^*_{\log T^*}$, and each $\calI^*_{\ell}$ contains $2^\ell$ intervals of length $\frac{T^*}{2^\ell}$ each.
  
\smallskip
For a job $j \in J^*$, let $I' \in \calI^*$ be the interval of smallest length such that it contains the entire support of $j$.  We say that $I'$ {\em owns} the job $j$. If $I' \in \calI^*_{\ell}$, then we also say that the level $\ell$ {\em owns} the job $j$. 
We use the notation $J^*_{\ell'}(\vecx)$ to denote the subset of jobs in $J^*$ that are owned by a level $\ell' \in [\log T^*]$. 
That is,
$$
J^*_{\ell'}(\vecx) : = \bigg \{j \in J^*: \text{ level } \ell' \text{ owns } j \bigg \}.
$$

Contrast this with notation $J^*(I,\vecx)$, which indicates the set of jobs that have full support in the interval $I$. Also notice that only jobs in $J^*$ are owned by intervals or levels.
We say that ``the algorithm conditions on job $j$'' or ``condition on an event $(a, i, t)$" to mean that our algorithm conditions on the event $x_{(a,i,t)} = 1$
for a task $a \in A(j)$. Similarly, we use the phrase ``the algorithm conditions on task $a$".  
Note that as the solution $x$ changes due to conditioning,
the sets of jobs owned by certain intervals and levels can change as well.
Recall that we call a subset of jobs $J' \subseteq J$ a {\em chain} if the precedence relation $\prec$ gives a {\em total ordering on $J'$}. For any subset  $J'$ of jobs, $\Delta(J')$ is defined as  the maximum of $\sum_{j \in J''} p_j$ over all chains $J'' \subseteq J'$.

\subsection{Main Steps of $\recursive$}

Now we give details about our algorithm, which consists of five main steps. The pseudo-code for the algorithm, which we call $\recursive$, is given in Algorithm~\ref{alg:recursive}.  

\medskip

\noindent \textbf{Step 1: Reducing Chain Length Among Top Jobs.} 
Let $\hat{J} : = \bigcup^{k^2-1}_{\ell = 0} J^*_{\ell}(\vecx) $ be the set of jobs that are owned by the first $k^2-1$ levels of $\calI^*$ in the solution $\vecx$. 
By appropriately choosing certain jobs in $\hat{J}$ and conditioning on events of the form $x_{a,i,t} = 1$, our algorithm maintains the invariant that there are no long chains in $\hat{J}$. 
In particular, the maximum chain length $\Delta(\hat{J})\leq k^2 \delta T^*$.

The LP-hierarchy solution $\vecx$ changes in the following way. 
For every job $j$ on which our algorithm does the conditioning, the entire support of job $j$ gets concentrated on a single machine. 
The reason is that when we condition on an event $x_{a, i, t} = 1$, the non-migratory constraints in our LP (Eq. \ref{e:nomigration}) force all the other tasks $a'$ of $j$ to also have their entire support on machine $i$. 
Since conditioning can only shrink the support of jobs, this property remains true regardless of future conditioning operations.   
For every job $j$ on which our algorithm does conditioning, we define $\sigma(j)$ as the machine on which the entire support of $j$ resides in $\vecx$.
Further, we insert the job into the set $\jspl^*$ of special jobs  and delete it from $J^*$. 
Hence, by induction, we are guaranteed that in any recursive call  to $\recursive$ with input parameters $J'$ and $\jspl'$, the invariant $J' \cap \jspl' = \emptyset$ is satisfied.

Another consequence of conditioning is that the support of some jobs in the set $\hat{J}$ shrinks, and they move down the levels in $\calI^*$. 
Note however that a job $j \in \hat{J}$ can move down only $k^2$ levels. 
We use this observation along with a simple counting argument to show that the total number of conditioning operations required to reduce the chain length among jobs in the set $\hat{J}$ is at most $K'$, which is independent of the length of the interval $I^*$.
In particular, we prove the following lemma in Section \ref{label:chain}. 

	\begin{restatable}{lemma}{cutchains}
	\label{lem:cutchains}
		Let $\hat{J} := \bigcup^{k^2-1}_{\ell = 0} J^*_{\ell} (\vecx) $ denote the set of jobs owned by the intervals in $\calI^*_0, \calI^*_1, \ldots, \calI^*_{k^2-1}$ according to the LP hierarchy solution $\vecx$. Then, after line~\ref{step:reduce-chain-end} in $\recursive$,  we have that
		$\Delta(\hat{J}) \leq k^2\delta T^*$. Moreover, the number of iterations we run the loop on line~\ref{step:reduce-chain-start} is at most $K$.
	\end{restatable}
	
\medskip

\noindent \textbf{Step 2: Splitting Special Jobs.} For every job in the set $\jspl^*$, we perform the operation of {\em splitting}, which is given by the procedure $\mysplit$ defined in Algorithm~\ref{alg:splitting}. Recall that for every $j \in \jspl^*$, $\sigma(j)$ is defined. 
Our algorithm guarantees that all the tasks of $j$ are scheduled on $\sigma(j)$.
The idea behind splitting is to ensure that every task of a job $j \in \jspl^*$ is pushed to the lowest level of the recursion, where it will eventually get scheduled by conditioning. 
This guarantees that even if different tasks of a special job get assigned to different intervals, and hence may be part of different recursive calls, all of them will be scheduled on the machine $\sigma(j)$.
In order to do so, we need the sub-instances to satisfy Property 2 in Definition~\ref{def:partial-scheduling}. That is, each task $a \in A(j)$ should either be completely scheduled in the sub-interval $I$ in the solution $x$ or not scheduled at all.  This is exactly what the procedure $\mysplit$ does.

Here is how the splitting procedure for a job $j \in \jspl^*$ works.
Let $A(I^*, j, \vecx) := \{1,2, \ldots, g \}$. 
Starting from the left, let us denote the intervals in $\calI^*_{k^2}$ by $I_1, I_2, \ldots, I_{2^{k^2}}$. 
Now, consider an interval $I_u$ for $u \in [2^{k^2}]$. Define $t^* := \max \{t \in I_u : (\exists a \in A(I^*, j, \vecx)) \ x_{(a, \sigma(j), t)} > 0 \}$ as  the rightmost time slot $t$ in the interval $I_u$ for which some $a \in \{1,2, \ldots, g \}$ has positive support on machine $\sigma(j)$. 
Let $a^*$ be such a task. 
Now we condition on the event  $x_{a^*, \sigma(j), t^*} = 1$. 
After conditioning, the fractional solution has the following property: every task $a' > a^*$  is scheduled completely in $I_{u+1} \cup I_{u+2} \cup \cdots \cup I_{2^{k^2}}$ and every task $a' \leq a^*$ is scheduled completely in $I_1 \cup I_2 \cup \cdots \cup I_u$. This follows from the precedence constraints in our LP and the choice of $t^*$. So, by repeating the operation for $u$ from $1$ to $2^{k^2}$, we partition the tasks in $A(I^*, \jspl^*, \vecx)$ into $2^{k^2}$ ``chunks", each of which contains a subset of consecutive tasks from $A(I^*, j, \vecx)$. We apply the procedure first for all the jobs in the original set $\jspl^*$. When a new job $j$ is added to $\jspl^*$ because of Step 1 of our algorithm, we also apply the procedure to $j$. Notice that splitting may shrink the support of other jobs $j' \in J^* \cup \jspl^*$ because of the conditioning operations; the running of step 2 is actually interleaved with the running of step 1.

For each special job we perform at most $2^{k^2}$ conditioning operations during splitting, which we show is  acceptable for our targeted running time. A crucial observation is that the total number of special jobs at any level of recursion is small, hence the number of extra conditioning operations performed by our algorithm is small.\footnote{Note that our algorithm never conditions on negative events.} 
This guarantees that the number of levels left in the LP hierarchy solution $\vecx$ satisfies the requirements of Lemma \ref{lem:critical}.
 	
\medskip		
\noindent \textbf{Step 3: Partitioning Jobs into Top, Middle and Bottom Jobs.} Let $\vecx$ be the LP-hierarchy solution after the first two steps, and $J^*$ be the current set of jobs which have the entire support in $I^*$. Note that the set $J^*$ may have reduced in size after the first two steps, as we move some jobs from $J^*$ to $\jspl^*$ in Step 1. 
Next, we select an index $\ell^* \in \{k, \ldots, k^2\}$ and partition the jobs in $J^*$ into three sets.    
	
\begin{enumerate}
		\item A set of {\em top jobs} denoted by $\jtop^*$. These are the jobs owned by the levels $0, \ldots, \ell^*-k-1 $; formally, $ \jtop^* = \bigcup^{\ell^*-k-1}_{\ell = 0} J^*_{\ell}(\vecx) $.
		\item A set of {\em middle jobs} denoted by $\jmid^*$.  These are the jobs owned by the levels $\ell^*-k,\ldots, \ell^*-1 $; formally, $ \jmid^* = \bigcup^{\ell^*-1}_{\ell = \ell^*-k} J^*_{\ell}(\vecx) $.
		\item A set of {\em bottom jobs} denoted by $\jbot^*$.  These are the jobs owned by the levels $\ell^*,\ldots, \log T^* $; formally, $ \jbot^* = \bigcup^{\log T^*}_{\ell = \ell^*} J^*_{\ell}(\vecx) $.
\end{enumerate}
	
Our algorithm completely discards the middle jobs. 
\cite{LR16} showed using a counting argument that there exists an index $\ell^* \in \{k, \ldots, k^2\}$ satisfying the following relation:
\begin{eqnarray}
\label{e:jmiddle}
|A(\jmid^*)| \leq \frac{\epsilon}{4} \cdot \frac{T^*}{\log T} + \frac{\epsilon}{2m} \cdot \left(|A(\jmid^*)| + |A(\jtop^*)| \right) \label{discard1}
\end{eqnarray}
	
In this paper we assume that such an index exists and we refer the reader to \cite{LR16} for more details. \medskip

\noindent \textbf{Step 4: Recursing on Bottom Jobs.} In this step we find a partial schedule for the bottom jobs.
	Notice that $\jbot^*$ is the union of the $2^{\ell^*}$ disjoint sets $\set{J^*(I, \vecx)}_{I \in \calI^*_{\ell^*}}$.
	For each interval $I \in \calI^*_{\ell^*}$, the algorithm finds a partial schedule of jobs in $J^*(I, \vecx) \subseteq \jbot^*$ in the interval $I$ by recursively invoking the procedure with input parameters $(I, J^*(I, \vecx), \jspl^*, \sigma, \vecx)$. It is crucial to note that every invocation of the recursive algorithm receives an {\em independent} copy of the LP-hierarchy solution $\vecx$; a programmer might say that $x$ is passed by value. In other words,  the conditioning done in recursive call $\recursive(I, J^*(I, \vecx), \jspl^*, \vecx)$ has no effect on the recursive call 
	$\recursive\allowbreak(I', J^*(I', \vecx), \jspl^*, \vecx)$ if $I, I' \in \calI^*_{\ell^*}$ are different intervals. In fact, it can be imagined as being done in {parallel}.
	
	A recursive application of Lemma \ref{lem:critical} returns a feasible schedule of some tasks in $A(J^*(I, \vecx)) \cup A(I, \jspl^*, \vecx))$ satisfying the conditions stated in the lemma. Let $A_{I,\textrm{discarded}}$ be the set of tasks discarded by our algorithm in the interval $I \in \calI^*_{\ell^*}$. Let 
	$$\calS_I: A(J^*(I, \vecx)) \cup A(I, \jspl^*, \vecx) \setminus A_{I,\text{\textrm{discarded}}} \rightarrow [m] \times I$$
	denote the schedule returned by our recursive calls for each $I \in \calI^*_{\ell^*}$.
	Let $A_{\mathrm{bottom\hyphen discarded}} = \bigcup_{I \in \calI^*_{\ell^*}} A_{I,\text{\textrm{discarded}}}$. Then a schedule $\calS$ of non-discarded tasks in bottom jobs and the tasks of special jobs which have complete support in $I$  can be obtained by combining the schedules $\calS_I$. Let 
	$$
	\calS: A(\jbot^*) \cup A(I^*,\jspl,\vecx) \setminus A_{\mathrm{bottom\hyphen discarded}} \rightarrow [m] \times I^*
	$$
	denote this combined schedule. Formally, for a task $a$ that belongs to an interval $I$, $\calS(a) := \calS_I(a)$. From our construction, $\calS$ is a valid partial schedule.

\medskip		
\noindent \textbf{Step 5: Scheduling Top Jobs.}  At this stage, it remains to assign tasks in the set $A(\jtop^*)$ in a non-migratory fashion. Recall that after the first four steps of our algorithm, we have a partial schedule of tasks belonging to the bottom jobs and the special jobs.
For convenience, define  $\hat{A}:= A(\jbot^*) \cup A(I^*,\jspl,\vecx) \setminus A_{\mathrm{bottom\hyphen discarded}}$. We want to extend the schedule $\calS$ to include $A(\jtop^*)$. We achieve this in two stages.

\begin{itemize}
	\item  In the first stage, we build a {\em tentative assignment} of tasks in  $A(\jtop^*)$ in the slots left by $\calS$. 
	During this step, we pretend that each task in  $A(\jtop^*)$ is an independent entity with no precedence constraints to any other task in the set $A(\jtop^*)$. 
	Furthermore, we {\em do not} enforce the non-migratory constraints as well. 
	However, this step guarantees that the capacity constraints and the precedence constraints between tasks in  $A(\jtop^*)$  and $\hat{A}$ are satisfied. 
	The precedence constraints are satisfied using the following strategy.
	For every top job $j$ we define an interval $[r_j, d_j]$, such that if $j$ is scheduled within this interval, then it satisfies
	all the precedence constraints to jobs in the bottom intervals. 
	The interval $[r_j,d_j]$ is defined assuming the worst possible schedule of the bottom jobs.
	Hence, it is possible that $[r_j, d_j]$ may be shorter than the interval defined by the support of LP solution for $j$.
	This implies that that some of the slots used to schedule top jobs in the LP solution may not be available for our algorithm. 
	To argue that this does not create serious concerns, we make use of the fact that our algorithm completely discarded the jobs in $\jmid^*$ (middle jobs).
	Because of this, the length of a top interval is signficantly longer than a bottom interval.
	This helps us to argue that for most jobs the interval $[r_j, d_j]$ in which $j$ needs to be scheduled is not much shorter than the interval in which LP schedules job $j$.
	In fact, an easy argument from \cite{LR16} shows  that most of the jobs in $\jtop^*$ can still be scheduled in the intervals $[r_j, d_j]$, and the number of discarded tasks is small.
	
	\item In the second stage, we convert the tentative schedule into a feasible schedule. 
	Here we run into a bigger technical hurdle: How do we schedule the top jobs such that   
	a) Every job is scheduled in the interval $[r_j, d_j]$ on a single machine; b) The precedence constraints among the top jobs are satisfied. 
	To accomplish this, we design a new algorithm, which considers the top jobs in Earliest Deadline First (EDF) order of their $d_j$ values assigned in the first stage. 
	Then, the job is assigned to the machine on which it will have the Earliest Completion Time (ECT).
	During this process we give our algorithm the flexibility of scheduling a job partially and discard
	the tasks which it cannot schedule feasibly.
	We give a careful argument to show that EDF and ECT policies together will guarantee that only a small number of tasks are discarded during this process.
	Moreover, our partial schedule also satisfies the constraints (a) and (b).
\end{itemize} 

Our algorithms for the above two steps guarantee that the final partial schedule obtained is indeed a valid partial schedule respecting all the precedence constraints. 
Let $A_{\mathrm{top\hyphen discarded}}$ denote the total number of tasks discarded from the set $A(\jtop^*)$ in the above stages.  We show the following lemma.

\begin{restatable}{lemma}{topdiscardedtsks}
	\label{lem:topdiscardedtsks}
	The valid schedule $\calS: \hat{A} \rightarrow I^*$ can be extended to a valid schedule 
	$$
	\calS^*: \left (\hat{A} \cup A(\jtop^*) \right) \setminus A_{\mathrm{top\hyphen discarded}}  \rightarrow [m] \times I^*
	$$
	such that $|A_{\mathrm{top\hyphen discarded}}| \leq \frac{\epsilon}{4} \cdot \frac{T^*}{\log T} $.
\end{restatable}
 
\subsection{Proof of the Main Lemma (Lemma~\ref{lem:critical})}
With all the components, we can finish the proof of the main lemma. 
\begin{proof}
We first check that the instance given to a recursive call of $\recursive$ satisfies the properties of Definition~\ref{def:partial-scheduling}.   By Lemma~\ref{lem:cutchains}, the total number of jobs we add to $\jspl^*$ in the loop on line~\ref{step:reduce-chain-start} is at most $K$, since we add one job to $\jspl^*$ in each iteration of the loop. Initially, we have $|\jspl^*| \leq \log \frac{T}{|I^*|} \cdot K$. Thus, after line~\ref{step:reduce-chain-end} in $\recursive$, we have that $$|\jspl^*| \leq \log \frac{T}{|I^*|} \cdot K + K \leq \log \frac{T}{|I^*|/2^{\ell^*}}\cdot K.$$ 

	Hence, the number of special jobs given as input is small. 
	Now we consider the number $r'$ of levels the LP-hierarchy solution $x$ has after line~\ref{step:reduce-chain-end}.  Initially, the number of levels $x$ had was $r = \log |I^*|(\log T)K'$.
	The total number of conditioning operations performed on line~\ref{step:initial-condition} and in the loop on line~\ref{step:reduce-chain-start} is at most
	\ifsoda
\else
		\begin{align*}
			\log \frac{T}{|I^*|}\cdot K \cdot  2^{k^2} + K \cdot 2^{k^2} = \left(\log\frac{T}{|I^*|} + 1\right) K' \leq (\log T)K'.
		\end{align*}
	\fi
	Thus, the number $r'$ of levels $x$ has after line~\ref{step:reduce-chain-end} is at least $\log |I^*|(\log T)K' - (\log T)K' \geq \log\frac{|I^*|}{2^{\ell^*}}\cdot (\log T)K'$. 
	
	To sum up, since the interval $I$ in each sub-instance has length $\frac{|I^*|}{2^{\ell^*}}$, we have that $|\jspl^*|$ is small. Further, the number of levels $x$ has is sufficiently large for each sub-instance. All the other properties in Definition~\ref{def:partial-scheduling} follow directly from the description of $\recursive$. \medskip

	Now we prove that the algorithm returns a schedule satisfying the properties stated in the lemma.
	The schedule $\calS^*$ guarantees that  every job $j \in J^*$ is assigned to a single machine. 
On the other hand, our splitting operation guarantees that for every job $j \in \jspl^*$, the set $A(I,j,\vecx)$ of tasks that have complete support in $I$ is scheduled on machine $\sigma(j)$. 
This follows from the fact that all the  tasks of special jobs are scheduled by conditioning at the lowest level of recursion.

Thus to prove the lemma, it remains to bound the number of tasks discarded by our algorithm.
\ifsoda
\else
This calculation is same as in~\cite{LR16} and we do it for the sake of completeness. In the following, the values of $J^*$ and $\jspl^*$ are considered at the time after line~\ref{step:reduce-chain-end}. 

\smallskip
Our recursive algorithm discards tasks in the following three steps.

\begin{itemize}
	\item The entire set of middle jobs is discarded. From the discussion in Step 3 (Equation \ref{e:jmiddle}) of our algorithm we know that  
\begin{eqnarray}
|A(\jmid^*)| \leq \frac{\epsilon}{4} \cdot \frac{T^*}{\log T} + \frac{\epsilon}{2m} \cdot \left(|A(\jmid^*)| + |A(\jtop^*)| \right). \label{discard1}
\end{eqnarray}

	\item For each interval $I \in \calI^*_{\ell^*}$, our algorithm recursively schedules tasks in $A(J^*(I, \vecx)) \cup A(I, \jspl^*, \vecx)$ in the interval $I$.  By recursive application of Lemma~\ref{lem:critical} we conclude that
	
$$
A_{I,\text{\textrm{discarded}}} \leq \frac{\epsilon}{2} \cdot \frac{\log (\frac{T^*}{2^{\ell^*}})}{\log T} \cdot \frac{T^*}{2^{\ell^*}} + \frac{\epsilon}{2m} \cdot |A(J^*(I, \vecx)) \cup A(I, \jspl^*, \vecx))|, 
$$
where we used the fact that $|I| = \frac{T^*}{2^{\ell^*}}$. Therefore,
\begin{eqnarray}
|A_{\mathrm{bottom\hyphen discarded}}| &=& \sum_{I \in \calI^*_{\ell^*}} |A_{I,\text{\textrm{discarded}}}| \nonumber \\
&\leq&  \sum_{I \in \calI^*_{\ell^*}} \bigg( \frac{\epsilon}{2} \cdot \frac{\log (\frac{T^*}{2^{\ell^*}})}{\log T} \cdot \frac{T^*}{2^{\ell^*}} + \frac{\epsilon}{2m} \cdot |A(J^*(I, \vecx)) \cup A(I, \jspl^*, \vecx))| \bigg)  \nonumber \\
&=& \frac{\epsilon}{2} \cdot \frac{\log T^* - \ell^*} {\log T} \cdot T^* + \frac{\epsilon}{2m} \cdot \bigg(|A(\jbot^*)| + |A(I^*,\jspl^*,x)| \bigg). \label{discard2}
\end{eqnarray}
We used the fact that for every $a \in A(I^*, \jspl^*, x)$, $a$ is completely scheduled inside some $I \in \calI^*_{\ell^*}$.

\item We discard some more tasks from the set $A(\jtop^*)$ while scheduling top jobs. By  Lemma~\ref{lem:topdiscardedtsks},
\begin{eqnarray}
	|A_{\mathrm{top\hyphen discarded}}| &\leq& \frac{\epsilon}{4} \cdot \frac{T^*}{\log T}. \label{discard3}
\end{eqnarray}
\end{itemize}

By combining \eqref{discard1}, \eqref{discard2} and \eqref{discard3}, we get 
\begin{eqnarray*}
|A_{\text{\textrm{discarded}}}| &= & |A(\jmid^*)|+ |A_{\mathrm{bottom\hyphen discarded}}| +  |A_{\mathrm{top\hyphen discarded}}| \nonumber \\
&\leq& \frac{\eps}4\cdot\frac{T^*}{\log T} + \frac{\eps}{2}\cdot \frac{\log T^* - \ell^*}{\log T} \cdot T^* + \frac{\eps}{4}\cdot \frac{T^*}{\log T}\\
& & +\quad \frac{\eps}{2m}\left(|A(J^*_\mid)| + |A(J^*_\top)| + |A(\jbot^*)| + |A(I^*, \jspl^*, x)|\right)\\
&\leq&  \frac{\epsilon}{2} \cdot \frac{\log T^*}{\log T} \cdot T^* + \frac{\epsilon}{2m} \cdot | A(J^*) \cup A(I^*,\jspl^*,\vecx) |,
\end{eqnarray*}
where we used $\ell^* \geq 1$ and the fact that the sets $\jtop^*, \jmid^*, \jbot^*$ define a partition of $J^*$ and $J^* \cap \jspl^* = \emptyset$.  We want to bound the number of discarded jobs using the original $J^*$ and $\jspl^*$, i.e, the sets specified in the input. However, it is fairly straightforward to see that the set $A(J^*) \cup A(I^*,\jspl^*,\vecx)$ does not change when moving jobs from $J^*$ to $\jspl^*$. The lemma follows by the definition of $T^*$ and $I^*$.
\fi
\end{proof}


\medskip

\textbf{Organization of the rest of Section \ref{sec:makespan}:} The rest of this section is devoted to proving the lemmas mentioned in Steps 1 and 5. In Subsection \ref{label:chain} we prove Lemma \ref{lem:cutchains}. In Subsection \ref{label:topjobs} we prove Lemma \ref{lem:topdiscardedtsks}.

\subsection{Reducing the Chain Length and Bounding Number of Conditionings}
\label{label:chain}
In this section, we give more details about the first two steps of our algorithm. 
We first show that the total number of conditionings required to reduce chain length among top jobs is small, which in turn implies small number of conditionings during splitting operations. 
We begin by proving Lemma \ref{lem:cutchains} from Step 1, which we restate below for convenience.

\cutchains*

\begin{proof}
	After line~\ref{step:reduce-chain-end} in $\recursive$, we have that for every $\ell \in \set{0, 1, \ldots, k^2-1}$ and every interval $I \in \calI^*_\ell$, the maximum chain length of jobs owned by $I$ is at most $\delta|I| = \delta \frac{T^*}{2^\ell}$.  Thus, the maximum chain length of jobs owned by level $\ell$ is at most $\delta T^*$, since there  are exactly $2^\ell$  intervals.
	   There are at most $k^2$  levels, which implies that the maximum chain length of $\hat J$, i.e, jobs owned by levels $0, 1, \ldots, k^2-1$, is at most $ k^2 \delta T^*$.

	Thus it remains to bound the number of iterations of the loop on line~\ref{step:reduce-chain-start}.	
	We focus on some interval $I \in \calI^*_{\ell}$ for some $\ell \in \set{0, 1, \ldots, k^2-1}$. Consider the chain $\calC$ of jobs owned by the interval $I$ of length (total size) at least $\delta |I|$ and the job $j \in \mathcal{C}$ that we condition on in an iteration of the loop.  Let $I_1$ and $I_2$ be the left and right children of $I$ in the laminar tree; so $I_1, I_2 \in \calI^*_{\ell+1}$.
	Notice that $j$ will be removed from $J^*$ and thus $\hat J$ during the iteration.  
	On the other hand, all the other jobs in $\mathcal{C}$ will be owned by some sub-interval of $I_2$ after the conditioning.
	This is true for the following reason. 
	Since $j$ was owned by the interval $I$, it means that for the last task $a$ of $j$ and some time $t \in I_2$ we have $x_{(a,i,t)} > 0$.
	Furthermore, $j$ was the first job in the chain $\mathcal{C}$. 
	Therefore the entire support of the jobs in $\mathcal{C} \setminus \{j\}$ will lie inside $I_2$ after the conditioning on $(a, i, t)$. Thus, every $j' \in \calC \setminus \{j\}$ is now owned by some level $\ell' \geq \ell + 1$.

	Appealing to the fact that support can only shrink upon conditioning, a single job (thus a task) can move down at most $k^2$ times before it is removed from $\hat J$.
	Due to the capacity constraints, there are at most $m T^*$ tasks owned by all intervals $I \subseteq I^*$ in total,
	so there will be at most $m T^* k^2$ events that a task moves down.
	On the other hand, in one iteration of the loop on line~\ref{step:reduce-chain-start}, each task in $\calC$, of which there are at least $\delta |I|$ many, is either removed from the set $A(\hat{J})$ or moved down by at least one level.  
	Note that  $\delta |I| \geq \frac{\delta T^*}{2^{k^2}}$. Together we get that the number of iterations of the loop is at most
	$$
		\frac{mT^* k^2}{ \delta \frac{T^*}{2^{k^2}}} = \frac{mk^2 \cdot 2^{k^2}}{\delta} = K.\hfill \ifsoda\else\qedhere\fi
	$$ 
\end{proof}

\subsection{Scheduling the Top Jobs}
\label{label:topjobs}
In this section we give an algorithm to schedule the top jobs $\jtop^*$.  
Recall that after the first four steps of our algorithm, we have a partial schedule of tasks belonging to the bottom jobs and the special jobs. Formally,
\newcommand{\thefunctionsignature}{\calS: (A(\jbot^*) \cup A(I^*,\jspl,\vecx)) \setminus A_{\mathrm{bottom\hyphen discarded}} \rightarrow [m] \times I^*}
\ifsoda
\else
\[ \thefunctionsignature \,. \]
\fi

 Recall that for convenience we have defined  $\hat{A}:= (A(\jbot^*) \cup A(I^*,\jspl,\vecx)) \setminus A_{\mathrm{bottom\hyphen discarded}}$. 
 We want to extend $\calS$ to a schedule $\calS^*$ that includes most tasks of the set $A(\jtop^*)$. 
 
 \smallskip
 As described in Step 5 of the algorithm, we accomplish this in two stages.
 In the first stage, we build a {\em tentative assignment} of tasks in  $A(\jtop^*)$ in the slots left by $\calS$. 
 At this stage, we {\em do not} satisfy the non-migratory constraints but we guarantee that the capacity constraints and the precedence constraints between tasks in  $A(\jtop^*)$  and $\hat{A}$ are satisfied. 
 During this step we discard some tasks from the set $A(\jtop^*)$, which we denote by $A^1_{\mathrm{top\hyphen discarded}}$.
 In the second stage, we convert the tentative schedule into a valid partial schedule respecting all the precedence constraints.  
 This step also involves discarding some tasks from the set $A(\jtop^*)$, which we denote by $A^2_{\mathrm{top\hyphen discarded}}$. 

 Define $A_{\mathrm{top\hyphen discarded}} :=  A^1_{\mathrm{top\hyphen discarded}} + A^2_{\mathrm{top\hyphen discarded}}$. 
 The goal of this section is to prove Lemma \ref{lem:topdiscardedtsks}, which we re-state for convenience.

\topdiscardedtsks*

\subsubsection{Tentative Schedule}  
\label{tentativeschedule}
As we allow migration of jobs in this step, it gives us the flexibility to treat each task in $A(\jtop^*)$ almost independently. 
This lets us use the algorithm in \cite{LR16} as a black box to build the tentative schedule. 
We briefly sketch proofs of the claims made in this subsection, and refer the reader to Section 5 in \cite{LR16} for full details.  

Recall that $T^*$ denotes the length of interval $I^*$. 
For convenience we re-index the time slots in $I^*$ so that $I^* := \{1,2,\ldots, T^*\}$. 
Let $\calI^*$ denote the binary laminar family of intervals, and let $I_1, I_2, \ldots, I_p$ be the set of bottom intervals in $\calI^*_{\ell^*}$. 
From our construction $p = 2^{\ell^*}$. 
For each interval $I_{v}, v \in \{1,2,\ldots, p\}$, let $\text{begin}(I_v)$ and $\text{end}(I_v)$ denote the first and  the last time slots that belong to $I_v$.

Our recursive algorithm schedules most of the tasks in $A(\jbot^*) \cup A(I, \jspl^*, x)$ in the intervals $I_1, I_2, \ldots, I_p$.  
Recall that our algorithm recurses for each interval $I \in \calI^*_{\ell^*}$ on the subset of jobs $J^*(I, \vecx)$ and $\jspl^*$. 
By definition, every job $j \in J^*(I, \vecx)$ has its entire support in the interval $I$. 
Furthermore, the tasks of jobs in the set $\jspl^*$ that are scheduled in $I$ are precisely those tasks which have entire support in $I$. 
For each $I \in \calI^*_{\ell^*}$, our algorithm schedules a subset of tasks in $A(J^*(I, \vecx)) \cup A(I, \jspl^*, \vecx)$. 
This implies that $|A(\jtop^*)| \leq m|I^*| - |\hat{A}|$. 
In other words, there are enough slots to schedule tasks in $A(\jtop^*)$ in the slots left open by $\calS$.

In fact, we can say something stronger. 
Let $[r_j, d_j]$ denote the {\em smallest} interval in which the entire support of any job $j \in \jtop^*$ lies in the solution $\vecx$.
That is, $\sum_{i} \sum_{t \in [r_j, d_j]} x_{(i,a,t)} = 1$ for  every task $a \in A(j)$. 
Let $u_j$ be the largest index such that $\text{begin}(I_{u_j}) \leq r_j$.
Similarly let $v_j$ be the smallest index such that $\text{end}(I_{v_j}) \geq  d_j$.
In other words, the interval $[\text{begin}(I_{u_j}), \text{end}(I_{v_j})]$ is the smallest interval that completely contains $[r_j, d_j]$ and whose beginning and ending coincide with a beginning and ending of intervals in $I_1, I_2, \ldots, I_p$.
Let $[r'_j, d'_j] := [\text{begin}(I_{u_j}), \text{end}(I_{v_j})]$.
The current LP solution $\vecx$ also guarantees that there is enough space to {\em fractionally} schedule tasks in $A(\jtop^*)$ so that for every job $j \in \jtop^*$, all the tasks in $A(j)$ are scheduled in the interval $[r'_j, d'_j]$. 
This can be proved by setting up a fractional bipartite matching between the jobs and the empty slots in the intervals $I_1, I_2, \ldots, I_p$.  
The existence of a fractional matching in a bipartite graph implies that there exists an integral matching. 
Hence, there is an assignment of tasks of $A(\jtop^*)$ such that every job $j$ is  scheduled in the interval $[r'_j, d'_j]$.

Unfortunately, the schedule obtained using the above argument can violate the precedence constraints between tasks in $A(\jtop^*)$ and $\hat{A}$. 
The main reason is that the schedule of tasks in $\hat{A}$ had been constructed without taking $A(\jtop^*)$ into account. 
We will define a {\em release time} $r^*_j$ and a {\em deadline} $d^*_j$ for all jobs $j \in \jtop^*$ in such a way that if all the tasks of $j$ are scheduled  in the interval  $[r^*_j, d^*_j]$, then  the precedence  constraints between $A(\jtop^*)$ and $\hat{A}$ will be satisfied.  
A natural way to define the interval $[r^*_j, d^*_j]$ is as follows:
\begin{equation}
\label{e:release-dead}
r^*_j := \text{begin}(I_{u_{j} + 1}) \quad \text{and} \quad d^*_j := \text{end}(I_{v_{j} - 1}).
\end{equation}

We claim that if our algorithm schedules all tasks of the job $j$ in the interval $[r^*_j, d^*_j]$, then the precedence  constraints between tasks in $A(\jtop^*)$ and $\hat{A}$ will be satisfied.
For this, note that every task in $\hat{A}$ is supported on a single interval in $\calI^*_{\ell^*}$ in $x$; for tasks of bottom jobs this follows from the definition, and for tasks of special jobs it is a consequence of our splitting procedure (indeed every task of a special job is supported on a single interval in $\calI^*_{k^2}$, and $\ell^* \le k^2$).
Consider a task $a \in \hat{A}$ and a job $j \in \jtop^*$ with $a \prec j$.
Then $a$ must be supported on an interval $I_u$ with $u \le u_j$.
(Suppose otherwise that $u > u_j$, and let $a'$ be the first task of $j$; then $a'$ is scheduled with positive $x$-fraction on the interval $I_{u_j}$ by the definition of $r_j$, yet the entire support of $a$ in $x$ lies to the right of $I_{u_j}$, contrary to $a \prec a'$ and the constraints (\ref{e:precconstraints}) of the LP.)
Thus, as long as $j$ is scheduled after $r^*_j = \text{begin}(I_{u_{j} + 1})$, the constraint $a \prec j$ will be satisfied.
The argument for precedence constraints of the form $j \prec a$ is symmetric.

Now observe that $[r^*_j, d^*_j] \subsetneq [r'_j, d'_j]$: the interval in which we are allowed to schedule the job $j$ is shorter than the interval in which the LP solution schedules the tasks.
Therefore, it is no longer clear that we will have enough slots to schedule all tasks in $A(\jtop^*)$ and our algorithm may have to discard some tasks. 
An argument in~\cite{LR16}  shows that the number of tasks discarded is not large.

\begin{lemma}[\cite{LR16}]
\label{lem:tentative}
	A feasible partial schedule $\calS: \hat{A} \rightarrow [m] \times I^*$ can be extended to a new schedule $\calS':  \left (\hat{A} \cup A(\jtop^*) \right) \setminus A^1_{\mathrm{top\hyphen discarded}} \rightarrow [m] \times  I^*$ satisfying the following properties: 
	\begin{enumerate}
		\item Consider $j \in \jtop^*$ and let $a \in A(j)$. Then in the schedule $\calS'$,  either $a$ is assigned in the interval $[r^*_j, d^*_j]$ or $ a \in A^1_{\mathrm{top\hyphen discarded}}$.
			\item The total number of discarded tasks $| A^1_{\mathrm{top\hyphen discarded}}| \leq 4m2^{-k}T^*$. 
		\item The precedence constraints between tasks in $A(\jtop^*) \setminus A^1_{\mathrm{top\hyphen discarded}}$ and $\hat{A}$ are respected.
		\item The capacity constraints are satisfied.
	\end{enumerate}
	\end{lemma}

The proof of this lemma in \cite{LR16} crucially relies on the following observation. Notice that although the interval in which we can schedule the job $j$ shrinks, the loss is small. More precisely, the LP solution $\vecx$ would have scheduled some portion of the job in the intervals $I_{u_j}$ and $I_{v_j}$, and now we will not be able to use them.
However,  the intervals in $\calI^*_{\ell^*}$ are much shorter than the intervals that own the top jobs.
This is because we discarded  a set of $k$ consecutive middle intervals.
Hence, each interval $I \in \calI^*_{\ell^*}$ is at least $2^{-k}$ times shorter than an interval $I' \in \calI_0 \cup \cdots \cup \calI_{\ell^*-k-1}$.
This fact can be used to show that we only need to throw away few tasks to account for the loss of flexibility in scheduling top jobs.

\subsubsection{From Tentative Schedule to Final Schedule}
\label{deadlinproblem}

In this final step of our algorithm, we convert the schedule $\calS':  \left (\hat{A} \cup A(\jtop^*) \right) \setminus A^1_{\mathrm{top\hyphen discarded}} \rightarrow [m] \times I^*$ into a non-migratory schedule that satisfies all the precedence constraints among top jobs.  
The main goal of this subsection is to prove Lemma~\ref{lem:topdiscardedtsks}.  
Our final schedule $\calS^*$ satisfying the requirements of the lemma  guarantees the following two invariants:
\begin{itemize}
\item For every job $j \in \jtop^*$, all the non-discarded tasks are scheduled within $[r^*_j, d^*_j]$.
\item The assignment of tasks in $\hat{A}$ remains the same as that in the schedule $\calS'$.
\end{itemize}

The above two invariants together will imply that all the precedence constraints are satisfied among non-discarded tasks in  $\calS^*$.
We build schedule $\calS^*$  by first designing an algorithm for a new stand-alone scheduling problem, then using it as a black box for scheduling top jobs.  
We first define the new scheduling problem.  

\medskip
\textbf{A Deadline Scheduling Problem with Precedence Constraints:}  Consider a set of jobs $J$ with processing lengths $p_j$, release times $r_j$, deadlines $d_j$ and precedence constraints. 
As before, we assume that each job $j \in J$ is made up of $p_j$ unit-length tasks. 
The precedence constraints among jobs satisfy the property that  if $j \prec j'$, then $r_j \leq r_{j'}$ and  $d_j \leq d_{j'}$. 
Recall that we use $\Delta(J)$ to denote the maximum chain length in $J$. 
The time horizon $[T]$ is partitioned into $p$ equal sized intervals $I_1, I_2, \ldots, I_p$. 
The release times and deadlines of jobs fall at beginnings and ends of the intervals. 

For each machine $i$, we are given a capacity function $\capa_i:[T] \rightarrow \{0, 1\}$. 
If $\capa_i(t) = 1$, then the time slot $t$ on machine $i$ is available to schedule a task in $A(J)$.  
For an interval $I =  \{t', \ldots, t''\}$, we overload the notation to refer $\capa_i(I)$ as the number of available slots in the interval $I$; that is, $\capa_i(I) = \sum^{t''}_{t = t'}\capa_i(t)$.  
Similarly, we use $\capa:[T] \rightarrow [m]$ to denote the number of available slots at time $t$ across all the $m$ machines; $\capa(t) = \sum^{m}_{i = 1}\capa_i(t)$. 

Suppose there is a schedule of tasks $\calS': A(J) \rightarrow [m] \times [T]$ that assigns each task in $A(J)$ to a machine-timeslot pair such that no two tasks are assigned to the same machine and the same time slot; that is, the capacity constraints on machines are satisfied. 
Moreover, $\calS'$ ensures that for each job $j \in J$,  all the tasks in $A(j)$ are scheduled within $[r_j, d_j]$. We remark that $\calS'$ may not respect the precedence constraints in $J$ and the no-migration constraints. 
Our goal is to schedule each $j \in J$ on a single $i$ so that precedence constraints among the jobs are satisfied and enforce the no-migration constraints.

\begin{algorithm}[h!]
		\caption{\textsf{EDF+ECT}}
		\label{alg:edf+ect}
		\textbf{Input:} A set of jobs $J$ with release times, deadlines and precedence constraints; capacity function $\capa_i: [T] \rightarrow \{0, 1 \}$ for each machine $i$.\\
		\textbf{Output:} Schedule $\calS: A(J) \rightarrow [m] \times [T]$  such that for $a \sim a'$, either $a$ or $a'$ belongs to $\taskdis$ or $\calS_\mac(a) = \calS_\mac(a')$.
		\vspace*{-5pt}
		
		\noindent\rule{\linewidth}{0.5pt}
		
	\begin{algorithmic}[1]
		\State Sort the jobs in $J$ in the increasing order of their deadlines. Reindex the jobs so that $J:= \{1,2, \ldots, n\}$ and $d_1 \leq d_2 \leq  \ldots \leq d_n$. 
		\State Initialize $\taskdis = \emptyset$.
	
	\For{ $j = 1$ to $n$}
		\State Find the earliest time slot $t \in [T]$ such that following conditions hold:
			i) $\capa_i(t) = 1$ for some machine $i \in [m]$ and  $r_j \leq t \leq d_j$;
			ii) $C_{j'} < t$ for all $j' \prec j$.   \label{invariant3}
			
		\State  If no such $t$ exists, then set $B_j := \drop$ and add all the tasks $A(j)$ to the set $\taskdis$.
		\State  If there is a $t$ satisfying the conditions above, set $B_j = t$ and do the following.
		\State  Find the set of machines $M' = \{ i \in [m] : \capa_i(B_j, d_j) \geq p_j \}$.
	
		\If {$|M'| = 0$}
			\State $i^* = \displaystyle \text{argmax}_{i} \{ \capa_i(B_j, d_j) \}.$ \label{invariant1}
			\State Schedule $\capa_{i^*}(B_j, d_j)$ tasks of the job $j$ in the interval $[B_j, d_j]$ on the machine $i^*$.
			\State Set $C_j = d_j$. Add the remaining $p_j - \capa_{i^*}(B_j, d_j)$ tasks to the set $\taskdis$. 
			\State Update the capacity function $\capa_{i^*}$ for the machine $i^*$ and the schedule $\calS$.
		\EndIf
			
		\If {$|M'| \geq 1$}
			\State Find the earliest time slot $t^*$ such that there exists a machine $i^* \in M'$  and $\capa_{i}(B_j, t^*) = p_j$. \label{invariant2}
			\State Set $C_j = t^*$. Schedule all the tasks $A(j)$ of $j$ in the interval $[B_j, C_j]$ on the machine $i^*$.
			\State Update the capacity function $\capa_{i^*}$ for the machine $i^*$ and the schedule $\calS$.
		\EndIf
	\EndFor
	\State \textbf{return} $\calS$.
	\end{algorithmic}
\end{algorithm}

We prove the following theorem in this section.

\begin{theorem}
	\label{thm:edf}
	There exists an algorithm that in polynomial time  converts the schedule $\calS'$ into a valid schedule that satisfies the following properties:
	\begin{enumerate}
	\item
	It partially schedules every job on exactly one machine (no-migration). 
	For a job that is partially scheduled, we discard the remaining tasks;  
	for the sake of the precedence constraints, we assume that every partially scheduled job or a fully discarded job is completely processed. 
	\item The precedence constraints among the jobs are satisfied. 
	\item The total number of discarded tasks is at most $2p^2m\Delta(J)$.
	\end{enumerate}
\end{theorem}

We prove the above theorem using the following scheduling algorithm, which considers jobs in the Earliest Deadline First (EDF) order and assigns  jobs to machines on which they will have {\em the earliest completion time} (ECT). Our algorithm is given in Algorithm 3, which we call $\textsf{EDF+ECT}$.
See Figure \ref{fig:topbjobs} for an illustration.

For a job $j$, we call the interval $[B_j, C_j]$ the {\em active} interval. 
$B_j$ denotes the first time instant when the job is {\em ready} to be scheduled {\em and there is an available slot} on some machine $i$. 
Note that it can be the case that no task of job $j$ is scheduled at the time slot $B_j$, because we choose the machine on which the job will have the earliest completion time.
If for a job~$j$, $B_j = \drop$ then we say that the job is {\em fully discarded}; otherwise, if not all of its tasks are scheduled, we say that it is {\em partially discarded}. 
For a fully discarded job, all its tasks belong to the set $\taskdis$. 
We say that a job $j$ is either {\em partially or fully discarded in the interval $I \in \{I_1, I_2, \ldots, I_p \}$} if $d_j \in I$ and $j$ is either partially or fully discarded.

We call a time slot $t$ on machine $i$ {\em idle} if $\capa_i(t) = 1$ and our schedule $\calS$ does not assign any task at time $t$ on machine $i$. 
Roughly speaking, the idle time slots correspond to the number of tasks we discard, and our goal going forward is to show that there are not too many idle time slots in $\calS$. 
We make the following observations that will help in proving Theorem \ref{thm:edf}.

\begin{observation}
\label{o:activejob}
	Fix an interval $I_q$ for some $q \in [p]$. Suppose the following two conditions hold:
	\begin{itemize}
	\item  There is a time slot $t^* \in I_q$ that is idle on machine $i$ in $\calS$.
	\item  There exists a job $j^*$ with $t^* \in [r_{j^*}, d_{j^*}]$ and either $B_{j^*} > t^*$ or $B_{j^*} = \drop$.
	\end{itemize}
	 Then there exists a job $j$ such that $t^* \in [B_{j}, C_{j}]$ and $j \prec j^*$.
	 Moreover, $j$ is not scheduled on machine $i$.
\end{observation}

\begin{proof}
If no such job $j$ exists, then, when our algorithm considers the job $j^*$ for scheduling, it would set $B_{j^*} = t^*$. This follows from line \ref{invariant3} of our algorithm.
Moreover, if $j$ were scheduled on $i$, then it would have been scheduled on every time slot between $B_j$ and $C_j$ that is idle on $i$, which contradicts the facts that $t^* \in [B_{j}, C_{j}]$ and that $t^*$ is idle on $i$.
\end{proof} 

\begin{observation}
\label{o:conservation}
	Consider a job $j \in J$ with $B_j \neq \drop$ that is active in the interval $[B_j, C_j]$ on machine $i$. Let $p'_j$ be the total number of tasks of $j$ scheduled in $[B_j, C_j]$. Then for any other machine $i' \neq i$, the number of idle time slots in the interval $[B_j, C_j]$ is at most $p'_j$.
\end{observation}

\begin{proof}
Consider the case when $p'_j = p_j$. In this case, the lemma follows since our algorithm assigns jobs to machines on which they will have the earliest completion time.
Now consider the case when $p'_j \neq p_j$. 
As job $j$ was scheduled on machine $i$, from line~\ref{invariant1}, it follows that $i$ had the maximum number of empty slots in the interval $[B_j, d_j]$, which is equal to $p'_j$.
Since in this case we set $C_j = d_j$, no machine $i' \neq i$ had more than $p'_j$ empty slots.
\end{proof}

\ifsoda
	\begin{figure*}[h]
\else
	\begin{figure}[t]
\fi
\centering
\ifsoda
\else
	\includegraphics[width=1.0\textwidth]{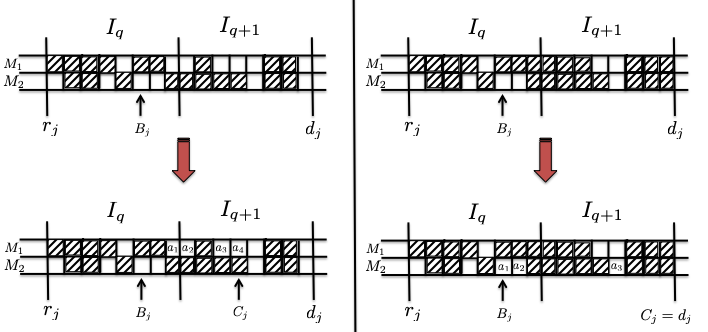}
\fi
 \caption{Figure illustrates scheduling a job $j$ with $p_j = 4$ in the non-migratory fashion when $m = 2$. On the left, $j$ is fully scheduled in the interval $[B_j, C_j]$ on $M_1$. On the right, the first three tasks are scheduled on machine $M_2$, and the last task $a_4$ is discarded.} \label{fig:topbjobs}
\ifsoda
	\end{figure*}
\else
	\end{figure}
\fi

The above two observations will help us argue that the number of idle time slots on any machine is small.

\begin{lemma}
\label{l:emptyslots}
	Consider any arbitrary time  interval  $I := \{t', \ldots, t''\} \subseteq I_q$ for some $q \in [p]$. Suppose there is at least one job $j^*$ with $I\subseteq [r_{j^*}, d_{j^*}]$ and either $B_{j^*} > t''$ or $B_j = \drop$. Then, for any machine $i \in [m]$ the number of idle time slots in $I$ is at most $\Delta(J)$.
\end{lemma}

\begin{proof}
	Consider a machine $i$; we will show that it has at most $\Delta(J)$ idle time slots in the interval $I$. 
	Let $t^* \in I$ be the last such time slot. 
	Since $j^*$ is available at $t^*$, by Observation \ref{o:activejob} there must exist a  job $j_1 \prec j^*$ that is active at time $t^*$, i.e.~$t^* \in [B_{j_1}, C_{j_1}]$. 
	Now consider the last time slot $t_1 < B_{j_1}$ that is empty on machine $i$. 
	We claim that $j_1$ is available for processing at time $t_1$.
	This follows from our assumptions that if $j_1 \prec j^*$, then $r_{j_1} \leq r_{j^*}$ and the release times and the deadlines of $J$ align with the beginnings and the endings of the intervals.
	Therefore, there must be a job $j_2 \prec j_1$ such that $t_1 \in [B_{j_2}, C_{j_2}]$. 
	Moreover, $C_{j_2} < B_{j_1}$ as  $j_2 \prec j_1$. 
	We continue by induction to construct a chain of jobs $j_y \prec j_{y-1} \prec  \ldots \prec j_1$ such that $[B_{j_y}, C_{j_y}] \cup [B_{j_{y-1}}, C_{j_{y-1}}] \cup \ldots \cup [B_{j_1}, C_{j_1}]$ covers all the empty slots on machine $i$ in the interval $I$. These intervals are pairwise disjoint.
	The total size of the jobs in the chain $j_y \prec j_{y-1} \prec  \ldots \prec j_1$, $\sum^{y}_{v = 1}p_{j_v}$, is at most $\Delta(J)$, since the maximum chain length among $J$ is at most $\Delta(J)$.
By Observation~\ref{o:activejob}, jobs $j_1, \ldots, j_y$ are not scheduled on $i$.
Thus, by Observation \ref{o:conservation},  for any job $j_v$ belonging to the chain, there can be at most $p_{j_v}$ empty slots on the machine $i$ in the interval $[B_{j_v}, C_{j_v}]$. 
	Therefore, the number of all time slots idle on $i$ in the union of these intervals is at most $\sum^{y}_{v = 1}p_{j_v} \leq \Delta(J)$. 
	We conclude by recalling that $[B_{j_y}, C_{j_y}] \cup [B_{j_{y-1}}, C_{j_{y-1}}] \cup \ldots \cup [B_{j_1}, C_{j_1}]$ covers all the idle slots on machine $i$.
\end{proof}

The above lemma implies the following useful corollary.

\begin{corollary}
\label{c:reverseempty}
Suppose there is an interval $I_q$ and a machine $i$ with more than $\Delta(J)$ idle time slots. If there is a job $j^*$ such that $B_{j^*} \in I_{q+1} \cup I_{q+2} \cup \ldots \cup I_p$, then its release time $r_{j^*} \in I_{q+1} \cup I_{q+2} \cup \ldots \cup I_p$.
\end{corollary}
\begin{proof}
For contradiction, let us assume that the release time of $j^*$ belongs to $I_{q'}$, where $q' < q+1$.
Recall that all jobs are released at the beginning of the intervals.
Now, we invoke the previous lemma on the interval $I_q$ and the job $j^*$, which gives us the contradiction.
\end{proof}

The next lemma shows that even if a job is partially discarded, the number of idle time slots on a machine cannot be too large.
This lemma accounts for the number of idle slots due to the no-migration constraints, that is, due to the fragmentation of jobs.

\begin{lemma}
\label{l:totalemptyslots}
	Consider any arbitrary time  interval  $I := \{t', \ldots, t''\} \subseteq I_q$ for some $q \in [p]$. Suppose there is at least one job $j^*$ with $I\subseteq [r_{j^*}, d_{j^*}]$ and $A(j^*) \cap \taskdis \ne \emptyset$.
	Then, for any machine $i$ the number of idle time slots in $I$ is at most $2\Delta(J)$.
\end{lemma}

\begin{proof}
If $B_{j^*} > t''$, then the statement follows from Lemma \ref{l:emptyslots}. 
Therefore, $B_{j^*} \leq t''$ and some tasks of $j^*$ got discarded. 
By Observation \ref{o:conservation}, in the interval $[B_{j^*}, C_{j^*}]$ there cannot be more than $p_{j^*} \leq \Delta(J)$ idle slots on any machine $i$ where $j^*$ was not scheduled (and there cannot be any such slots on the machine where $j^*$ was scheduled since it was partially discarded).
\textit{A~fortiori}, there are at most $\Delta(J)$ such slots in the interval $[B_{j^*}, t''] \subseteq [B_{j^*}, d_{j^*}] = [B_{j^*}, C_{j^*}]$
(we have $d_{j^*} = C_{j^*}$ since $j^*$ was partially discarded).
If $B_{j^*} \le t'$, then we are done since $I \subseteq [B_{j^*}, t'']$.
Otherwise, i.e.~if $B_{j^*} > t'$,
we apply Lemma \ref{l:emptyslots} to the interval $[t', B_{j^*}-1]$ and the job $j^*$, obtaining that in the interval  $[t', B_{j^*}-1]$ there can be at most $\Delta(J)$ idle time slots. 
We conclude by combining the two intervals $[t', B_{j^*}-1]$ and $[B_{j^*}, t'']$.
\end{proof}

With all the above lemmas, it is easy to prove Theorem \ref{thm:edf}. 
For brevity, we use $\calS^{-1}(I)$ to denote the set of tasks scheduled in the interval $I$ in $\calS$.

\begin{proof}[Proof of Theorem \ref{thm:edf}]

Our algorithm guarantees that in the schedule $\calS$ all the tasks $A(j)$ of a job $j \in J$ are assigned to a single machine and the precedence constraints among the jobs are satisfied. 
Therefore, it only remains to show that the number of discarded tasks $|\taskdis| \leq 2 p^2m\Delta(J)$. 
The proof is based a double-counting argument, which was also used in \cite{LR16}.
Let $I_h$, $h \in [p]$ be any interval, and let $\alpha$ be the total number of tasks discarded in $I_h$. 
We argue that $\alpha \le 2pm\Delta(J)$. As there are $p$ intervals, this will complete the proof.

Let $j_s$ be the lowest-priority (i.e.~the highest-index) job that got discarded either completely or partially in the interval $I_h$ (if no such job exists, then we are done).
Consider the set of jobs $J':= \{ j_1, j_2, \ldots, j_s\}$ and the intervals $I_1, I_2, \ldots, I_h$.  
Imagine the following thought experiment: Let us reschedule the jobs in the set $J'$ using our algorithm in the intervals $I_1, I_2, \ldots, I_h$. 
Let  $\calS^*$ be this new schedule.  
As our algorithm considers the jobs in EDF order, the schedule $\calS^*$ will be identical to $\calS$ for the set $J'$. 
In particular, this means that the interval $I_h$ will have $\alpha$ discarded tasks in $\calS^*$. 
This is true as we never discard a job until we hit its deadline, and by our choice $j_s$ is the highest-indexed job that got discarded. 
So all the tasks discarded in the interval $I_h$ belong  to  the jobs in $J'$ in both schedules $\calS$ and $\calS^*$. 

From now on we will focus on the schedule $\calS^*$ and the set of jobs $J'$. 
By Lemma \ref{l:totalemptyslots}, the interval $I_h$ cannot have more than $2\Delta(J)$ idle time slots on any machine $i$ as there is a job $j_s$ that got discarded in $I_h$.  
Let $g \in \{1, 2, \ldots, h\}$ be the minimal index such that no machine $i$ has more than $2\Delta(J)$ empty slots in any of the intervals $I_g, \ldots, I_h$.
Let $I' := I_g \cup \ldots \cup I_h$. 
Consider the following set of jobs:
\ifsoda
\else
	\[ J'' :=  \left \{ j \in  \{j_1, j_2, \ldots, j_s \}:  [B_j, C_j] \subseteq  I' \hspace{2mm}  \text{or} \hspace{2mm}  { B_j = \drop}  \right \} . \]
\fi

For every job $j \in J''$ we claim that $[r_j, d_j] \subseteq I'$.
To get this, it is enough to show that $r_j \ge \text{begin}(I_g)$.
Then we have $\text{begin}(I_g) \le r_j \le d_j \le d_{j_s} = \text{end}(I_h)$.
Our claim is clearly true if $g=1$. Otherwise we argue as follows. 
From our choice of the index $g$,  $I_{g-1}$ has more than $\Delta(J)$ (in fact, $2\Delta(J)$) many idle time slots on some machine $i$. Now consider a job $j \in J''$. The first case is that $B_j \in I'$.
Then from Corollary \ref{c:reverseempty} (applied to $I_{g-1}$) we have $r_j \in I'$.
The second case is that $B_j = \drop$. Towards a contradiction suppose that $r_j \leq \text{begin}(I_{g-1})$.
By applying Lemma \ref{l:emptyslots} to $I_{g-1}$, we get that $I_{g-1}$ has at most $\Delta(J)$ idle slots on any machine, which contradicts our choice of the index $g$.
Therefore again $r_j \in I'$.

Conversely, any job $j \in J'$ that has a task scheduled in $I'$
must be in $J''$.
Suppose otherwise;
then we would have $g>1$,
$B_j < \text{begin}(I_g) \le C_j$,
and $r_j \le \text{begin}(I_{g-1})$.
This, together with the fact that $I_{g-1}$
has more than $2 \Delta(J)$ idle slots on some machine,
implies
that $j$ should have been scheduled fully in $I_1 \cup \cdots \cup I_{g-1}$
(via an argument similar to the proof of Lemma~\ref{l:emptyslots}),
which is a contradiction.

Now we can use a simple double-counting trick to bound $\alpha$. 
Let $\taskdis(J'') \subseteq A(J'')$ denote the subset of tasks discarded by our schedule $\calS^*$.
Then, 
\ifsoda
\else
	\begin{equation}
	\label{e:bound}
	|A(J'')|  \geq  |(\calS^*)^{-1}(I')| + |\taskdis(J'')|  \geq \capa(I') - 2pm\Delta(J) + \alpha.
	\end{equation}
\fi

The last inequality above follows from the fact that there are at most $p$ intervals in $\{I_g, \ldots, I_h\}$ and, by our definition of $g$, any interval $I_g, \ldots, I_h$ has at most $2\Delta(J)$ empty slots on any machine. 
But note that every task in $A(J'')$ needs to be scheduled in $I'$, because for every job $j \in J''$, $[r_j, d_j] \subseteq I'$. 
By our assumption, there is a (possibly invalid) schedule $\calS'$ that assigns every task in $A(J'')$  in the interval $I'$. 
Therefore,  $A(J'') \leq \capa(I')$.  
Combining this with \eqref{e:bound} we have $\alpha \leq 2pm \Delta(J)$. 
This completes the proof.
\end{proof}

\medskip

	Now we are ready to prove Lemma \ref{lem:topdiscardedtsks}, which bounds the total number of tasks discarded in the process of scheduling top jobs.

\begin{proof} [Proof of Lemma \ref{lem:topdiscardedtsks}]
We obtain the schedule $\calS^*$ by applying Theorem \ref{thm:edf} on the jobs in the set $\jtop^*$. 
For every job $j \in \jtop^*$, we define the truncated job length $p'_j$ by taking into account the discarded tasks in $A^1_{\mathrm{top\hyphen discarded}}$. 
We define the release time  of  $j$ as $r^*_j$ and the deadline as $d^*_j$, where $r^*_j, d^*_j$ are defined in Equation (\ref{e:release-dead}).  
For each machine $i$,  $\capa_i(t) = 1$ if time $t$ on machine~$i$ is not assigned any task in the schedule $\calS$. Recall that $\calS$ gives an assignment of a subset of tasks of the bottom jobs and the special jobs ($\hat{A}$) in the interval $I^*$. 

By Lemma \ref{lem:cutchains}, the maximum chain length of jobs in $\jtop^*$ is at most $k^2 \delta T^*$, where $T^*$ is the length of interval $I^*$. 
Therefore, by Theorem \ref{thm:edf}, the total number of tasks discarded in converting the tentative schedule into an actual schedule is given by
$$
|A^2_{\mathrm{top\hyphen discarded}}| \leq 2p^2 m \cdot k^2 \delta T^*.
$$

By Lemma \ref{lem:tentative}, we have 

$$
|A^1_{\mathrm{top\hyphen discarded}}| \leq 4m 2^{-k} T^*.
$$

Therefore, 
\begin{eqnarray*}
|A_{\mathrm{top\hyphen discarded}}| &=& |A^1_{\mathrm{top\hyphen discarded}}| + |A^2_{\mathrm{top\hyphen discarded}}|  \\
&\leq& 2p^2 m \cdot k^2 \delta T^* + 4m 2^{-k} T^*   \\
&\leq& \frac{\epsilon}{4} \cdot \frac{T^*}{\log T} 
\end{eqnarray*}

The last inequality follows from substituting  $p = 2^{\ell^*} \leq 2^{k^2}$, $k = \frac{O(1) m}{\epsilon} \log \log T$, and $\delta = \frac{\epsilon}{8k^2m2^{2k^2} \log T}$.
\end{proof}

\ifsoda\else
	
\section{Minimizing Makespan Under Precedence Constraints with Communication Delays}
\label{sec:communication}

Now we consider the problem of minimizing the makespan when jobs have precedence constraints and communication delay constraints.
For the rest of this section, we sometimes refer to the problem we studied in the first part of the paper , $Pm|\text{prec}, \text{pmtn} | C_{\max}$,  as ``the no-delay problem" or "the first problem" and  we refer the problem, $Pm|\text{prec}, \text{pmtn},  c_{j,j'}|C_{\max}$, as ``the delay problem".
We follow the notation developed for the no-delay problem, and whenever necessary remind the readers their meaning. 
If we use a notation without definition in this section, it implies that its meaning is same as in the no-delay problem.

Similar to the no-delay problem, the input to the problem consists of a set of jobs $J$, where each job $j \in J$ has a processing length $p_j$, the jobs have precedence constraints, and we are given a set of $m$ machines. 
If there is a precedence constraint $j \prec j'$, then we require that job $j'$  can only start after job $j$ is completed. 
Furthermore, if $j$ and $j'$ are processed {\em on two different machines}, then the processing of job $j'$ cannot start earlier than $c_{j,j'} > 0$ time units after the completion time of $j$. 
We assume without loss of generality that $c_{j, j'}$ are natural numbers.
If, however, $j$ and $j'$ are processed on the same machine, then $j'$ can start right after the completion of job $j$.
We study the case when $\max_{j \prec j'} \{c_{j,j'}\} = O(1)$ case, which is a generalization of the well studied case of when communication delays are all equal to 1.  
However, to keep the notation simple, we first present our proof assuming $c = 1$.
At the end, when it will be clear that our framework is general enough to handle $\max_{j \prec j'} \{c_{j,j'}\} = O(1)$, and we give a sketch of our proof for  $\max_{j \prec j'} \{c_{j,j'}\} = O(1)$.
The goal is to schedule jobs satisfying the precedence and communication delay constraints so as to minimize makepsan.
In the three field notation, the problem is denoted by $Pm|\text{prec}, \text{pmtn}, c= 1|C_{\max}$.

Our goal is to assign each job to a single machine and specify the schedule of tasks such that the precedence constraints and communication delay constraints among jobs are satisfied. 
Formally, we define a valid schedule as follows. 

	A  schedule $\calS$ for a subset $A' \subseteq A(J)$ of tasks to an interval $I \subseteq [T]$ with integer length is a function 
	$\calS: A' \rightarrow [m] \times I$ that indicates the (machine, slot) pair that each task is assigned to.  For every $a \in A'$, we then use $\calS_\mac(a)$ and $\calS_\tim(a)$ to denote the first and second component of $\calS(a)$ respectively. 

\begin{definition}
\label{def:comvalid}
	A schedule $\calS$ for $A' \subseteq A(J)$ is valid if it satisfies the following constraints.
	\begin{itemize}
		\item Capacity Constraints: for every two tasks $a \neq a' \in A'$, we have $\calS(a) \neq \calS(a')$. 
		\item No-migration Constraints:  For every pair of tasks $a \sim a' \in A'$, we have $\calS_\mac(a) = \calS_\mac(a')$.
		\item Precedence Constraints:  For every pair of tasks $a \prec a' \in A'$, we have $\calS_\tim(a) < \calS_\tim(a')$.
		\item Communication Delay Constraints: Consider a pair of jobs  $j \prec j'$ with precedence constraints. 
		  Suppose $a_{1,j'} \in A'$ and  $a_{p_{j},j} \in A'$. If $\calS_\mac(a_{1,j'}) \neq \calS_\mac(a_{p_{j},j})$, then 
		  $$\calS_\tim(a_{1,j'}) > \calS_\tim(a_{p_{j},j}) + 1$$
	\end{itemize}
\end{definition}
 
The first three constraints of the above definition is same as Definition \ref{def:valid} for the no-delay problem. 
So, let us focus on the communication delay constraints. 
For a pair of jobs $j \prec j'$ with precedence constraints, we enforce the communication delay constraints only if the first task of $j'$ (that is $a_{1,j'}$)
and the last task of $j$ ( that is $a_{p_{j},j}$) are in the set $A'$.
The reason is that if either one of them is not in the set $A'$, it means that our algorithm has {\em discarded} that task.
For such jobs, we enforce the communication delay constraints when inserting the discarded tasks back.

\subsection{LP Relaxation}
Similar to our first result, the algorithm to prove Theorem \ref{thm:delayresult} is also based on rounding Sherali-Adams lift of a LP for the problem. 
We now give a new LP relaxation for minimizing makespan with precedence and communication delay constraints, which extends the LP for the no-delay problem.
Similar to the LP for the no-delay problem, we use the variables $x_{(a, i, t)}$, which are intended to be 1 if the task $a \in A(J)$ is assigned to machine $i \in [m]$ at time $t \in [T]$.
Let $[m]_{-i}$ denote the set $\{ [m] \setminus \{i\} \}$. In our new LP, we have all the constraints in the LP for the precedence constraints problem (i.e, Constraints \eqref{e:taskschedule} to \eqref{e:nonnegativity}).  The only new set of constraints we introduce is the following one:

\begin{align}
	x_{(a_{1,j'}, i, t + 1)} + \sum_{i' \in [m]_{-i}} x_{(a_{p_{j},j}, i', t)}  &\leq 1 &\quad &\forall j \prec j', i \in [m], t \in [T-1] \label{e:commconstraints}
\end{align}

Consider a pair of jobs $j \prec j'$. 
We want to guarantee that if  $j$ and $j'$ are scheduled on different machines, then $j'$ starts at least 1 time step after the completion of job $j$.
This is equivalent to satisfying the communication delay constraint on the last task of job $j$ and the first task of job $j'$. 
This further implies that  if $x_{(a_{p_{j},j}, i, t)} = 1$ for some machine $i$ and time slot $t$, then  $x_{(a_{1,j'}, i', t)} = 0$ for all machines $i'$. 
The constraints (\ref{e:commconstraints}) precisely guarantee this.
Therefore, if there is an optimal integral solution with makespan at most $T$, then there is a feasible solution to the LP. 
We use $\calP(T)$ (or $\calP$) to denote the polytope define by LP (\ref{e:taskschedule}-\ref{e:nonnegativity}, \ref{e:commconstraints}). Our main goal in this section is to prove the following theorem.

\begin{theorem}
		\label{thm:comm}
		Let $T$ be the smallest value for which the Sherali-Adams lift of LP (\ref{e:taskschedule}-\ref{e:nonnegativity}) to $r = (\log n)^{((m^2/\epsilon^2).\log \log n)}$ rounds has a feasible solution $\vecx$. Then any feasible solution $\vecx$ can be rounded to produce valid a schedule $\calS: A(J) \rightarrow [m] \times [(1+\epsilon) T]$ satisfying the Definition \ref{def:comvalid}. Therefore, the makespan of $\calS$ is at most  $(1+\epsilon) T$. 
\end{theorem}

Note that above theorem immediately implies Theorem  \ref{thm:delayresult}.

\medskip
\textbf{Remark:} We note that constraints \eqref{e:commconstraints} are not the strongest inequalities we can write for the communication delay constraints. 
We can indeed describe exactly the convex hull of all integral schedules restricted to the 2 jobs $j \prec j'$.  However, our algorithm uses the constraints only at the lowest level of recursion where it schedules jobs by conditioning. So any constraint that gives the correct set of integral solutions will be sufficient for our goal.

\subsection{Rounding Algorithm}
Towards proving Theorem \ref{thm:comm}, we first design an algorithm that only schedules a subset $A(J) \setminus \taskdis$ of tasks.

\begin{lemma}
	\label{lem:maincommunication}
	Let $T$ be the smallest value for which the Sherali-Adams lift of LP (\ref{e:taskschedule}-\ref{e:nonnegativity}, \ref{e:commconstraints}) to $r = (\log n)^{((m^2/\epsilon^2).\log \log n)}$ rounds has a feasible solution $\vecx$. Then there is a valid schedule $\calS : A(J) \setminus \taskdis \rightarrow [m] \times [T]$ with $|\taskdis| \leq \epsilon T$. 
\end{lemma}
\medskip

Our algorithm to prove Lemma \ref{lem:maincommunication} uses $\recursive$ from the first part, but modifies it to satisfy the communication delay constraints.
Before we proceed with description of how we accomplish that, let us summarize the three places in which the algorithm for the no-delay problem actually makes the assignment of tasks to time slots.

\begin{enumerate}
	\item Scheduling of tasks by conditioning:  Consider the first line of $\recursive$ procedure. Here, if the length of the interval $|{I^*}| < 2^{k^2}$,  then the algorithm schedules tasks by conditioning on the LP solution. 
        \item Scheduling of top jobs: The second place where our algorithm assigns tasks to time slots is when inserting the top jobs in the last line of the procedure  $\recursive$. 	
        \item Scheduling of discarded tasks: Lastly, our algorithm assigns tasks to time slots when scheduling the discarded tasks. 
\end{enumerate}

The high level idea of how we deal with the communication delay constraints is the following:
The constraints (\ref{e:commconstraints}) of the LP guarantee that if tasks are scheduled by conditioning, then the communication delay constraints are satisfied.
For every discarded task, we simply create an empty time slot on either side of the time slot on which it is scheduled.
This guarantees that no matter how other jobs that have precedence relationship with it are scheduled, the communication delay constraints are satisfied.
Finally, when scheduling top jobs we will make use of the fact that the chain lengths among top jobs are small, hence the communication delay constraints do not pose a big problem.
We give a formal proof of validity of our schedule later in proof of Theorem \ref{thm:comm}.

We give the modified $\recursive$ algorithm for the delay problem is given in Algorithm \ref{alg:comrecursive}, which we call $\commrecursive$.
First we highlight the places where the two algorithms differ.
\begin{enumerate}
\item Consider the first line of $\commrecursive$. Here, if the length of the interval is sufficiently small, then our algorithm schedules the remaining tasks by conditioning the LP solution $x$. 
         Next, it also discards completely the set of the tasks scheduled in the last time slot of the interval $I^*$; there can be $m$ such tasks. 
         We will make use of this fact in the proof of Theorem \ref{lem:maincommunication} to show that the communication constraints are satisfied across all jobs. 
         Note that the length of the interval at the last level of recursion is at least $\Omega(\log n)$, hence, the number of such  discarded tasks across all the intervals is very small.
\item The main difference between  $\commrecursive$ and $\recursive$ is in scheduling the top jobs.  
For this step we give a new algorithm to insert top jobs such that communication delay constraints taken into account.
         \end{enumerate}

\begin{algorithm}[h]
	\caption{\textsf{PARTIAL-SCHEDULE-COMM}$\left({I^*}, J^*, \jspl^*, \sigma,  x\right)$}
	\label{alg:comrecursive}
	\textbf{Input:} a partial-scheduling instance $({I^*}, J^*, \jspl^*, \sigma, x)$ satisfying Definition \ref{def:partial-scheduling}\\
	\textbf{Output:} a schedule of $\calS^*: A(J^*) \cup A({I^*}, \jspl^*, x) \setminus \taskdis \rightarrow [m] \times I^*$ for some $\taskdis$
	\vspace*{-5pt}
	
	\noindent\rule{\linewidth}{0.2pt}
	\begin{algorithmic}[1]
		\State \textbf{if} $(|{I^*}| < 2^{k^2} )$ \textbf{then} 
		\State \quad schedule the tasks by conditioning on the LP solution $x$.
		\State \quad discard all the tasks scheduled at the last time slot $|I^*|$ of the interval $I^*$. 
		\State \quad return.
		
		\State \textbf{for} {every job $j \in \jspl^*$} \textbf{do}: $x \gets \mysplit(x,  j)$ \label{step:initial-condition-comm}
		\While{there exists ${{I}} \in \calI^*_{0} \cup \ldots  \cup \calI^*_{k^2-1}$ and a chain $\mathcal{C}$ of jobs owned by $I$ with total size at least $\delta|{{I}}|$} \label{step:reduce-chain-start-comm}
		\State $j \gets$ the first job in $\calC$, $a \gets $ last task of $j$, 
		\State take $(i, t)$ such that $x_{(a, i, t)} > 0$  with the largest $t$
		\State $x \gets $ $x$ conditioned on the event $(a, i, t)$
		\State $J^* \gets J^* \setminus \{j\}, \jspl^* \gets \jspl^* \cup \{j\}, \sigma(j) \gets i$
		\State $x \gets \mysplit(x, j)$
		\EndWhile \label{step:reduce-chain-end-comm}
		
		\State Partition the jobs in the set $J^*$ as follows:
		
		$ \jtop^* = \union^{\ell^*-k-1}_{\ell = 0} J^*_{\ell}(\vecx) $; 
		$ \jmid^* = \union^{\ell^*-1}_{\ell = \ell^*-k} J^*_{\ell}(\vecx) $;
		$ \jbot^* = \union^{\log T^*}_{\ell = \ell^*} J^*_{\ell}(\vecx) $,
		
		\State where $\ell^* \in \{k, \ldots, k^2\}$ is chosen satisfying the condition below:
		\begin{eqnarray*}
			|A(\jmid^*)| \leq \frac{\epsilon}{4} \cdot \frac{T^*}{\log T} + \frac{\epsilon}{2m} \cdot \left(|A(\jmid^*)| + |A(\jtop^*)| \right)
		\end{eqnarray*}
		
		\For{every interval ${I} \in \calI^*_{\ell^*}$} 
		\State  $\recursive\big({I}, \jbot^*({I},\vecx), \jspl^*, \sigma, \vecx\big)$
		\EndFor
		\State Insert $\jtop^*$ to ${I^*}$ using Lemma \ref{lem:commtopdiscardedtsks}.			
	\end{algorithmic}
\end{algorithm}	

From the pseudo-code of $\commrecursive$, it is clear that most of the lemmas proved for the precedence constrained scheduling directly extends to the communication delay.
Thus it only remains to argue that one can schedule the top jobs without discarding too many tasks even when there are communication delay constraints. 

\subsection{Scheduling the Top Jobs With Communication Delay Constraints}
\label{label:commtopjobs}
Now we give more details about $\commrecursive$ to schedule top jobs respecting the communication delay constraints.
Fix a recursive invocation of the procedure with the input instance  instance $(I^*, J^*, \jspl^*, \sigma, x)$.
We assume that the input satisfies Definition \ref{def:partial-scheduling} given in the first half of the paper.
As, the first 4 steps of $\commrecursive$ procedure remains exactly the same as  $\recursive$, we assume that we have a partial schedule of tasks belonging to the bottom jobs and the special jobs. 
Formally,

$$
 \calS: (A(\jbot^*) \cup A(I^*,\jspl,\vecx)) \setminus A_{\mathrm{bottom\hyphen discarded}} \rightarrow [m] \times I^*
$$ 

Let $\hat{A}:= (A(\jbot^*) \cup A(I^*,\jspl,\vecx)) \setminus A_{\mathrm{bottom\hyphen discarded}}$ for rest of this subsection. 
 We want to extend the schedule $\calS$ to $\calS^*$ which includes most of the tasks of the set $A(\jtop^*)$. 
 In particular, we want to prove the following lemma which is a counterpart to Lemma \ref{lem:topdiscardedtsks}, but with communication delay constraints.
 
 \begin{restatable}{lemma}{commtopdiscardedtsks}
	\label{lem:commtopdiscardedtsks}
	The valid schedule $\calS: \hat{A} \rightarrow I^*$ can be extended to a valid schedule 
	$$
	\calS^*: \left (\hat{A} \cup A(\jtop^*) \right) \setminus A_{\mathrm{top\hyphen discarded}}  \rightarrow [m] \times I^*
	$$
	such that $|A_{\mathrm{top\hyphen discarded}}| \leq \frac{\epsilon}{4} \cdot \frac{T^*}{\log T} $.
\end{restatable}

 \smallskip
 Our strategy to prove the above lemma follows the same framework as in the no-delay problem.
 In the first stage, we build a {\em tentative assignment} of tasks in  $A(\jtop^*)$ in the slots left open by $\calS$. 
 At this stage, we only make sure that capacity constraints (that is, only 1 task is scheduled in 1 time slot on a machine) and the precedence constraints between tasks in  $A(\jtop^*)$  and $\hat{A}$ are satisfied. 
 Both the communication delay constraints and the non-migratory constraints may be violated, and will be fixed in the second stage.
 During this step we discard some tasks from the set $A(\jtop^*)$, which we denote by $A^1_{\mathrm{top\hyphen discarded}}$.
We obtain this schedule by applying the tentative schedule algorithm from Section \ref{tentativeschedule} for the no-delay problem. 
Then we get the following lemma.

\begin{lemma}[\cite{LR16}]
\label{lem:commtentative}
	A feasible partial schedule $\calS: \hat{A} \rightarrow I^*$ of tasks in bottom jobs can be extended to a new schedule $\calS':  \left (\hat{A} \cup A(\jtop^*) \right) \setminus A^1_{\mathrm{top\hyphen discarded}} \rightarrow [m] \times I^*$ satisfying following properties: 
	\begin{enumerate}
		\item Consider $j \in \jtop^*$ and let $a \in A(j)$. Then in the schedule $\calS'$,  either $a$ is assigned in the interval $[r^*_j, d^*_j]$ or $ a \in A^1_{\mathrm{top\hyphen discarded}}$.
		\item The total number of discarded tasks $| A^1_{\mathrm{top\hyphen discarded}}| \leq 4m2^{-k}T^*$.
		\item The precedence constraints between tasks in $\left ( A(\jtop^*) \setminus A^1_{\mathrm{top\hyphen discarded}} \right)$ and $\hat{A}$ are respected.
		\item The capacity constraints are satisfied.
	\end{enumerate}
	\end{lemma}
 
Note that above conditions do not guarantee that the communication delay constraints hold between jobs in $\left ( A(\jtop^*) \setminus A^1_{\mathrm{top\hyphen discarded}} \right)$ and $\hat{A}$.
Next,  we convert the tentative schedule into a valid partial schedule respecting all the precedence constraints and communication delay constraints.  
During this step our algorithm  discards some more tasks from the set $A(\jtop^*)$, which we denote by $A^2_{\mathrm{top\hyphen discarded}}$. 
Define $A^1_{\mathrm{top\hyphen discarded}}  = A^1_{\mathrm{top\hyphen discarded}} + A^2_{\mathrm{top\hyphen discarded}}$.
Our final schedule
 
$$\calS^*:  \left (\hat{A} \cup A(\jtop^*) \right) \setminus A_{\mathrm{top\hyphen discarded}} \rightarrow [m] \times I^*$$

satisfies the following guarantees:

\begin{itemize}
\item For every job $j \in \jtop^*$, all the non-discarded tasks are scheduled within $[r^*_j, d^*_j]$. 
\item The assignment of tasks in $\hat{A}$ remains same as that in the schedule $\calS'$.
\item The communication delay constraints are satisfied among all non-discarded job.
\end{itemize}

The first two invariants guarantee that the precedence constraints are satisfied among all non-discarded tasks. 
We will argue $\calS^*$ also satisfies the communication delay constraints, and hence is a valid partial schedule of non-discarded tasks satisfying Definition \ref{def:comvalid}.
Similar to what we did in the first half of the paper, we build schedule $\calS^*$  by first designing an algorithm for a new stand-alone scheduling problem, then using it as a black-box for scheduling top jobs.

\subsection{A Deadline Scheduling Problem with Precedence and Communication Delay Constraints}

We are given a set of jobs $J$ with processing lengths $p_j$, release times $r_j$, deadlines $d_j$. 
The jobs have precedence constraints and communication delay constraints.
The precedence constraints among jobs satisfy the property that  if $j \prec j'$, then $r_j \leq r_{j'}$ and  $d_j \leq d_{j'}$. 
Recall that we use $\Delta(J)$ to denote the maximum chain length in $J$. 
The time horizon $\{1, 2, \ldots, T\}$ is partitioned into $p$ equal sized intervals $I_1, I_2, \ldots, I_p$. 
The release times and deadlines of jobs correspond to the beginning and the end of the intervals. 
For each machine $i$, we are given a capacity function $\capa_i:[T] \rightarrow \{0, 1\}$. 
If $\capa_i(t) = 1$, then the time slot $t$ on machine $i$ is available to schedule a task in $A(J)$.  

Suppose there is a schedule of tasks $\calS':A(J) \rightarrow [m] \times [T]$ that assigns each task in $A(J)$ to a machine, time slot pair such that no two tasks are assigned to the same machine and the same time slot; that is, capacity constraints on machines are satisfied. 
Moreover, $\calS'$ ensures that for each job $j \in J$,  all its tasks $A(j)$ are scheduled within $[r_j, d_j]$. 
However, $\calS'$ may not respect the precedence and communication delay constraints in $J$, the schedule may not be non-migratory.  
Our goal is to schedule each $j \in J$ on a single $i$ such that the precedence and communication delay constraints among the jobs is satisfied, and the schedule is  non-migratory.
We prove the following theorem in this subsection.

\begin{theorem}
	\label{thm:commedf}
	There exists an algorithm that in polynomial time  converts the schedule $\calS'$ into a valid schedule $\sigma$ that satisfies the following properties:
	\begin{enumerate}
	\item It partially schedules every job on exactly one machine (no-migration). 
	For a job that is partially scheduled, we discard the remaining tasks;  
	For the sake of the precedence constraints, we assume that every partially scheduled job or a fully discarded job is completely processed. 
	\item The precedence constraints among the jobs is satisfied. 
	\item The communication delay constraints among the jobs is satisfied, as given in Definition \ref{def:comvalid}.
	\item The total number of discarded tasks is at most $6p^2m\Delta(J)$.
\end{enumerate}
\end{theorem}

We prove the above theorem by extending the procedure $\textsf{EDF+ECT}$ described in Algorithm \ref{alg:edf+ect} to satisfy the communication delay constraints.
We give the pseudocode in Algorithm \ref{alg:edf+ect+comm}, and refer to the procedure as \textsf{EDF+ECT+COMM}.

\begin{algorithm}[h]
		\caption{\textsf{EDF+ECT+COMM}}
		\label{alg:edf+ect+comm}
		\textbf{Input:} A set of jobs $J$ with release times, deadlines and precedence constraints; capacity function $\capa_i: [T] \rightarrow \{0, 1 \}$ for each machine $i$.\\
		\textbf{Output:} Schedule $\calS: A(J) \rightarrow [m] \times ([T])$  such that for $a \sim a'$, either $a$ or $a'$ belongs to $\taskdis$ or $\calS_\mac(a) = \calS_\mac(a')$.
		\vspace*{-5pt}
		
		\noindent\rule{\linewidth}{0.5pt}
		
	\begin{algorithmic}[1]
		\State Sort the jobs in $J$ in the increasing order of their deadlines. Reindex the jobs so that $J:= \{1,2 \ldots, n\}$ and $d_1 \leq d_2 \leq  \ldots \leq d_n$. 
		\State Initialize $\taskdis = \emptyset$.

	\For{ $j = 1$ to $n$}
		\State Find the earliest time slot $t \in [T]$ such that following conditions hold:
			i) $\capa_i(t) = 1$ for some machine $i \in [m]$ and  $r_j \leq t \leq d_j$;
			ii) $C_{j'} < t$ for all $j' \prec j$.   
			
		\State  If no such $t$ exists, then set $B_j := \drop$ and add all tasks $A(j)$ to the set $\taskdis$.
		\State  If there is a $t$ satisfying the conditions above, set $B_j = t$ and do the following. \label{comm:bj}
		\State  Find the set of machines $M' \subset [m]$ such that $\forall i \in M'$, $\capa_i(B_j + 1, d_j - 1) \geq p_j$.
	
		\If {$|M'| = 0$}
			\State $i^* = \displaystyle \max_{i} \{ \capa_i(B_j + 1, d_j -1) \}.$ \label{invariant2}
			\State Schedule $\capa_i(B_j + 1, d_j - 1)$ tasks of the job $j$ in the interval $[B_j + 1, d_j -1]$ on the machine $i^*$. \label{comm:assign}
			\State Set $C_j = d_j$. Add the remaining $p_j - \capa_i(B_j + 1, d_j - 1)$ tasks to the set $\taskdis$. \label{comm:1dj}
			\State Update the capacity function $\capa_{i^*}$ for the machine $i^*$.
		\EndIf
			
		\If {$|M'| \geq 1$}
			\State Find the earliest time slot $t^* \leq d_j - 1$ such that there exists a machine $i^* \in M'$  and $\capa_{i}(B_j + 1, t^*) = p_j$. \label{invariant1comm}
			\State Set $C_j = t^*+ 1$. Schedule the tasks $A(j)$ in the interval $[B_j + 1, C_j - 1]$ on the machine $i^*$ and update $\calS$. \label{comm:2dj}
			\State Update the capacity function $\capa_{i^*}$ for the machine $i^*$.
		\EndIf
	\EndFor
	\State \textbf{return} $\calS$.
	\end{algorithmic}
\end{algorithm} 

To prove how our algorithm satisfies the communication delay constraints, we need the following simple observation.

\begin{observation}
\label{o:emtpyslots}
	Fix a job $j \in J$ for which $B_j \neq \drop$. Consider the active interval $[B_j, C_j]$ as set by the procedure {\textsf{EDF+ECT+COMM}}. Then, no task of job $j$ is scheduled at the time slots $B_j$ and $C_j$.
\end{observation}
\begin{proof}
See Figure \ref{fig:comtop} for a proof by pictures.
Consider the definition of $B_j$ in line \ref{comm:bj} of the algorithm.
This is the first time slot at which job $j$ can be scheduled while respecting the precedence constraints; that is, for every $j' \prec j$, $C_{j'} < B_j$.
Similarly, we define $C_j$ as either $d_j$ or the earliest time slot $t^*+1$ such that there are $p_j$ empty time slots on some machine $i$ in the interval $[B_j, t^*]$; See the lines \ref{comm:2dj} and \ref{comm:1dj}.
From lines \ref{comm:assign} and \ref{comm:2dj}, it is clear that our algorithm does not schedule any task of job $j$ at the time steps $B_j$ and $C_j$.
\end{proof}

\begin{figure}[!ht]
	\centering
	\includegraphics[width=1.0\textwidth]{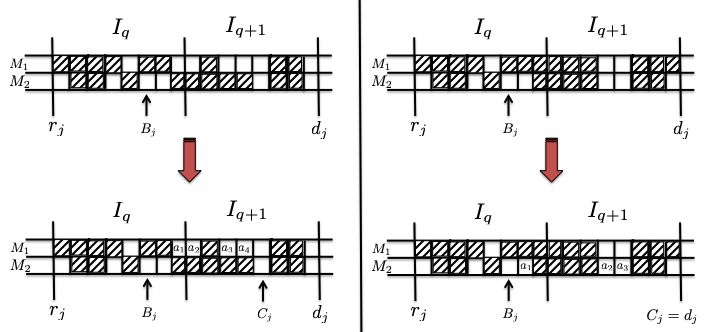} 	
	\caption{\label{fig:comtop} The figure illustrates inserting a job $j$ with $p_j =4$ on two machines. On the left, the job is fully scheduled in the interval $[B_j, C_j]$ on machine $M_1$. On the right, the task $a_4$ is discarded but the remaining tasks are scheduled on $M_2$.  In the both cases, notice that no task of $j$ is scheduled either at $B_j$ or at $C_j$.}
\end{figure}
 
We will argue in the proof of Theorem \ref{thm:comm} that above observation immediately implies that the communication delay constraints are satisfied for a pair of jobs $j$ and $j'$ if either one of them happen to be a top job.  

For now, we focus on arguing that {\textsf{EDF+ECT+COMM}} did not discard too many tasks.
Interestingly, we show that the total number of tasks discarded by our algorithm to enforce communication delay constraints is only factor 5 more than {\textsf{EDF+ECT}}.
The intuition is that every time that was left empty to satisfy a communication delay constraint should also account for the decrease in chain length.
As the chain length among top jobs is small, we argue that wasted time slots is also small.
We now give more details about {\textsf{EDF+ECT+COMM}}.

We call a time slot $t$ on machine $i$ as {\em idle} if $\capa_i(t) = 1$ and our schedule $\calS$ does not assign any task at time $t$ on machine $i$. 
The idle time slots correspond to the number of tasks we discard, and our goal going forward is to show that there are not too many idle time slots in $\calS$. 
The following observations are needed for proving Theorem \ref{thm:edf}.

\begin{observation}
\label{o:commactivejob}
	Fix an interval $I_q$ for some $q \in [p]$. Suppose the following two conditions hold:
	\begin{itemize}
	\item  There is a time slot $t^* \in I_q$ that is idle on machine $i$ in $\calS$.
	\item  There exists a job $j^*$ with $t^* \in [r_{j^*}, d_{j^*}]$ and $B_{j^*} > t^*$ or all the tasks $ A(j^*) \in \taskdis$.
	\end{itemize}
	 Then there exists a job $j$ such that $t^* \in [B_{j}, C_{j}]$ and $j \prec j^*$.
\end{observation}
\begin{proof}
Proof of the above observation follows by the fact that if no such job $j$ exists, then when our algorithm considers the job $j^*$ then it would set $B_j = t^*$; refer to line 6 in  {\textsf{EDF+ECT+COMM}}.  
\end{proof}

\begin{observation}
\label{o:commconservation}
	Consider a job $j \in J$ with $B_j \neq \drop$ that is active in the interval $[B_j, C_j]$. Let $p'_j$ be the total number of tasks scheduled in $[B_j, C_j]$ on machine $i$. Then for any other machine $i' \neq i$, the number of idle time slots in the interval $[B_j, C_j]$ is at most $p'_j + 2$.
\end{observation}
\begin{proof}
See Figure \ref{fig:comtop} for a proof by pictures.
Consider the case when $p'_j = p_j$. 
In this case, the lemma follows from the observation that our algorithm assigns jobs to machines on which they will have earliest completion time; See line \ref{invariant1comm} of the algorithm.
Now consider the case when $p'_j \neq p_j$. 
As job $j$ was scheduled on machine $i$, from line \label{invariant2}, it follows that $i$ had the maximum number of empty slots in the interval $[B_j+1, d_j-1]$, which is at most $p'_j$.
Therefore, the maximum number of empty slots on any machine in the interval $[B_j, d_j]$ is at most $p'_j+2$. 
Since in this case we set $C_j = d_j$, we complete the lemma.
\end{proof}

We will use the above facts to argue that the number of idle time slots on any machine is small.

\begin{lemma}
\label{l:commemptyslots}
	Consider any arbitrary time  interval  $I := \{t', \ldots, t''\} \subseteq I_q$ for some $q \in [p]$. Suppose there is at least one job $j^*$ with $I\subseteq [r_{j^*}, d_{j^*}]$ and $B_{j^*} > t''$ or all the tasks $ A(j^*) \in \taskdis$ ($B_j = \drop$). Then, for any machine $i \in [m]$ the number of idle time slots in $I$ is at most 3$\Delta(J)$.
\end{lemma}

Before we proceed with the proof, contrast the above lemma with  Lemma \ref{l:emptyslots}; the number of idle slots increases in communication delay case increases by a factor of 3.

\begin{proof}
	We prove this by showing a contradiction that if there are more than $3\Delta(J)$ idle time slots on a machine, then the maximum chain length among jobs in  $J$ is more than $\Delta(J)$. 
	Consider a machine $i$ with $3\Delta(J) + 1$ idle time slots in the interval $I$. 
	Let $t^* \in I$ be the latest time slot on machine $i$ that is idle. 
	Since $j^*$ is available at $t^*$, by Observation \ref{o:commactivejob}, it must be the case that there exists a  job $j_1 \prec j^*$ that is active at time $t^*$, which implies $t^* \in [B_{j_1}, C_{j_1}]$. 
	Now consider the latest time slot $t' < B_{j_1}$ that is empty on machine $i$. 
	We claim that $j_1$ is available for processing at time $t'$.
	This follows from our assumptions that if $j_1 \prec j^*$, then $r_{j_1} \leq r_{j^*}$ and the release times and the deadlines of $J$ align with the beginnings and the endings of the intervals.
	Therefore, there must be a job $j_2 \prec j_1$ such that $t' \in [B_{j_2}, C_{j_2}]$. 
	Moreover, $C_{j_2} < B_{j_1}$ as  $j_2 \prec j_1$. 
	We continue by induction to construct a chain of jobs $j_y \prec j_{y-1} \prec  \ldots \prec j_1 \prec j^*$ such that $[B_y, C_y] \cup [B_{y-1}, C_{y-1}] \cup \ldots \cup [B_{j_1}, C_{j_1}]$ covers all the empty slots on machine $i$ in the interval $I$. Furthermore, for any two intervals $I', I'' \in \left\{ [B_y, C_y], [B_{y-1}, C_{y-1}], \ldots ,  [B_{j_1}, C_{j_1}] \right\}$, $I' \cap I'' = \emptyset$.
	The total processing lengths of the jobs in the chain $j_y \prec j_{y-1} \prec  \ldots \prec j_1 \prec j^*$, $\sum^{y}_{v = 1}p_{v}$, is at most $\Delta(J)$, since the maximum chain length among $J$ is at most $\Delta(J)$. By Observation \ref{o:commconservation},  for any job $j_v$ belonging to the chain, there can be at most $p_{j_v} + 2$ empty slots on the machine $i$ in the interval $[B_{j_v}, C_{j_v}]$. 
	Therefore, $\sum^{y}_{v = 1}\text{Idle}_i(B_{j_v}, C_{j_v}) \leq \sum^{y}_{v = 1}(p_{v} + 2) \leq 3 \Delta(J)$. 
	However, $[B_y, C_y] \cup [B_{y-1}, C_{y-1}] \cup \ldots \cup [B_{j_1}, C_{j_1}]$ covers all the empty slots on machine $i$, which is a contradiction.
\end{proof}

The above lemma implies the following useful corollary.

\begin{corollary}
\label{c:commreverseempty}
Suppose there is an interval $I_q$ and a machine $i$ with more than $3\Delta(J)$ idle time slots. If there is a job $j^*$ such that $B_{j^*} \in I_{q+1} \cup I_{q+2} \cup \ldots \cup I_p$, then the release time of job $r_{j^*} \in I_{q+1} \cup I_{q+2} \cup \ldots \cup I_p$.
\end{corollary}
\begin{proof}
For contradiction, let us assume that the release time of $j^*$ belongs to $I_{q'}$, where $q' < q+1$.
Recall that all jobs are released at the beginning of the intervals.
Now, we invoke the previous lemma on the interval $I_q$ and the job $j^*$, which gives us a necessary contradiction.
\end{proof}

The next lemma shows that even if a job is partially discarded then the number of idle time slots on a machine cannot be too much. 

\begin{lemma}
\label{l:commtotalemptyslots}
	Consider any arbitrary time  interval  $I := \{t', \ldots, t''\} \subseteq I_q$ for some $q \in [p]$. Suppose there is at least one job $j^*$ with $I\subseteq [r_{j^*}, d_{j^*}]$ and $B_{j^*} > t'$ or a subset of the tasks $ A(j^*) \in \taskdis$.
	Then, for any machine $i$ the number of idle time slots in $I$ is at most $6\Delta(J)$.
\end{lemma}

\begin{proof}
If $B_{j^*} \geq t''$, then the proof follows from Lemma \ref{l:commemptyslots}. 
Therefore, $B_{j^*} \leq t''$ and some tasks of $j^*$ got discarded. 
By Observation \ref{o:commconservation}, in the interval $[B_j, t'']$ there cannot be more than $p_{j^*} + 2 \leq \Delta(J) + 2$ idle slots on any machine $i$. 
By applying Lemma \ref{l:commemptyslots} to the interval $[t', B_{j^*}-1]$ and job $j^*$, we conclude that in the interval  $[t', B_{j^*}-1]$ there can be at most $3\Delta(J)$ idle time slots. 
Therefore, there can be at most $3\Delta(J) + \Delta(J) + 2 \leq 6\Delta(J)$ idle slots on any machine $i$.
\end{proof}

With above two observations, it is easy to prove Theorem \ref{thm:commedf}. 
For brevity, we use $\calS^{-1}(I)$ to denote the set of tasks scheduled in the interval $I$ in $\calS$.

\begin{proof}[Proof of Theorem \ref{thm:commedf}]

Our algorithm guarantees that in the schedule $\calS$ all the tasks $A(j)$ of a job $j \in J$ are assigned to a single machine and the precedence constraints among the jobs is satisfied. 
Moreover, communication delay constraints are satisfied as our algorithm while scheduling a top job $j$ always leaves the
first slot in the active interval $[B_j, C_j]$ empty.
Therefore, it only remains to show that the number of discarded tasks $|\taskdis| \leq 6p^2m\Delta(J)$, which now readily follows from the proof of Theorem \ref{thm:edf} from the first part.
\end{proof}
\medskip

Now we are ready to prove Lemma \ref{lem:commtopdiscardedtsks}, which bounds the total number tasks discarded in inserting top jobs.

\begin{proof}
We obtain the schedule $\calS^*$ by applying Theorem \ref{thm:commedf} on the jobs in the set $\jtop^*$. 
For every job $j \in \jtop^*$, we define the truncated job length $p'_j$ by taking into account the discarded tasks in $A^1_{\mathrm{top\hyphen discarded}}$. 
We define the release time  of  $j$ as $r^*_j$ and the deadline as $d^*_j$, where $r^*_j, d^*_j$ are defined in Lemma \ref{lem:commtentative}  
For each machine $i$,  $\capa_i(t) = 1$ if time $t$ on machine $i$ is not assigned any task in the schedule $\calS$. Recall that $\calS$ gives an assignment of subset of tasks in the bottom jobs and the special jobs ($\hat{A}$) in the interval $I^*$. 

The maximum chain length of jobs in $\jtop^*$ is at most $k^2 \delta T^*$, where $T^*$ is the length of interval $I^*$. 
This follows by applying Lemma \ref{lem:cutchains} for the no-delay problem to our setting.
Therefore, by Theorem \ref{thm:commedf}, the total number of tasks discarded in converting the tentative schedule into an actual schedule is given by
$$
|A^2_{\mathrm{top\hyphen discarded}}| \leq 5p^2 m \cdot k^2 \delta T^*.
$$

By Lemma \ref{lem:commtentative},

$$
|A^1_{\mathrm{top\hyphen discarded}}| \leq 4m 2^{-k} T^*.
$$

Therefore, 
\begin{eqnarray*}
|A_{\mathrm{top\hyphen discarded}}| &=& |A^1_{\mathrm{top\hyphen discarded}}| + |A^2_{\mathrm{top\hyphen discarded}}|  \\
&\leq& 5p^2 m \cdot k^2 \delta T^* + 4m 2^{-k} T^*   \\
&\leq& \frac{\epsilon}{4} \cdot \frac{T^*}{\log T} 
\end{eqnarray*}

The last inequality follows from substituting  $p = 2^{\ell^*} \leq 2^{k^2}$, $k = \frac{O(1) m}{\epsilon} \log \log T$, and we set $\delta = \frac{\epsilon}{16k^2m2^{2k^2} \log T}$.

\end{proof}

\subsection{Proof of Theorem \ref{thm:comm}}
First we note that guarantee of Lemma \ref{lem:commtopdiscardedtsks} is exactly same as the guarantee of Lemma \ref{lem:topdiscardedtsks} for the no-delay problem.
As rest of the steps of $\commrecursive$ remains exactly same as $\recursive$, all the lemmas we proved for $\recursive$ procedure also directly extend to the $\commrecursive$.
Hence it is not hard to see that a proof of Lemma  \ref{lem:maincommunication} follows by repeating the arguments in the proof of Lemma \ref{lem:mainmakespan} for the no-delay problem. 
We omit the proof as the details are fairly straightforward.

\medskip
Now we have all the ingredients to prove our main result for the delay problem. 

\begin{proof} [Proof of Theorem \ref{thm:comm}.]
Set $\epsilon' = \epsilon/3$. 
By Lemma \ref{lem:maincommunication}, there is a partial schedule for $\calS$ for $A(J) \setminus \taskdis$ of makespan $[T]$ with $|\taskdis| \leq \epsilon'T$.
We extend $\calS$ to a valid schedule $\calS^*$ for  $A(J)$ with makespan $[T + 3|\taskdis|]$ as follows. 
We give a procedure to insert one discarded task such that all the precedence and communication constraints are satisfied.
The final schedule is constructed by repeating this procedure for each discarded task.
Consider a discarded task $a$. 
Let $t$ be the earliest time step in $\calS$  where $a$ can be scheduled respecting the precedence constraints.
Now, create three new private slots at $t, t+1$ and $t+2$ for  the task $a$. 
Schedule the task $a$ at the time step $t+1$ on the machine $i$ such that non-migratory constraints are satisfied.
The makespan increases by an additive factor of 3 for every discarded task, and hence the total increase in makespan is  $3|\taskdis|$.
As $|\taskdis| \leq \epsilon T/3  $, we conclude that makespan of the our schedule is at most $(1+\epsilon)T$.

It is easy to see that the precedence constraints and non-migratory constraints are satisfied by our schedule.
It remains to argue about the communication delay constraints.

Fix any two jobs $j$ and $j'$ such that $j \prec j'$. Recall that $a_{p_j, j}$  denotes the last task of job $j$ and $a_{1,j'}$ denotes the first task of job $j'$.
Let $t, t'$ be the time slots at which  the tasks $a_{p_j, j}$ and $a_{1,j'}$ are scheduled by our algorithm.
We will argue that if $a_{p_j, j}$ and  $a_{1,j'}$ are scheduled on two different machines then, $t' > t + 1$.
We consider the following cases.

\begin{itemize}
\item Our algorithm discarded either $a_{p_j, j}$ or $a_{1,j'}$. Let us, without loss of generality, assume that $a_{1,j'}$ was discarded. 
From our description of the algorithm to schedule discarded jobs, it follows that the time slot $t'-1$ was empty. Therefore, $t' > t + 1$.

\item Our algorithm scheduled either $a_{p_j, j}$ or $a_{1,j'}$ using the procedure {\textsf{EDF+ECT+COMM}}.  
Consider the case when $a_{1,j'}$ was in the set of top jobs. 
From Observation \ref{o:emtpyslots}, it follows that $t' \geq B_{j'} +1$, where $B_{j'}$ denotes the earliest time step when the job $j'$ can be scheduled respecting the precedence constraints. 
Note that $B_{j'} > t$. 
Hence, $t' \geq B_{j'} +1 > t+1$.
On the other hand, consider the case when $a_{p_j, j}$ was scheduled using {\textsf{EDF+ECT+COMM}}. 
Again from  Observation \ref{o:emtpyslots}, it follows $t < C_j \leq d_j$, as no task of job $j$ is scheduled at the time step $d_j$.
Further,  due to the precedence constraints $t' > d_j$. Thus, $t' > t + 1$.

Observe that in these cases we simply assumed  the worst scenario that the jobs are scheduled on different machines by our algorithm.

\item Both $a_{p_j, j}$ and $a_{1,j'}$ was scheduled by our algorithm by conditioning on the LP solution in the first step of $\commrecursive$. 
We consider two cases. 
Suppose both $a_{p_j, j}$ and $a_{1,j'}$  belonged to the same interval $I^*$, and were scheduled on two different machines. In this case, $t' > t + 1$ follows by the LP constraints (\ref{e:commconstraints}).
Consider the second case where $a_{p_j, j}$ is scheduled by conditioning in the interval $I^*$ and $a_{1, j}$ is scheduled by conditioning in a different interval $I^{**}$.
Clearly, $I^*$ has to be to the left of  $I^{**}$.
Now consider the first line of procedure $\commrecursive$. 
Here, we completely discard all the tasks scheduled in the last time slot in the interval $I^*$.
If there was at least one such discarded task, it is clear that $t' > t + 1$ as every discarded task creates one empty time slot to the right of it.
If no tasks were discarded, then it implies that in last time slot of the interval $I^*$ was completely empty.
In this case also, $t' > t + 1$. 

Therefore, we conclude that our algorithm satisfies communication delay constraints.
\end{itemize}
\end{proof}

Thus to prove our second main Theorem, we only need to argue that proof of Theorem \ref{thm:comm} also extends when.
\begin{proof} [Proof of Theorem \ref{thm:delayresult}]
Define $\beta = \max_{j \prec j'} \{ c_{j,j'} \}$.
We need the following changes to complete our proof.
\begin{itemize}
	\item The communication constraints in our LP become: 
		$$x_{(a_{1,j'}, i, t + 1)} + \sum_{i' \in [m]_{-i}} \sum^{t}_{t' = t-\beta} x_{(a_{p_{j},j}, i', t')}  \leq 1 \quad \forall j \prec j', i \in [m], t \in [T-1]$$ 
	\item In the first line of procedure $\commrecursive$, we completely discard all the tasks scheduled in the last $\beta$ time slots in the interval $I^*$.
	\item We modify the schedule obtained by running \textsf{EDF+ECT+COMM} on the set of top jobs as follows: For every top job $j$, we discard first $\beta-1$ and the last $\beta-1$ tasks scheduled in the interval $[B_j, C_j]$.
	\item While inserting a discarded task $a$, we create $2\beta + 1$ time slots, and insert the task in the middle.
\end{itemize}
By repeating the proofs for the case of $\beta = 1$, it is easy to see that the total number of discarded tasks increases by a factor of $O(\beta)$.	
So, by appropriately choosing $\epsilon' = \epsilon/O(\beta)$, and running our entire algorithm by fixing $\epsilon'$ we can show that the makespan of the schedule is at most $(1+\epsilon)$T. 
It is also easy to see that all the communication delay constraints and the precedence constraints are also satisfied by this strategy.
This completes the proof.
\end{proof}

	\newcommand{\tprec}{\mathrm{prec}}
\newcommand{\minUo}{1|r_j, d_j|\sum_j p_j U'_j}

\section{Integrality Gap Instance For Sherali-Adams Hierarchy for $P2|\tprec|C_{\max}$?}
\label{sec:integrality}

In this section, we give some evidence that an $o(\log n)$-level Sherali-Adams lift of the basic LP for $P2|\tprec|C_{\max}$ may not lead to a $(1+\epsilon)$-approximation for non-preemptive precedence constraints problems with arbitrary job lengths. Due to a certain technical difficulty, which will become clear later, we do not quite prove this exact statement;  instead, we introduce a new scheduling problem that is equivalent to a special case of $P2|\tprec|C_{\max}$, and show an integrality gap result for the new problem. However, we believe  the instance we construct is the right one for proving an integrality gap
for $P2|\tprec|C_{\max}$.

\subsection{A Scheduling Problem on Single Machine}
	In our problem, we have a single machine and $n$ jobs $J$, each $j \in J$ having a size $p_j \in \Z_{> 0}$, a release time $r_j \in \Z_{\geq 0}$ and a deadline $d_j \geq r_j+p_j$. There are no precedence constraints among jobs, but they have to be processed during their respective $(r_j, d_j]$ windows. 
	
	In the problem, any job $j \in J$ can be ``partially processed''. More specifically, we can choose a length $p'_j \in [0, p_j]$ for any $j \in J$ and process $j$ non-preemptively only for $p'_j$ units of time in $(r_j, d_j]$.  The objective we consider is then to minimize $\sum_{j\in J}(p_j - p'_j)$, the total job units that are not processed (or ``discarded'').
	We shall use $\sum_j p_jU'_j$ to denote this  objective of minimizing $\sum_{j\in J}(p_j - p'_j)$, where $U'_j = \frac{p_j-p'_j}{p_j}$ is the fraction of job $j$ that is unfinished.\footnote{In the literature, $U_j$ indicates if $j$ is not scheduled; so we use $U'_j \in [0, 1]$ to indicating the fraction of $j$ that is not processed.}   We denote the problem by $\minUo$, and use the tuple $(J, p, r, d)$ to denote an instance of this problem.

Suppose we are given an instance $\calI = (J, p, r, d)$ of $\minUo$. Let $T = \sum_{j \in J}p_j$ and we assume $T \geq \max_{j \in J}d_j$.  We construct an equivalent  instance $\calI'$ of $P2|\tprec|C_{\max}$ as follows. The jobs in $\calI'$ will be $J \cup J^{\chain}$, where $J^{\chain} = \{j^\chain_1, j^\chain_2, \cdots, j^\chain_T\}$ is a set of $T$ unit-length jobs.    For simplicity we define the lengths, release times and deadlines of $J^\chain$ as follows: for every $t \in [T]$, we have $p_{j^\chain_t}=1$ and $r_{j^\chain_t}=t-1$ and $d_{r^\chain_t} = t$. Then the precedence constraints are defined as follows: for every $j, j' \in J \cup J^\chain$, $j \prec j'$ if and only if $d_j \leq r_{j'}$. In particular, this implies $j^\chain_1 \prec j^\chain_2 \prec \cdots \prec j^\chain_T$.

Our goal for $\calI'$ is to schedule the precedence-constrained jobs $J \cup J^\chain$ non-preemptively on 2 machines so as to minimize the makespan. We use $\calI' = (J \cup J^\chain, p, \prec)$ to denote this instance of $P2|\tprec|C_{\max}$. Ideally, we would like to use  one machine to process $J$, and the other one to process $J^\chain$. Then the time window constraints in $\calI$ will correspond to  the precedence constraints in $\calI'$.  By discarding and inserting job units, a good schedule for $\calI$ can be converted to a good schedule for $\calI'$ and vice versa.  This technique was crucially used the in \cite{LR16} for the makespan minimization problem on unit-length jobs.

Formally, the following lemma establishes  the equivalence between the two instances $\calI$ and $\calI'$:
\begin{lemma}
	Let $\calI = (J, p, r, d)$ be an instance of $\minUo$, $T := \sum_{ j \in J}p_j$ and assume $T \geq \max_{j \in J}d_j$. Let $\calI' = (J \cup J^\chain, p, \prec)$ be the instance of $P2|\tprec|C_{\max}$ constructed using the above procedure. Let $\delta \geq 0$ be any constant. Then,
	\begin{itemize}
		\item Any solution to $\calI$ with cost $\delta T$ can be efficiently converted to a solution to $\calI'$ with makespan at most $(1+\delta) T$.
		\item Any solution to $\calI'$ with makespan $(1+\delta) T$ can be efficiently converted to a solution to $\calI$ with cost at most $2\delta T$. 
	\end{itemize}
\end{lemma}
\begin{proof}
	
	Given a schedule for $\calI$ on one machine with at most $\delta T$ job units discarded, we construct a schedule for $\calI'$ as follows. We first add a second machine and schedule all the jobs $J^\chain$ naturally on the machine: we schedule $j^\chain_t$ at slot $(t-1, t]$. The precedence constraints are satisfied since they defined according to $r_j$ and $d_j$ values of jobs, and the schedule respects the time window constraints. Then we can insert the discarded job units back, increasing the makespan by at most $\delta T$. 
	
	Now suppose we are given a solution to $\calI'$ with makespan $(1+\delta) T$.  We focus on the time slots where no jobs in $J^\chain$ are scheduled.  We remove from the schedule these time slots, as well as the job units processed in these slots. We removed at most $2\delta T$ job units.  The resulting schedule has makespan exactly $T$ and each job $j^\chain_t$ is scheduled in $(t-1, t]$. We can assume jobs in $J^\chain$ are processed on the same machine. Since the precedence constraints are satisfied, we have that all jobs $j \in J$ are scheduled within the window $(r_j, d_j]$. Removing the machine for $J^\chain$ gives us a schedule for $\calI$ with at most $2\delta T$ job units discarded. 
\end{proof}

The factor of $2$ in the lemma does not create an issue for our reduction: we are interested in deciding whether the instance $\calI'$ has makespan at most $(1+\epsilon)T$ or at least $cT$ for some absolute constant $c > 1$.  This is equivalent to deciding whether the instance $\calI$ has cost at most $\epsilon T$ or at least $c'T$ for some absolute constant $c' > 0$. Our main theorem is that an $o(\log n)$-level Sherali-Adams lifting of some natural LP relaxation for $\minUo$ can not distinguish between the two cases.  We define the natural LP relaxation first and then give our theorem. 

\subsection{Integrality Gap Result for $\minUo$}
As is typical in the LP/SDP lifting framework, we specify an upper bound $B$ on the solution cost and impose a constraint for the objective function.  In the LP, $x_{j, t, p'}$ indicates weather $j$ is processed during the interval $(t, t+p'] \subseteq (r_j, d_j]$. We require $p' \in [p_j]$ and $t \in [r_j, d_j-p']$; notice that we do not have variables for the cases in which a job $j$ is not processed at all.  For simplicity, we assume all the variables $x_{j, t, p'}$ with $(j, t, p')$ not satisfying the property are identically 0. The LP relaxation is as follows:

\noindent\begin{minipage}{0.5\textwidth}
	\begin{align}
		\sum_{j}\left(p_j- \sum _{t, p'} x_{j, t, p'}p'\right)  &\leq B \label{LP:objective} \\
		\sum_{t,p'} x_{j, t, p'}  &\leq 1 &\quad &\forall j \in J \label{LP:scheduled}
	\end{align}
\end{minipage}
\begin{minipage}{0.48\textwidth}
	\begin{align}
		\sum_{j, t, p' : t' \in (t, t+p']} x_{j, t, p'} &\leq 1 &\quad &\forall t' \in [T] \label{LP:congestion}\\
		x_{j, t, p'} &\geq 0 &\quad &\forall j,  p', t \label{LP:positive}
	\end{align}
\end{minipage} \smallskip

\eqref{LP:objective} says that we can discard at most $B$ job units; \eqref{LP:scheduled} says each job $j$ can be scheduled at most once.  \eqref{LP:congestion} requires that at any time $t'$, at most 1 job can be processed, and \eqref{LP:positive} requires all variables to be non-negative.   For a fixed instance $\calI=(J, p, r, d)$ and a bound $B$, we use $\calP_\calI(B)$ to denote the above polytope.

Our main theorem of the section is as follows:
\begin{theorem}
	\label{thm:negative-main}
	There exists some constant $c > 0$ such that the following holds for every small enough $\epsilon > 0$ and infinitely many integers $n > 0$. There is an instance $\calI = (J, p, r, d)$ of  $\minUo$ with $n = |J|$ jobs and $T:=\sum_{j \in J}p_j  = \max_{j \in J}d_j$ such that the following holds. 
	\begin{itemize}
		\item The optimum solution to $\calI$ has cost at least $0.2 T$. 
		\item $\SA(\calP_\calI(\epsilon T), \floor{c\epsilon \log n}) \neq \emptyset$.
	\end{itemize}
\end{theorem}
\medskip 

	The  remaining part of this section is to prove the theorem.  Throughout, let $\epsilon > 0$ be a small enough constant as in the theorem statement. Let $L > 0$ be a large enough integer; we assume that $L+1$ is an integer power of $2$.   We then define the instance $\calI = (J, p, r, d)$  of $\minUo$ with $n := |J| = 2^{L+1}-1$ and $T := \sum_{j \in J}p_j = (L+1)2^L$.  
	
	The set of jobs $J$ in $\calI$ form a full binary tree of $L+1$ levels, where each node corresponds to a job. The root of the tree is at level $0$ and the leaves are at level $L$. So there are  $2^\ell$ jobs at level $\ell$; all these jobs $j$ have size $p_j = 2^{L - \ell}$; for every $\ell \in [0, L]$ and $k \in [2^\ell]$, the $k$-th job from the left-side in the $\ell$-th level of the tree has release time $r_j  = (L+1)(k-1)2^{L-\ell}$ and deadline $d_j = (L+1)k2^{L-\ell}$.  So, the window size $d_j - r_j$ for every job $j$ is exactly $(L+1)p_j$; the windows of all the jobs form a laminar family represented by the tree structure.  It is easy to verify that the number $n$ of jobs is $2^{L+1}-1$ and the total size of the jobs is $T = (L+1)2^L$, which is equal to $\max_{j \in J}d_j$.  See Figure~\ref{fig:tree} for the illustration of the instance.
	
	\begin{figure}[h]
	  \centering
	  \includegraphics[width=0.7\textwidth]{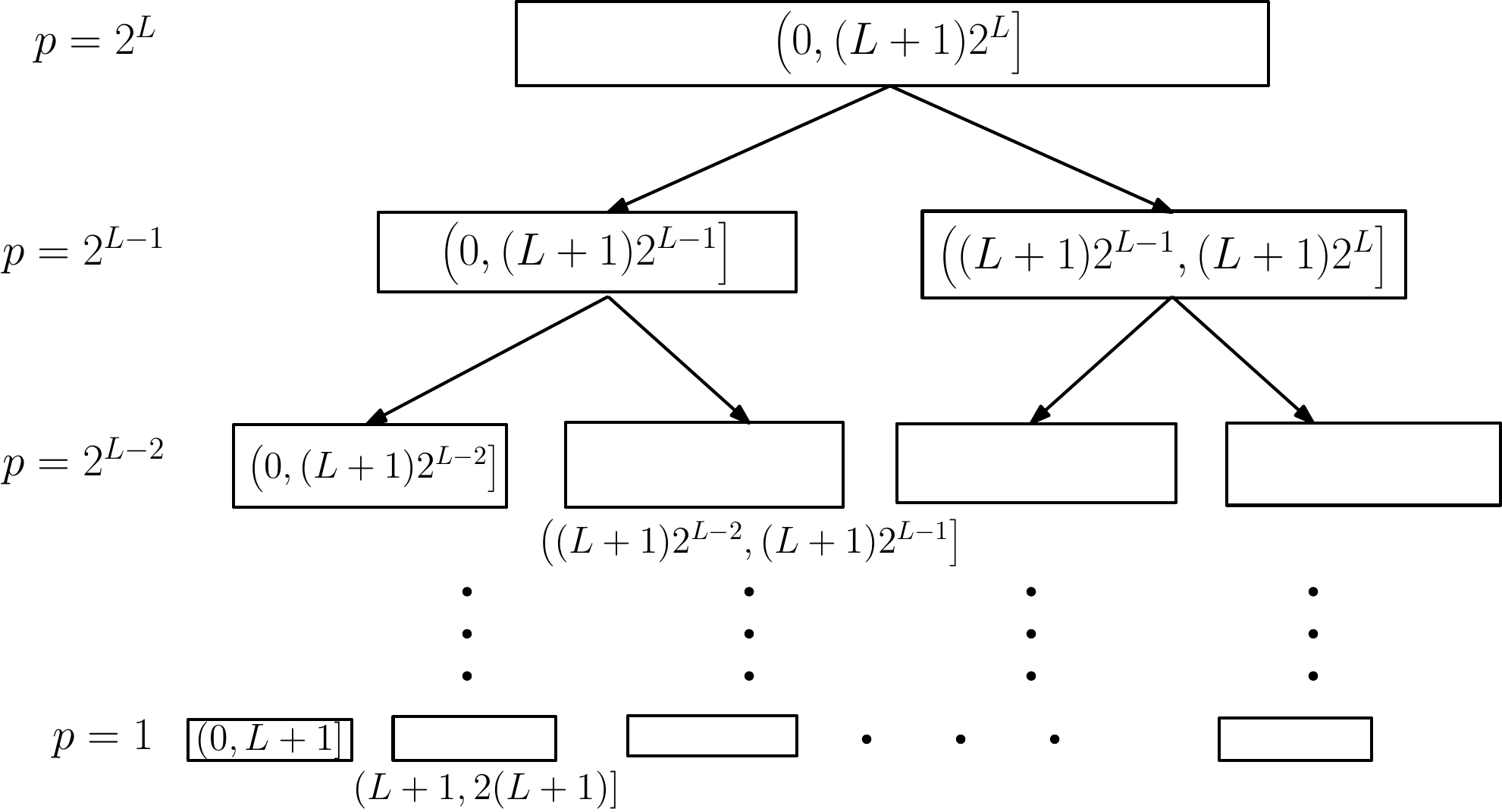}
	  \caption{The gap instance for $\minUo$. The jobs form a tree; the sizes of jobs are $2^L, 2^{L-1}, 2^{L-2}, \cdots, 1$ from top to bottom of the tree. The windows of jobs form a laminar family.} \label{fig:tree}
	\end{figure}

	First we show that the optimum solution to the $\minUo$ instance $\calI = (J, p, r, d)$ is large.
	\begin{lemma}
		Any valid solution to the instance $\calI$ of $\minUo$ has cost at least $0.2T$. 
	\end{lemma}
	\begin{proof}
		We choose some integer $A$ much smaller than $L$, whose exact value will be decided later.  We break the $L+1$ levels of the job tree into 3 classes: the top $(L+1)/2-A$ levels, the middle $A$ levels and the bottom $(L+1)/2$ levels. Roughly speaking, we shall assume that the jobs in the middle levels are processed for free and analyze the conflicts between top and bottom levels.
		
		Focus on a bottom-level $\ell \in [(L+1)/2, L]$, and a top-level job $j$ and assume it is scheduled in $(a_j, b_j]$ in the optimum solution.  Let $a'_j \geq a_j$ be the smallest integer that is a multiple of $(L+1)2^{L-\ell}$ and $b'_j \leq b_j$ be the smallest integer that is a multiple of $(L+1)2^{L-\ell}$. Then we have that $(a'_j, b'_j]$ is the disjoint union of  windows of some level-$\ell$ jobs.  So, these jobs at level $\ell$ can not be processed at all in the optimum solution; we call these jobs forbidden jobs.  Taking all top-level jobs $j$ into consideration (notice that they have disjoint scheduling intervals in the optimum solution), the total number of level-$\ell$ forbidden jobs is at least:
		\begin{align*}
			\sum_{j \text{ top job}} \frac{(b_j - a_j - 2\cdot (L+1)2^{L-\ell})}{(L+1)2^{L-\ell}} = \frac{1}{(L+1)2^{L-\ell}}\sum_{j \text{ top job}}(b_j - a_j) - 2(\#\text{top jobs}).
		\end{align*}
		Let $P_{\mathrm{top}} = \sum_{j \text{ top job}}(b_j - a_j)$; notice that this is the total job units processed for top jobs.   The total number of top jobs is $2^{(L+1)/2-A} - 1 < 2^{(L+1)/2-A} $. Since each level-$\ell$ job has size $2^{L-\ell}$, the total size of level-$\ell$ forbidden jobs is at least
		\begin{align*}
			\frac{P_{\mathrm{top}}}{L+1}- 2^{(L+1)/2-A}\cdot 2^{L-\ell} \geq \frac{P_{\mathrm{top}}}{L+1}- 2^{L-A} = \frac{P_{\mathrm{top}}}{L+1}- \frac{T}{(L+1)2^A}.
		\end{align*}
		The total size of forbidden jobs at all bottom levels is at least $(L+1)/2$ times the above quantity, which is $\frac{P_{\textrm{top}}}2 - \frac{T}{2^{A+1}}$.
		
		Let $P_{\textrm{bot}}$ be the total job units processed for bottom jobs in the optimal solution. Then, we have 
		\begin{align*}
			P_{\textrm{bot}} \leq \frac{T}{2} - \frac{P_{\textrm{top}}}2 +\frac{T}{2^{A+1}}.
		\end{align*}
		So, we have 
		\begin{align*}
			P_{\textrm{top}} + P_{\textrm{bot}} = P_{\textrm{bot}} + \frac{P_{\textrm{top}}}2 +  \frac{P_{\textrm{top}}}2 \leq \frac{T}{2} +\frac{T}{2^{A+1}} + \frac{1}{2}\cdot \frac{((L+1)/2-A)T}{L+1} = \frac{3T}{4} + \frac{T}{2^{A+1}}-\frac{AT}{2(L+1)}.
		\end{align*}
		Considering the middle level jobs, the total scheduled job units in the optimum solution is at most $\frac{3T}{4}+\frac{T}{2^{A+1}}+\frac{AT}{2(L+1)}$. Let $A = \log(L+1)$, then the scheduled jobs is at most $\frac{3T}{4} + O(\log L/L)T = \left(\frac34+o(1)\right)T$. So, if $L$ is large enough, the total scheduled job units is at most $0.8T$, finishing the proof of the lemma.
	\end{proof}
	
	Then we shall give an LP hierarchy solution to $\calI$ with small cost.  Recall that we are allowed to partially process a job; however our fractional solution does not need to take the advantage:  for each job $j$, it either processes it completely, or does not process it at all. Moreover, the fractional solution only starts a job $j$ at a time that is a multiple of $p_j$. Since $d_j - r_j = (L+1)p_j$, there are exactly $L+1$ possible starting times for a fixed job $j$.  For simplicity, we use $x_{j, t}$ to indicate the event that $j$ is processed in $(t, t+p_j]$.  We let $\calD$ denote the set of $(j, t)$ pairs for which the variables can take positive value; that is $\calD:= \{(j, t): j \in J, t \in [r_j, d_j)\text{ is a multiple of } p_j\}$. For convenience, we also treat each $(j, t)$ pair as the interval $(t, t + p_j]$. Then, removing the variables that are identically $0$, and replace $B$ with $\epsilon T$, the LP~\eqref{LP:objective}-\eqref{LP:positive} becomes
	
	\noindent\begin{minipage}{0.48\textwidth}
		\begin{align}
			\sum_{(j, t) \in \calD} x_{j, t} p_j &\geq (1-\epsilon)T \label{LP-new:objective}\\
			\sum_{t: (j, t) \in \calD} x_{j, t} &\leq 1 &\quad &\forall j \in J \label{LP-new:scheduled}
		\end{align}
	\end{minipage}
	\begin{minipage}{0.5\textwidth} \vspace*{-5pt}
			\begin{align}
				\sum_{(j, t) \in \calD: t' \in (t, t+p_j]} x_{j, t} &\leq 1 &\quad &\forall t' \in [T] \label{LP-new:congestion}\\[3pt]
				x_{j, t} &\geq 0 &\quad &\forall (j, t) \in \calD \label{LP-new:positive}
			\end{align}
	\end{minipage}\bigskip

	Let $\epsilon' = \epsilon/4$ let $q = \floor{\epsilon' (L+1)}$ be the number of rounds we shall allow.  Since $L = \Theta(\log n)$, we have $q \geq \floor{c\epsilon \log n}$ if $c$ is small enough.  As $\calI$ and $B=\epsilon T$ are fixed, we can use $\calP'$ to denote the above polytope. Our goal is to prove that $\SA(\calP_\calI(\epsilon T), q) \neq \emptyset$ is not empty. Since from $\calP'$ is obtained from $\calP_{\calI}(\epsilon T)$ by setting some variables to be $0$, it suffices to prove $\SA(\calP' , q) \neq \emptyset$; this is our goal for remaining part of the section. 

	For a 1-round solution, we can simply set $x_{j,t}=1/(L+1)$ for every $(j, t) \in \calD$.  Notice that there are $L+1$ levels of jobs and each job has $L+1$ possible intervals, the LP solution is valid and has cost $0$. We show that with a small loss, we can raise the solution by $q$ levels using the Sherali-Adams hierarchy. 

	Now, we define the solution in $\SA(\calP', q)$. For a set $S \subseteq \calD$, $|S| \leq r$ of variables, we set $x_{S} = \left(\frac{1-\epsilon'}{L+1}\right)^{|S|}$ if $S$ does not lead to a contradiction and $0$ otherwise.   Here, $S$ leads to a contradiction iff either some $j \in J$ appears in more than one pair in $S$, or for two distinct pairs $(j, t), (j', t') \in S$, $(t, t+p_j]$ and $(t', t'+p_{j'}]$ overlap. 
	
	We consider the constraints \eqref{LP-new:objective}-\eqref{LP-new:congestion} one by one, and prove that their respective induced constraints in the LP hierarchy are satisfied.  First consider \eqref{LP-new:scheduled}; we need to prove 
	\begin{align}
		\sum_{R' \subseteq R, t: (j,t) \in \calD} (-1)^{|R'|} x_{S \cup R' \cup \{(j, t)\}} &\leq \sum_{R'\subseteq R} (-1)^{|R'|}x_{S \cup R'}, \qquad \forall S, R \subseteq \calD \text{ with }|S| + |R| \leq r, j \in J, \label{inequ:original}
	\end{align}
	We can assume $S$ does not lead to a contradiction. If $j$ appears in $S$, then $t$ has to be the value satisfying $(j, t) \in S$ to make sure $x_{S \cup R' \cup \{(j, t)\}} \neq 0$. Then \eqref{inequ:original} holds with equality.  So, we assume $j$ does not appear in $S$. We can also assume that $R$ and $S$ are disjoint; otherwise, both sides of \eqref{inequ:original} are $0$. For the fixed $S, R$ and $j$, \eqref{inequ:original}  is equivalent to 
	\begin{align*}
		\sum_{R' \subseteq R} (-1)^{|R'|}\left(x_{S \cup R'} - \sum_{t: (j, t) \in D} x_{S \cup R'\cup\{(j, t)\}}\right) &\geq 0.  \label{inequ:to-prove}
	\end{align*}
	Consider any $R' \subseteq R$. If $S \cup R'$ leads to a contradiction, or $j$ appears in $R'$,  then it is easy to see that $x_{S \cup R'} - \sum_{t} x_{S \cup R' \cup \{(j, t)\}} =0$.  Otherwise, we have $x_{S \cup R'} - \sum_{t} x_{S \cup R' \cup \{(j, t)\}} \geq \epsilon' x_{S \cup R'}$. This holds since $x_{S \cup R'} = \left(\frac{1-\epsilon'}{L+1}\right)^{|S \cup R'|}$ and for each relevant $t$, we have $x_{S \cup R' \cup \{(j, t)\}}$ is either $0$ or $\left(\frac{1-\epsilon'}{L+1}\right)^{|S \cup R'|+1}$.
	
	Let $*$ be the family of subsets $R'$ of $R$ such that $S \cup R'$ does not lead to a contradiction and $j$ does not appear in $R'$. Then, in order to prove \eqref{inequ:original}, it suffices to show the following:
	\begin{align}
		\sum_{R' \in *, |R'|\text{ even}} \epsilon' x_{S \cup R'} - \sum_{R'' \in *, |R''|\text{ odd}} x_{S \cup R''} \geq 0.
	\end{align}
	Notice that we assumed that $S$ and $R$ are disjoint.  For every $R' \in *$ of even size, the $\epsilon' x_{S \cup R'}$ budget can be used to cover the negative sum $\sum_{R'' \in *: R' \subseteq R'', |R''| = |R'| + 1}x_{S \cup R''}$.   This holds since for every $R''$ in the summation, we have $x_{S \cup R''} = \frac{(1-\epsilon')x_{S\cup R'}}{L+1}$. Since there are at most $|R| \leq r \leq \epsilon' (L+1)$ terms in the summation,the budget $x_{S \cup R'}$ is at least the sum. Also, notice that each odd set $R'' \in *$ is covered at least once. So, we have proved \eqref{inequ:to-prove}, which implies \eqref{inequ:original}. 
	
	Now we turn to \eqref{LP-new:congestion}. We need to show
	\begin{align}
		\sum_{(j, t) \in \calD: t' \in (t, t+p_j]} \sum_{R' \subseteq R} (-1)^{|R'|} x_{S \cup R' \cup \{(j, t)\}} \leq \sum_{R'\subseteq R} (-1)^{|R'|}x_{S \cup R'}, \nonumber\\ 
		\forall S, R \subseteq \calD \text{ with } |S| + |R| \leq r, t' \in [T]. \label{inequ:original-1}
	\end{align}
	
	This can be proved similarly as \eqref{inequ:original}.  First, we can assume that $S$ does not lead to a contradiction and the intervals in $S$ do not cover $t'$.  Then, we can prove that $x_{S \cup R'} - \sum_{(j, t)\text{ covers }t'} x_{S \cup R' \cup \{(j, t)\}}$ is $0$ if $S \cup R'$ leads to a contradiction or intervals in $S \cup R'$ cover $t'$. Otherwise, the quantity is at least $\epsilon' x_{S \cup R'}$.  We can similarly define $*$ to be the family of subsets $R'$ of $R$ for which we have the latter case. Then using the same way we can prove \eqref{inequ:to-prove}, which implies \eqref{inequ:original-1}.
	
	Finally we consider \eqref{LP-new:objective}, the constraint for the objective value.  By reorganizing the terms, we need to prove that  for every $S, R \subseteq \calD$ with $|S| + |R| \leq r$, 
	\begin{align}
		 \sum_{R' \subseteq R} (-1)^{|R'|}\left(\sum_{(j, t) \in \calD} x_{S \cup R'\cup\{(j, t)\}}\cdot p_j - (1-\epsilon)T\cdot x_{S \cup R'} \right) \geq 0 \label{l2}
	\end{align}
	
	Let us focus on some $R' \subseteq R$ such that $S\cup R'$ does not lead to a contradiction. Define $\calD'$ be the set of $(j, t)$ pairs in $\calD \setminus (S \cup R')$ such that $S \cup R' \cup {(j, t)}$ does not lead to a contradiction. Then 
	\begin{align*}
		Q:=\sum_{(j, t) \in \calD} x_{S \cup R'\cup\{(j, t)\}}\cdot p_j = x_{S \cup R'}\left(\sum_{(j, t) \in S \cup R'}p_j + \frac{1-\epsilon'}{L+1} \sum_{(j, t) \in \calD'}p_j \right).
	\end{align*}
	We are interested in upper and lower bounds of $Q$.  For the upper bound , we have
	\begin{align*}
		Q \leq x_{S \cup R'}\left(q\cdot \frac{T}{L+1} + \frac{1-\epsilon'}{L+1}\cdot T(L+1) \right) \leq x_{S \cup R'} \left(\epsilon' T+(1-\epsilon')T\right) = x_{S \cup R'}T.
	\end{align*}
	Above, we used that $|S \cup R'| \leq q \leq \epsilon' (L+1)$, every $j \in J$ has $p_j \leq \frac{T}{L+1}$, $\sum_{j \in J}p_j = T$ and every $j$ appears in $\calD$ exactly $L+1$ times. 
	
	For the lower bound of $Q$,  we lower bound the term $\sum_{(j, t) \in \calD'}p_j$.  This is at least 
	\begin{align*}
		(L+1)T - \sum_{(j, t) \in \calD: j\text{ is in some pair in }S \cup R'}p_j - \sum_{(j, t) \in \calD \text{ intersects some interval in } S\cup R'}p_j.
	\end{align*} 

We upper bound the subtrahends one by one. The first subtrahend is at most $q (L+1)\cdot \frac{T}{L+1} \leq \epsilon' (L+1)T$.  The second term is maximized when $S \cup R'$ contains $q$ disjoint intervals of length $T/(L+1)$, with boundaries being multiply of $T/(L+1)$. (Recall that we assumed $L+1$ is a power of $2$; this does not correspond to an actual $S \cup R'$ since we only have 1 job of length $T/(L+1)$; but it will give an upper bound.)  In this case, for any job length $2^\ell$, we have $\Pr_{(j, t) \sim \calD}\left[(j, t)\text{ intersects a interval in } S \cup R' |p_j = 2^\ell\right]  = q/(L+1) \leq \epsilon' $.   So, the second subtrahend is at most $\epsilon' (L+1)T$.  Overall, we have $\sum_{(j, t) \in \calD'}p_j \geq (L+1)T - 2\epsilon'(L+1)T = (1-2\epsilon')(L+1)T$. This implies $Q \geq x_{S \cup R'}\cdot\frac{1-\epsilon'}{L+1}(1-2\epsilon')(L+1)T \geq (1-3\epsilon')x_{S \cup R'}T$.
		
	With the upper and lower bounds, we can prove \eqref{l2}.  Let $*$ be the family of subsets $R'$ of $R$ such that $S \cup R'$ does not lead to a contradiction. Then, the left side of \eqref{l2} is at least
	\begin{align*}
		&\quad \sum_{R' \in *: |R'| \text{ even}} \big((1-3\epsilon')x_{S \cup R'}T-(1-\epsilon)x_{S \cup R'}T\big) - \sum_{R' \in *: |R'| \text{ odd}} \big(x_{S \cup R'}T - (1-\epsilon)x_{S \cup R'}T\big) \\
		&=\epsilon'T\sum_{R \in *:|R'| \text{even}}  x_{S \cup R'} - \epsilon T\sum_{R \in *:|R'| \text{odd}}  x_{S \cup R'},
	\end{align*}
	where we used that $\epsilon' = \epsilon/4$.  Using the same covering idea as before and that $q \leq (L+1)/4$, we can prove the quantity is at least 0. This implies that \eqref{l2} holds, which finishes the proof of Theorem~\ref{thm:negative-main}.
	
	\subsection{Discussion}
	We showed that an $o(\log n)$-level Sherali-Adams lift of a natural LP relaxation for $1|r_j,d_j|\sum_j p_j U'_j$ can not distinguish between whether an instance $\calI$ has cost at most $\epsilon T$ or at least $0.2T$.  The instance $\calI'$ of $P2|\tprec|C_{\max}$ constructed from $\calI$ will have makespan at most $(1+\epsilon)T$ and at least $1.1T$ for the two cases. 
	
	One could ask if the natural LP relaxation for $P2|\tprec|C_{\max}$ on the instance $\calI'$ has large intergrality gap when raised to $o(\log n)$ levels. Unfortunately we could not prove such a result.  It is known that some small modifications with no effect on the basic LP can change the feasibles solutions in the hierarchy.  For example, we can introduce new variables $x_{j, \bot}$ in LP (\ref{LP:objective}-\ref{LP:positive}) to indicate whether $j$ is not scheduled, and require $x_{j, \bot} = 1 -\sum_{t, p'} x_{j, t, p'}$.  This does not change the LP, but we need to consider the new variables when deriving the constraints in the lifted LP.  We do not know how to give a fractional solution for this new lifted LP.  This seems to be a barrier to extend the negative result to $P2|\tprec|C_{\max}$, as in the problem, we do require every job to be processed to an extent of 1. Nevertheless, we believe  the instance we constructed is the right one for proving an integrality gap result for $P2|\tprec|C_{\max}$.
	
	We also remark that in the gap instance for $\minUo$, the windows of all jobs form a laminar tree of depth $O(\log n)$. Such an instance can be solved efficiently and exactly using dynamic programming.  Thus for this problem, the Sherali-Adams hierarchy does not capture the dynamic programming idea using a small number of rounds.  Moreover, using the dynamic programming technique in \cite{ILM15}, one may obtain a QPTAS for $\minUo$. This suggests that our reduction from $\minUo$ to $P2|\text{prec}|C_{\max}$ is unlikely to give APX-hardness for the latter problem, if such a result indeed holds.

\section*{Acknowledgments}
We thank Shashwat Garg for several helpful discussions on the topic. 
\fi
\bibliographystyle{plain}
\bibliography{reflist} 
 
\ifsoda\else
	\appendix
	\section{Comparison of Three Optimal Schedules}
\label{sec:3schedules}
\begin{figure}[h]
\centering
\includegraphics[width=1.0\textwidth]{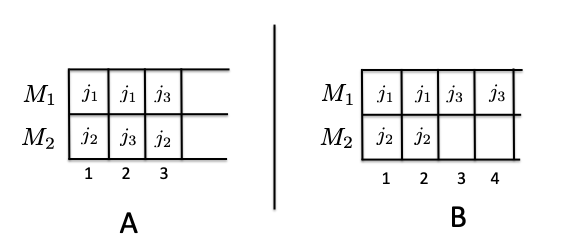}
 \caption{The figure illustrates a gap of at least 4/3 between the optimal schedules for A and B.} \label{fig:ab}
\end{figure}

Recall that when jobs have arbitrary processing lengths, the optimal schedule for minimizing makespan with precedence constraints can be of three types:

\medskip

A) Fully preemptive ($Pm|\text{prec}, \text{pmtn}, \text{migration}|C_{\max}$); 

B) Preemptive but non-migratory ($Pm|\text{prec}, \text{pmtn}|C_{\max}$); 

C) Non-preemptive ($Pm|\text{prec}|C_{\max}$).

\medskip
Clearly, the value of optimal makespan for A is at most B which is at most C.
Here we give instances to prove that these three schedules can be constant factor away from each other, and hence ruling out a black-box approach to the design of $(1+\epsilon)$-approximation algorithms for these problems.

\subsection{Gap Between A and B}
Refer to Figure \ref{fig:ab}. In the instance we have 3 jobs each of length 2, and no precedence constraints. 
The optimal schedules for A and B are shown in the figure. There is a gap of $4/3$ between A and B. We can generalize the instance to have $m$ machines and $m+1$ jobs of length $m$. Then model A can achieve makespan $m+1$, whereas model B needs to have makespan $2m$.  Thus the gap can be made arbitrarily close to $2$. This also gives a gap close to $2$ between A and C, since C is more restricted than B.

\begin{figure}[h]
\centering
\includegraphics[width=1.0\textwidth]{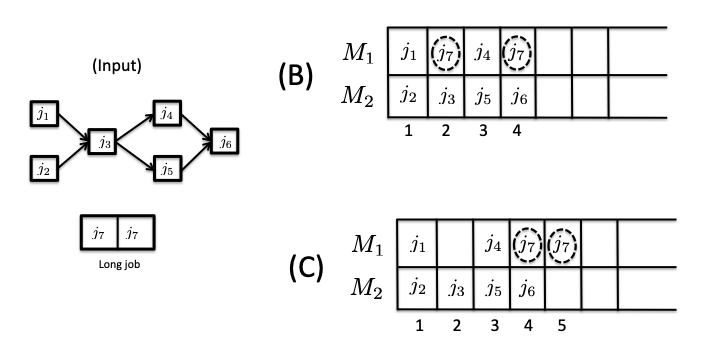}
 \caption{The figure illustrates a gap of at least 5/4 between the optimal schedules B and C.} \label{fig:bc}
\end{figure}

\subsection{Gap Between B and C}
Refer to Figure \ref{fig:bc}. 
In the instance we have 7 jobs.
The first six jobs have unit processing lengths, and the precedence relationship is as shown in the figure.
The job 7 is long and has processing length of 2.
The optimal schedules for B and C are shown in the figure. 
By making the DAG contain $3n$ unit length jobs, and the processing length of long job as $n$, we can make the gap approach 1.5.

	\section{Removing the Polynomial Size Assumption on Job Lengths}
\label{sec:N-big}
We give a brief sketch of how to handle the case when $p_j$'s are not polynomially bounded. Let $p_{\max}=\max_j p_j$.  We can round each job size down to the nearest multiple of $\epsilon p_{\max}/n$; in particular, if there is a job $j$ such that $p_j< \epsilon p_{max} /n$, we discard it. Thus, the total size  of jobs we discarded is at most $\epsilon p_{max}$. It is easily seen that the optimum value must be at least $p_{max}$ and hence the total size we discarded is at most $\epsilon$ times the optimum makespan. 
Scaling down all job sizes by a factor of $\epsilon p_{\max}/n$, we obtain an instance where all job sizes are integers between 1 and $n/\epsilon$. 
Hence, we can assume that $\sum_{j}p_j$ is polynomial in $n$.

\fi

\end{document}